\documentclass{lmcs}
\pdfoutput=1
\usepackage[utf8]{inputenc}

\usepackage{lastpage}
\lmcsdoi{22}{1}{4}
\lmcsheading{}{\pageref{LastPage}}{}{}%
{Dec.~19,~2024}{Jan.~09,~2026}{}

\usepackage{xcolor}
\usepackage{amsmath}
\usepackage{amssymb}
\usepackage{tikz}
\usepackage{caption}
\usepackage{subcaption}
\usepackage{hyperref}

\newcommand{\query}[1]{query(#1)}
\newcommand{\powset}[1]{\mathcal{P}(#1)}
\newcommand{\powsetnop}{\mathcal{P}}

\newcommand{\allmatches}[2]{\mathcal{M}^{#1}_{#2}}

\newcommand{\results}[4]{\mathcal{R}(#1, #2, #3, #4)}
\newcommand{\resultsnop}{\mathcal{R}}
\newcommand{\resultslocal}[5]{\mathcal{R}^\Phi(#1, #2, #3, #4, #5)}

\newcommand{\resultsstripped}[2]{\mathcal{R}^\Phi_X(#1, #2)}

\newcounter{thmappcount}
\newtheorem{thmapp}{Theorem}[thmappcount]

\usetikzlibrary {decorations.pathmorphing}
\usetikzlibrary{positioning}
\usetikzlibrary{arrows.meta}
\tikzset{retenode/.style = {outer sep=0cm, rectangle, text centered, text height = 0.25cm, draw=black, minimum width = 1.7cm, minimum height = 0.75cm}}
\tikzset{ms_retenode/.style = {outer sep=0cm, rectangle, text centered, text height = 0.25cm, draw=black, minimum width = 1.7cm, minimum height = 0.75cm, label={[xshift=0.75cm, yshift=-0.35cm]\scriptsize$\Phi$}}}
\tikzset{reteconfig/.style = {outer sep=0cm, rectangle, dashed, align=left, font=\tiny, text height = 0.25cm, draw=black, minimum width = 1.7cm, minimum height = 0.75cm}}

\tikzset{retenode_app/.style = {outer sep=0cm, rectangle, text centered, text height = 0.25cm, draw=black, minimum width = 1.7cm, minimum height = 0.75cm}}
\tikzset{ms_retenode_app/.style = {outer sep=0cm, rectangle, text centered, text height = 0.25cm, draw=black, minimum width = 1.7cm, minimum height = 0.75cm, label={[xshift=0.75cm, yshift=-0.35cm]\scriptsize$\Phi$}}}

\tikzset{graph/.style = {circle, text centered, minimum size = 1cm}}

\tikzset{vertex/.style = {circle, text centered, text height = 0.25cm, draw=black, minimum width = 1.75cm, minimum height = 1.5cm, outer sep = 0pt}}

\definecolor{darkgreen}{RGB}{0, 200, 0}

\begin{document}

\title[Localized RETE for Incremental Graph Queries with NGCs]{Localized RETE for Incremental Graph Queries with Nested Graph Conditions}

\author{Matthias Barkowsky\lmcsorcid{0000-0002-1138-2425}}
\author{Holger Giese\lmcsorcid{0000-0002-4723-730X}}

\address{Hasso Plattner Institute at the University of Potsdam,
Prof.-Dr.-Helmert-Str. 2-3, D-14482 Potsdam, Germany}
\email{\{matthias.barkowsky,holger.giese\}@hpi.de}

\begin{abstract}
The growing size of graph-based modeling artifacts in model-driven engineering calls for techniques that enable efficient execution of graph queries. Incremental approaches based on the RETE algorithm provide an adequate solution in many scenarios, but are generally designed to search for query results over the entire graph. However, in certain situations, a user may only be interested in query results for a subgraph, for instance when a developer is working on a large model of which only a part is loaded into their workspace. In this case, the global execution semantics can result in significant computational overhead.

To mitigate the outlined shortcoming, in this article we propose an extension of the RETE approach that enables local, yet fully incremental execution of graph queries, while still guaranteeing completeness of results with respect to the relevant subgraph.

We empirically evaluate the presented approach via experiments inspired by a scenario from software development and with queries and data from an independent social network benchmark. The experimental results indicate that the proposed technique can significantly improve performance regarding memory consumption and execution time in favorable cases, but may incur a noticeable overhead in unfavorable cases.
\end{abstract}

\maketitle


\section{Introduction} \label{sec:introduction}

In model-driven engineering, models constitute important development artifacts \cite{kent2002model}. With complex development projects involving large, interconnected models, performance of automated model operations becomes a primary concern.

Incremental graph query execution based on the RETE algorithm \cite{forgy1989rete} has been demonstrated to be an adequate solution in scenarios where an evolving model is repeatedly queried for the same information \cite{szarnyas2018train}. In this context, the RETE algorithm essentially tackles the problem of querying a usually graph-based model representation via operators from relational algebra \cite{codd1970relational}. Thereby, the RETE algorithm supports not only graph queries in the form of plain graph patterns, but also covers advanced formalisms for query specification such as nested graph conditions \cite{habel2009correctness}. On the flip side, the relational operators are not designed to exploit the locality typically found in graph-based encodings. Current incremental querying techniques consequently require processing of the entire model to guarantee complete results. However, in certain situations, global query execution is not required and may be undesirable due to performance considerations.

For instance, a developer may be working on only a part of a model loaded into their workspace, with a large portion of the full model still stored on disk. A concrete example of this would be a developer using a graphical editor to modify an individual block diagram from a collection of interconnected block diagrams, which together effectively form one large model. As the developer modifies the model part in their workspace, they may want to continuously monitor how their changes impact the satisfaction of some consistency constraints via incremental model queries. In this scenario, existing techniques for incremental query execution require loading and querying the entire collection of models, even though the user is ultimately only interested in the often local effect of their own changes.

The global RETE execution semantics then results in at least three problems: First, the computation of query results for the full model, of which only a portion is relevant to the user, may incur a substantial overhead on initial query execution time. Second, incremental querying techniques are known to create and store a large number of intermediate results \cite{szarnyas2018train}, many of which can in this scenario be superfluous, causing an overhead in memory consumption. Third, query execution requires loading the entire model into memory, potentially causing an overhead in loading time and increasing the overhead in memory consumption.

These problems can be mitigated to some extent by employing local search \cite{cordella2004sub,giese2009improved,arendt2010henshin}, which better exploits the locality of the problem and lazy loading capabilities of model persistence layers \cite{cdo,daniel2017neoemf}, instead of a RETE-based technique. However, resorting to local search can result in expensive redundant search operations that are only avoided by fully incremental solutions \cite{szarnyas2018train}.

In this article, we instead propose to tackle the outlined shortcoming of incremental querying techniques via an extension of the RETE approach. This is supported by a relaxed notion of completeness for query results that accounts for situations where the computation for the full model is unnecessary. Essentially, this enables a distinction between the full model, for which query results do not necessarily have to be complete, and a relevant model part, for which complete results are required. The extended RETE approach then anchors query execution to the relevant model part and lazily fetches additional model elements required to compute query results that meet the relaxed completeness requirement. Our approach thereby avoids potentially expensive global query execution and allows an effective integration of incremental queries with model persistence layers.

The remainder of this article is structured as follows:
Section~\ref{sec:preliminaries} provides a summary of our notion of graphs, graph queries, and the RETE mechanism for query execution.
Section~\ref{sec:incremental_queries_over_persistent_models} introduces a relaxed notion of completeness for results of model queries consisting of a plain graph pattern and subsequently presents an adaptation to the RETE querying mechanism that enables local, yet fully incremental execution of such plain model queries.
Our technique for local query execution via RETE is then extended to the case of more complex model queries featuring nested graph conditions in Section~\ref{sec:incremental_extended_queries_over_subgraphs}.
A prototypical implementation of our approach is evaluated regarding execution time and memory consumption in Section~\ref{sec:evaluation}.
Section~\ref{sec:related_work} discusses related work, before Section~\ref{sec:conclusion} concludes the article and provides an outlook on future work.

This article is an extension of our paper at ICGT 2024 \cite{barkowsky2024localized}. The main increment of this article compared to the conference paper version is the approach for localizing execution of graph queries with nested graph conditions presented in Section~\ref{sec:incremental_extended_queries_over_subgraphs}, which is also covered by an extended evaluation in Section~\ref{sec:evaluation}. Besides small improvements to text and figures and the inclusion of further examples, we have also added a theorem for recognizing locally relevant query results as another minor extension in Section~\ref{sec:incremental_queries_over_persistent_models}.


\section{Preliminaries} \label{sec:preliminaries}

In the following, we briefly reiterate the definitions of graphs and graph queries and summarize the RETE approach to incremental graph querying.

\subsection{Graphs and Graph Queries} \label{sec:preliminaries_graphs}

As defined in \cite{Ehrig+2006}, a \emph{graph} is a tuple $G = (V^G, E^G, s^G, t^G)$, with $V^G$ the set of vertices, $E^G$ the set of edges, and $s^G : E^G \rightarrow V^G$ and $t^G : E^G \rightarrow V^G$ functions mapping edges to their source respectively target vertices.

A mapping from a graph $Q$ into another graph $H$ is defined via a \emph{graph morphism} $m : Q \rightarrow H$, which is characterized by two functions $m^V : V^Q \rightarrow V^H$ and $m^E : E^Q \rightarrow E^H$ such that $s^H \circ m^E = m^V \circ s^Q$ and $t^H \circ m^E = m^V \circ t^Q$. A graph morphism is called a monomorphism if $m^{V}$ and $m^ {E}$ are injective and an isomorphism if $m^{V}$ and $m^{E}$ are also surjective. A \emph{partial graph morphism} consists of two partial functions $m^V : V^Q \rightarrow V^H$ and $m^E : E^Q \rightarrow E^H$ with $s^H \circ m^E \leq m^V \circ s^Q$ and $t^H \circ m^E \leq m^V \circ t^Q$ \cite{ribeiro1999parallel}. In this context, for two partial functions $f$ and $g$, $f \leq g$ denotes inclusion of the domain of $f$ in the domain of $g$ and equality of mappings in $f$ and $g$ for all domain elements for which $f$ is defined.

A \emph{graph modification} from a graph $G$ into a modified graph $G'$ can be described via an intermediate graph $K$ and a span of graph monomorphisms $f : K \rightarrow G$ and $g : K \rightarrow G'$ \cite{taentzer2014fundamental}, which can also be denoted $G \xleftarrow{f} K \xrightarrow{g} G'$. Intuitively, elements in $G \setminus f(K)$ are deleted and elements in $G' \setminus g(K)$ are created by the modification.

A graph $G$ can be typed over a type graph $TG$ via a graph morphism $type^G : G \rightarrow TG$, forming the \emph{typed graph} $G^T =(G, type^G)$. A \emph{typed graph morphism} from $G^T = (G, type^G)$ into another typed graph $H^T = (H, type^H)$ with type graph $TG$ is a graph morphism $m^T : G \rightarrow H$ such that $type^H \circ m^T = type^G$.

A \emph{model} is then simply given by a typed graph. In the remainder of the article, we therefore use the terms graph and model interchangably and implicitly assume typing when talking about graphs.

In practice, graphs often contain more edges than nodes, a characteristic that may influence the performance of certain algorithms. To capture this property for each individual edge type in a typed graph, we introduce the notion of \emph{edge-dominated} graphs:

\begin{defi} \label{def:edge_dominated_graphs} \emph{(Edge-Dominated Typed Graphs)}
We say that a typed graph $G$ is \emph{edge-dominated} if $\forall e_{TG} \in E^{TG} : |\{e \in E^G \mid type(e) = e_{TG}\}| \geq max(|\{v \in V^G \mid type(v) = s^{TG}(e_{TG})\}|, |\{v \in V^G \mid type(v) = t^{TG}(e_{TG})\}|)$.
\end{defi}

The above definition will be relevant in the context of an analytical evaluation of our localized querying technique, which technically requires graphs to be edge-dominated in order to guarantee desirable characteristics regarding computational complexity.

A \emph{plain graph query} as considered in this article is characterized by a query graph $Q$ and can be executed over a host graph $H$ by finding graph morphisms $m : Q \rightarrow H$. We also call these morphisms \emph{matches} and denote the set of all matches from $Q$ into $H$ by $\allmatches{Q}{H}$. Typically, a set of matches is considered a complete query result if it contains all matches in $\allmatches{Q}{H}$:

\begin{defi} \label{def:completeness} \emph{(Completeness of Plain Query Results)}
	We say that a set of matches $M$ from query graph $Q$ into host graph $H$ is complete
	iff $\allmatches{Q}{H} \subseteq M$.
\end{defi}

Note that by virtue of $M$ being defined as a set of matches from $Q$ into $H$, the statement $\allmatches{Q}{H} \subseteq M$ in the above definition implies that $M = \allmatches{Q}{H}$.

In some cases, the formalism of plain graph queries may not be sufficiently expressive to specify the desired query, for instance if matches for the query should not be extensible by a specific negative pattern. To increase expressiveness, a plain graph query $Q$ can therefore be equipped with a \emph{nested graph condition} $\psi$ \cite{habel2009correctness}, yielding the \emph{extended graph query} $(Q, \psi)$.

Nested graph conditions can be defined recursively in the context of a graph morphism $m : Q \rightarrow H$, with $m \models \psi$ denoting that $m$ satisfies the condition $\psi$:
\begin{itemize}
\item $\texttt{true}$ is a nested graph condition that is always satisfied, that is, $m \models \texttt{true}$ for any graph morphism $m$.
\item $\neg\psi$ is a nested graph condition that is satisfied iff $m$ does not satisfy the nested graph condition $\psi$, that is, $m \models \neg\psi \Leftrightarrow m \not\models \psi$.
\item $\psi_1 \wedge \psi_2$ is a nested graph condition that is satisfied iff $m$ satisfies both nested graph conditions $\psi_1$ and $\psi_2$, that is, $m \models \psi_1 \wedge \psi_2 \Leftrightarrow (m \models \psi_1 \wedge m \models \psi_2)$.
\item $\exists(a: Q \rightarrow Q', \psi')$, where $a: Q \rightarrow Q'$ is a graph monomorphism, is a nested graph condition that is satisfied iff there exists a graph morphism $m': Q' \rightarrow H$ such that $m = m' \circ a$ and $m'$ satisfies the nested graph condition $\psi'$, that is, $m \models \exists(a: Q \rightarrow Q', \psi') \Leftrightarrow (\exists m' \in \allmatches{Q'}{H}: m' \models \psi' \wedge m = m' \circ a)$.
\end{itemize}

As in \cite{barkowsky2023host}, we also consider the case where the morphism $a$ of a nested graph condition $\exists(a: Q \rightarrow Q', \psi')$ can be a partial graph morphism from some subgraph $Q_p \subseteq Q$ of $Q$ into $Q'$. In this case, satisfaction of the nested graph condition is given by $m \models \exists(a: Q \rightarrow Q', \psi') \Leftrightarrow (\exists m' \in \allmatches{Q'}{H}: m' \models \psi' \wedge m|_{Q_p} = m' \circ a)$.

Equipping the query graph $Q$ with the nested graph condition $\psi$ restricts the set of valid matches for the extended graph query $(Q, \psi)$ to those matches for $Q$ that also satisfy $\psi$. In contrast to plain graph queries, completeness alone is thus insufficient as a criterion for describing what constitutes a desirable query result. Specifically, when talking about extended graph queries, a query result should not only contain all admissible matches, but also no matches that violate the equipped nested graph condition. In this article, we will hence talk about \emph{completeness} with respect to the results of plain graph queries and \emph{correctness} in the case of extended graph queries, with the latter defined as follows:

\begin{defi} \label{def:correctness_ngcs} \emph{(Correctness of Extended Query Results)}
	We say that a set of matches $M$ for the extended graph query $(Q, \psi)$ into host graph $H$ is correct
	iff $M = \{m \in \allmatches{Q}{H} | m \models \psi\}$.
\end{defi}

\subsection{Incremental Graph Queries with RETE} \label{sec:preliminaries_rete}

The RETE algorithm \cite{forgy1989rete} forms the basis of mature incremental graph querying techniques \cite{varro2016}. Therefore, the query is recursively decomposed into simpler subqueries, which are arranged in a second graph called RETE net. In the following, we will refer to the vertices of RETE nets as \emph{(RETE) nodes}.

For plain graph queries, each RETE node $n$ computes matches for a subgraph $\query{n} \subseteq Q$ of the full query graph $Q$. This computation may depend on matches computed by other RETE nodes. Such dependencies are represented by edges from the dependent node to the dependency node. For each RETE net, one of its nodes is designated as the \emph{production node}, which computes the net's overall result. A RETE net is thus given by a tuple $(N, p)$, where the \emph{RETE graph} $N$ is a graph of RETE nodes and dependency edges and $p \in V^N$ is the production node.

We describe the \emph{configuration} of a RETE net $(N, p)$ during execution by a function $\mathcal{C} : V^N \rightarrow \powset{\mathcal{M}_\Omega}$, with $\powsetnop$ denoting the power set and $\mathcal{M}_\Omega$ the set of all graph morphisms. $\mathcal{C}$ then assigns each node in $V^N$ a \emph{current result set}.

For a starting configuration $\mathcal{C}$ and host graph $H$, executing a RETE node $n \in V^N$ yields an updated configuration $\mathcal{C}' = execute(n, N, H, \mathcal{C})$, with

\[
\mathcal{C}'(n') =
	\begin{cases}
	\results{n}{N}{H}{\mathcal{C}}	& \quad \text{if } n' = n\\
	\mathcal{C}(n')			& \quad \text{otherwise},
	\end{cases}
\]

where $\resultsnop$ is a function defining the \emph{target result set} of a RETE node $n$  in the RETE graph $N$ for $H$ and $\mathcal{C}$ such that $\results{n}{N}{H}{\mathcal{C}} \subseteq \allmatches{Q}{H}$, with $Q = \query{n}$.

We say that $\mathcal{C}$ is \emph{consistent} for a RETE node $n \in V^N$ and host graph $H$ iff $\mathcal{C}(n) = \results{n}{N}{H}{\mathcal{C}}$. We furthermore define a consistent configuration as a configuration that is consistent for all nodes of the RETE net:

\begin{defi}[Consistent RETE Net Configurations] \label{def:consistent_configuration} 
We say that a configuration $\mathcal{C} : V^N \rightarrow \powset{\mathcal{M}_\Omega}$ for a RETE net $(N, p)$ is consistent for a host graph $H$ iff $\forall n \in V^N : \mathcal{C}(n) = \results{n}{N}{H}{\mathcal{C}}$.
\end{defi}
 
If $H$ is clear from the context, we simply say that $\mathcal{C}$ is consistent.

Given a host graph $H$ and a starting configuration $\mathcal{C}_0$, a RETE net $(N, p)$ is executed via the execution of a sequence of nodes $O = n_1, n_2, ..., n_x$ with $n_i \in V^N$. This yields the trace $\mathcal{C}_0, \mathcal{C}_1, \mathcal{C}_2, ..., \mathcal{C}_x$, with $\mathcal{C}_y = execute(n_y, N, H, \mathcal{C}_{y-1})$, and the result configuration $execute(O, N, H, \mathcal{C}_0) = \mathcal{C}_x$.

$(N, p)$ can initially be executed over $H$ via a sequence $O$ that yields a consistent configuration $\mathcal{C}_x = execute(O, N, H, \mathcal{C}_0)$, where $\mathcal{C}_0$ is an empty starting configuration with $\forall n \in V^N : \mathcal{C}_0(n) = \emptyset$. Incremental execution of $(N, p)$ can be achieved by retaining a previous result configuration and using it as the starting configuration for execution over an updated host graph. This requires incremental implementations of RETE node execution procedures that update previously computed result sets for changed inputs instead of computing them from scratch.

In the most basic form, a RETE net consists of two kinds of nodes, \emph{edge input nodes} and \emph{join nodes}. An edge input node $[v \rightarrow w]$ has no dependencies, is associated with the query subgraph $\query{[v \rightarrow w]} = (\{v, w\}, \{e\}, \{(e, v)\}, \{(e, w)\})$, and directly extracts the corresponding (trivial) matches from a host graph. A join node $[\bowtie]$ has two dependencies $n_l$ and $n_r$ with $Q_l = \query{n_l}$ and $Q_r = \query{n_r}$ such that $V^{Q_l \cap Q_r} \neq \emptyset$ and is associated with $\query{[\bowtie]} = Q_l \cup Q_r$. $[\bowtie]$ then computes matches for this union subgraph by combining the matches from its dependencies along the overlap graph $Q_\cap = Q_l \cap Q_r$.

In the following, we assume query graphs to be weakly connected and contain at least one edge. A RETE net that computes matches for a query graph $Q$ can then be constructed as a tree of join nodes over edge input nodes. The join nodes thus gradually compose the trivial matches at the bottom into matches for more complex query subgraphs. We call such tree-like RETE nets \emph{well-formed}. An execution sequence that always produces a consistent configuration is given by a reverse topological sorting of the net. The root node of the tree then computes the set of all matches for $Q$ and is designated as the net's production node.

Connected graphs without edges consist of only a single vertex, making query execution via \emph{vertex input nodes}, which function analogously to edge input nodes, trivial. Disconnected query graphs can be handled via separate RETE nets for all query graph components and the computation of a cartesian product.

Figure~\ref{fig:example_query_and_rete} shows an example plain graph query from the software domain and an associated RETE net. The query searches for paths of a package, class, and field. The RETE net constructs matches for the query by combining edges from a package to a class with edges from a class to a field via a natural join.

\begin{figure}
	\centering
	\begin{subfigure}{0.5\textwidth}
		\centering

\begin{tikzpicture}

\node (p) [vertex] {\textit{p:Pkg}};

\node[right = 0.5cm of p] (c) [vertex] {\textit{c:Class}};

\node[right = 0.5cm of c] (f) [vertex] {\textit{f:Field}};

\draw [-{Latex}] (p) -- (c) node [midway, above] {\textit{ce}};

\draw [-{Latex}] (c) -- (f) node [midway, above] {\textit{fe}};

\end{tikzpicture}
	\end{subfigure}\quad\vline\quad
	\begin{subfigure}{0.4\textwidth}
		\centering

\begin{tikzpicture}

\node (join) [retenode] {$\bowtie$};

\node[below left = 0.35cm and -0.5cm of join] (ce) [retenode] {$p \rightarrow c$};

\node[below right = 0.35cm and -0.5cm of join] (fe) [retenode] {$c \rightarrow f$};

\draw [-{Latex}] (join) -- (ce);

\draw [-{Latex}] (join) -- (fe);

\end{tikzpicture}
	\end{subfigure}
	\caption{Example plain graph query (left) and corresponding RETE net (right)} \label{fig:example_query_and_rete}
\end{figure}

The RETE approach to graph query execution also supports extended graph queries via two additional RETE node types: \emph{semi-join nodes} and \emph{anti-join nodes}.  Like join nodes, semi-join nodes and anti-join nodes have two dependencies $n_l$ and $n_r$ with $Q_l = \query{n_l}$ and $Q_r = \query{n_r}$ such that $V^{Q_l \cap Q_r} \neq \emptyset$. However, instead of combining matches from both dependencies, a semi-join $[\ltimes]$ simply computes all matches from the left dependency that have a compatible match in the right dependency. An anti-join $[\rhd]$ then computes all matches from the left dependency that do \emph{not} have a compatible match in the right dependency. Both semi-join and anti-join nodes are thus always associated with the same query graph as their left dependency, that is, $\query{[\ltimes]} = Q_l$ and $\query{[\rhd]} = Q_l$.

To support nested graph conditions of the form $\exists (a : Q \rightarrow Q', \psi)$, a semi-join can also be defined along a non-empty, potentially partial graph morphism \cite{barkowsky2023host}. The definition based on an overlap graph is then just a special case of a semi-join along a partial identity graph morphism.

Using semi-join and anti-join nodes and a RETE net $(N_Q, p_Q)$ for the plain graph query $Q$, a RETE net $(N_{(Q, \psi)}, p_{(Q, \psi)})$ for an extended graph query $(Q, \psi)$ can be constructed recursively according to our description in \cite{barkowsky2023host}, which is visualized in Figure~\ref{fig:extended_rete}.

\begin{figure} [t]
	\centering
	\begin{subfigure}[b]{0.36\textwidth}
		\begin{center}

\begin{tikzpicture}

\node[rectangle,
	minimum width = 4.15cm, 
	minimum height = 2.2cm] (r) at (0,0) {};

\node [below = -2.2cm of r, dashed] [retenode] {$N_Q$};

\end{tikzpicture}
		\end{center}
	\end{subfigure}\hspace{4pt}\vline\hspace{4pt}
	\begin{subfigure}[b]{0.55\textwidth}
		\begin{center}

\begin{tikzpicture}

\node[rectangle,
	minimum width = 6.85cm, 
	minimum height = 2.2cm] (r) at (0,0) {};

\node [below = -2.2cm of r] (join) [retenode] {$\ltimes$};

\node[below left = 0.35cm and -0.5cm of join, dashed] (q) [retenode] {$N_Q$};

\node[below right = 0.35cm and -0.5cm of join, dashed] (qp_psip) [retenode] {$N_{(Q', \psi')}$};

\draw [-{Latex}] (join) -- (q);

\draw [-{Latex}] (join) -- (qp_psip);

\end{tikzpicture}
		\end{center}
	\end{subfigure}

	\hrule

	\begin{subfigure}[b]{0.36\textwidth}
		\begin{center}

\begin{tikzpicture}

\node[rectangle,
	minimum width = 4.15cm, 
	minimum height = 3.5cm] (r) at (0,0) {};

\node [below = -3.25cm of r] (join) [retenode] {$\rhd$};

\node[below left = 0.35cm and -0.5cm of join, dashed] (q) [retenode] {$N_Q$};

\node[below right = 0.35cm and -0.5cm of join, dashed] (q_psip) [retenode] {$N_{(Q, \psi')}$};

\draw [-{Latex}] (join) -- (q);

\draw [-{Latex}] (join) -- (q_psip);

\end{tikzpicture}
		\end{center}
	\end{subfigure}\hspace{4pt}\vline\hspace{4pt}
	\begin{subfigure}[b]{0.55\textwidth}
		\begin{center}

\begin{tikzpicture}

\node[rectangle,
	minimum width = 6.85cm, 
	minimum height = 3.5cm] (r) at (0,0) {};

\node [below left = -3.25cm and -4.9cm of r] (join2) [retenode] {$\ltimes$};

\node [below left = 0.35cm and 0.25cm of join2] (join1) [retenode] {$\ltimes$};

\node[below right = 0.35cm and 0.25cm of join2, dashed] (q_psi2) [retenode] {$N_{(Q, \psi_2)}$};

\node[below left = 0.35cm and -0.5cm of join1, dashed] (q) [retenode] {$N_Q$};

\node[below right = 0.35cm and -0.5cm of join1, dashed] (q_psi1) [retenode] {$N_{(Q, \psi_1)}$};

\draw [-{Latex}] (join2) -- (join1);

\draw [-{Latex}] (join2) -- (q_psi2);

\draw [-{Latex}] (join1) -- (q);

\draw [-{Latex}] (join1) -- (q_psi1);

\end{tikzpicture}
		\end{center}
	\end{subfigure}
	\caption{RETE net $(N_{(Q, \psi)}, p_{(Q, \psi)})$ for an extended graph query $(Q, \psi)$ with nested graph condition of the form $\psi = \texttt{true}$ (top left), $\psi = \exists(a : Q \rightarrow Q', \psi')$ (top right), $\psi = \neg\psi'$ (bottom left), and $\psi = \psi_1 \wedge \psi_2$ (bottom right)} \label{fig:extended_rete}
\end{figure}

Figure~\ref{fig:example_query_and_rete_ngc} shows an example extended graph query on the left. The extended query searches for the same pattern as the plain query in Figure~\ref{fig:example_query_and_rete}, but also requires the class node to be connected to an interface node, without this additional node being part of the main match. Therefore, the query is defined as (Q, $\exists (a : Q \rightarrow Q', \texttt{true}))$ with a partial graph morphism $a$ that maps only the class node in $Q$ to the class node in $Q'$.

The corresponding RETE net is visualized on the right of Figure~\ref{fig:example_query_and_rete_ngc}. For computing matches for the base pattern, the RETE net replicates the structure in Figure~\ref{fig:example_query_and_rete}. The nested graph condition is implemented via a search for the pattern $Q'$, which consists of a single edge from the node $c$ to $i$, and a semi-join node at the top of the net, which checks the matches for $Q$ against the matches for $Q'$.

\begin{figure}
	\centering
	\begin{subfigure}[b]{0.48\textwidth}
		\begin{center}

\begin{tikzpicture}

\node (p) [vertex] {\textit{p:Pkg}};

\node[right = 0.5cm of p] (c) [vertex] {\textit{c:Class}};

\node[right = 0.5cm of c] (f) [vertex] {\textit{f:Field}};

\node[below = 0.5cm of c] (i) [vertex] {\textit{i:Intf}};

\draw [-{Latex}] (p) -- (c) node [midway, above] {\textit{ce}};

\draw [-{Latex}] (c) -- (f) node [midway, above] {\textit{fe}};

\draw [-{Latex}] (c) -- (i) node [midway, right] {\textit{ie}};

\node[above = -4.25cm of c, dashed, rectangle, draw = black, minimum width = 2.25cm, minimum height = 4.75cm, text depth = 4.2cm, text width = 1.75cm] (q_prime) {\textit{Q$'$}};

\node[above = -2cm of c, dashed, rectangle, draw = black, minimum width = 6.75cm, minimum height = 2.75cm, text depth = 2.25cm, text width = 6.25cm] (q) {\textit{Q}};

\node[rectangle,
	minimum width = 6.85cm, 
	minimum height = 5.5cm, below = -3cm of q] (r) {};
	
\end{tikzpicture}
		\end{center}
	\end{subfigure}\hspace{4pt}\vline\hspace{4pt}
	\begin{subfigure}[b]{0.48\textwidth}
		\begin{center}

\begin{tikzpicture}

\node[rectangle,
	minimum width = 6.85cm, 
	minimum height = 5.5cm] (r) at (0,0) {};
	
\node[below left = -4.5cm and -4.9cm of r] (semijoin) [retenode] {$\ltimes$};

\node [below left = 0.35cm and 0.25cm of semijoin] (join) [retenode] {$\bowtie$};

\node[below right = 0.35cm and 0.25cm of semijoin] (ie) [retenode] {$c \rightarrow i$};

\node[below left = 0.35cm and -0.5cm of join] (ce) [retenode] {$p \rightarrow c$};

\node[below right = 0.35cm and -0.5cm of join] (fe) [retenode] {$c \rightarrow f$};

\draw [-{Latex}] (semijoin) -- (join);

\draw [-{Latex}] (semijoin) -- (ie);

\draw [-{Latex}] (join) -- (ce);

\draw [-{Latex}] (join) -- (fe);

\end{tikzpicture}
		\end{center}
	\end{subfigure}
	\caption{Example extended graph query $(Q, \exists (a : Q \rightarrow Q', \texttt{true}))$ (left) and corresponding RETE net (right)} \label{fig:example_query_and_rete_ngc}
\end{figure}


\section{Incremental Plain Queries over Subgraphs} \label{sec:incremental_queries_over_persistent_models}

As outlined in Section~\ref{sec:introduction}, users of model querying mechanisms may only be interested in query results related to some part of a model that is relevant to them rather than the complete model. However, simply executing a query only over this relevant subgraph and ignoring context elements is often insufficient, for instance if the full effect of modifications to the relevant subgraph should be observed, since such modifications may affect matches that involve elements outside the relevant subgraph. Instead, completeness in such scenarios essentially requires the computation of all matches that somehow \emph{touch} the relevant subgraph.

In order to capture this need for local completeness, but avoid the requirement for global execution inherent to the characterization in Definition~\ref{def:completeness}, we define completeness under a relevant subgraph of some host graph as follows:\footnote{Note that by the definition of $H_p$ as a subgraph of $H$ and the definition of matches as graph morphisms, $m(Q) \cap H_p \neq \emptyset \Leftrightarrow m(V^Q) \cap H_p \neq \emptyset$. Our slightly adapted definition in this article thus coincides with the definition in the conference paper version \cite{barkowsky2024localized}.}

\begin{defi}[Completeness of Plain Query Results under Subgraphs] \label{def:completeness_subgraphs} 
	We say that a set of matches $M$ from query graph $Q$ into host graph $H$ is complete under a subgraph $H_p \subseteq H$
	iff $\{m \in \allmatches{Q}{H} \mid m(Q) \cap H_p \neq \emptyset\} \subseteq M$. 
\end{defi}

\subsection{Marking-sensitive RETE}

Due to its reliance on edge and vertex input nodes and their global execution semantics, incremental query execution via the standard RETE approach is unable to exploit the relaxed notion of completeness from Definition~\ref{def:completeness_subgraphs} for query optimization and does not integrate well with mechanisms relying on operation locality, such as model persistence layers based on lazy loading \cite{cdo,daniel2017neoemf}.

While query execution could be localized to a relevant subgraph by executing edge and vertex inputs nodes only over the subgraph, execution would then only yield matches where \emph{all} involved elements are in the relevant subgraph. This approach would hence fail to meet the completeness criterion of Definition~\ref{def:completeness_subgraphs}.

To enable incremental queries with complete results under the relevant subgraph, we instead propose to anchor RETE net execution to subgraph elements while allowing the search to retrieve elements outside the subgraph that are required to produce complete results from the full model.

This approach is based on the observation that, in order for a join node to create a match that involves elements from the relevant subgraph, at least one of the constituting matches already has to involve at least one such element. Intuitively, our localized querying technique then works as follows: At the leaves of a RETE join tree for some plain graph query, we initially compute only matches that touch the relevant subgraph via local edge navigation. We then also implement a mechanism that enables a lazy computation of certain complementary partial matches that do not touch the relevant subgraph at inner nodes of the join tree. This ensures completeness of results with respect to Definition~\ref{def:completeness_subgraphs}.

A na\"{i}ve implementation of this idea that retrieves complementary matches for all matches from each dependency of a join $[\bowtie]$ until a fixpoint is reached however runs the risk of triggering the computation of superfluous matches. In this scenario, some match $m_l$ for the dependency $n_l$ of $[\bowtie]$ that was only required to complement some relevant match $m_r$ from the other dependency $n_r$ would again trigger the computation of complementary matches for $n_r$, causing unnecessary computational effort.

Moreover, such an approach would create cyclic requirement relationships between the matches $m_l$ and $m_r$, where the presence of $m_l$ in the result for $n_l$ mandates the presence of $m_r$ in the result for $n_r$ and vice-versa. This may prevent the unloading of matches when elements are removed from the relevant subgraph.

To mitigate these issues, we extend the standard RETE mechanism by a marking for matches in the form of a natural number. Intuitively, we will use this marking to encode up to which height in a RETE join tree derived matches should trigger the retrieval of complementary matches. An appropriate assignment of such markings then prevents the computation of the aforementioned superfluos matches as well as cyclic requirement relationships between matches, while still guaranteeing completeness according to Definition~\ref{def:completeness_subgraphs}.

In our extension, an intermediate result in a \emph{marking-sensitive RETE net} is therefore characterized by a tuple $(m, \phi)$ of a match $m$ and a marking $\phi \in \overline{\mathbb{N}}$, where we define $\overline{\mathbb{N}} := \mathbb{N} \cup \{\infty\}$. A configuration for a marking-sensitive RETE net $(N^\Phi, p^\Phi)$ is then given by a function $\mathcal{C}^\Phi : V^{N^\Phi} \rightarrow \powset{\mathcal{M}_\Omega \times \overline{\mathbb{N}}}$.

Furthermore, in our extension, result computation distinguishes between the full host graph $H$ and the relevant subgraph $H_p \subseteq H$. For marking-sensitive RETE nodes, we hence extend the function for target result sets by a parameter for $H_p$. The target result set of a marking-sensitive RETE node $n^\Phi \in V^{N^\Phi}$ with $Q = \query{n^\Phi}$ for $H$, $H_p$, and a marking-sensitive configuration $\mathcal{C}^\Phi$ is then given by $\resultslocal{n^\Phi}{N^\Phi}{H}{H_p}{\mathcal{C}^\Phi}$, with $\resultslocal{n^\Phi}{N^\Phi}{H}{H_p}{\mathcal{C}^\Phi} \subseteq \allmatches{Q}{H} \times \overline{\mathbb{N}}$.

Consistency of marking-sensitive configurations is then defined analogously to the standard case:

\begin{defi}[Consistent Marking-Sensitive RETE Net Configurations] \label{def:consistent_marking_sensitive_configuration} 
We say that a configuration $\mathcal{C}^\Phi : V^{N^\Phi} \rightarrow \powset{\mathcal{M}_\Omega \times \overline{\mathbb{N}}}$ for a marking-sensitive RETE net $(N^\Phi, p^\Phi)$ is consistent for a host graph $H$ with relevant subgraph $H_p \subseteq H$ iff $\forall n^\Phi \in V^{N^\Phi} : \mathcal{C}^\Phi(n^\Phi) = \resultslocal{n^\Phi}{N^\Phi}{H}{H_p}{\mathcal{C}^\Phi}$.
\end{defi}
 
Note that technically, the definitions of marking-sensitive configurations and target result sets of marking-sensitive RETE nodes allow the same match to be associated with multiple markings for the same RETE node. However, as shown in \cite{preprint}, the marking of a match in the target result set of a marking-sensitive RETE node is always unique in the constructions presented in this article.

We adapt the \emph{join node}, \emph{union node}, and \emph{projection node} to marking-sensitive variants that assign result matches the maximum marking of related dependency matches and otherwise work as expected from relational algebra \cite{codd1970relational}. We also adapt the \emph{vertex input node} to only consider vertices in the relevant subgraph $H_p$ and assign matches the marking $\infty$. Finally, we introduce \emph{marking assignment nodes}, which assign matches a fixed marking value, \emph{marking filter nodes}, which filter marked matches by a minimum marking value, and \emph{forward} and \emph{backward navigation nodes}, which work similarly to edge input nodes but only extract edges that are adjacent to host graph vertices included in the current result set of a designated dependency. Note that an efficient implementation of the backward navigation node requires reverse navigability of host graph edges.

Formally, we define the target result sets of marking-sensitive RETE nodes as follows:

\begin{defi}[Target Result Sets of Marking-Sensitive RETE Nodes] \label{def:marking-sensitive_result_sets} 
Let $H$ be a host graph with relevant subgraph $H_p \subseteq H$, $(N^\Phi, p^\Phi)$ a containing marking-sensitive RETE net, and $\mathcal{C}^\Phi$ a configuration for $(N^\Phi, p^\Phi)$.
\begin{itemize}
	\item The target result set of a \emph{marking-sensitive join node} $[\bowtie]^\Phi$ with marking-sensitive dependencies $n^\Phi_l$ and $n^\Phi_r$ with $Q_l = \query{n^\Phi_l}$ and $Q_r = \query{n^\Phi_r}$ such that $V^{Q_\cap} \neq \emptyset$ for $Q_\cap = Q_l \cap Q_r$ is given by
	$\resultslocal{[\bowtie]^\Phi}{N^\Phi}{H}{H_p}{\mathcal{C}^\Phi} = \{(m_l \cup m_r, max(\phi_l, \phi_r)) \mid (m_l, \phi_l) \in \mathcal{C}^\Phi(n^\Phi_l) \wedge (m_r, \phi_r) \in \mathcal{C}^\Phi(n^\Phi_r) \wedge m_l|_{Q_\cap} = m_r|_{Q_\cap}\}$.
	\item The target result set of a \emph{marking-sensitive union node} $[\cup]^\Phi$ with a set of marking-sensitive dependencies $N^\Phi_\alpha$ such that $Q = \query{n^\Phi_\alpha}$ for all dependencies $n^\Phi_\alpha \in N^\Phi_\alpha$ is given by
	$\resultslocal{[\cup]^\Phi}{N^\Phi}{H}{H_p}{\mathcal{C}^\Phi} = \{(m, \phi_{max}) \mid (m, \phi_{max}) \in \bigcup_{n_\alpha \in N^\Phi_\alpha} \mathcal{C}^\Phi(n^\Phi_\alpha) \wedge \phi_{max} = max(\{ \phi' \mid (m, \phi') \in \bigcup_{n_\alpha \in N^\Phi_\alpha} \mathcal{C}^\Phi(n^\Phi_\alpha)\})\}$.
	\item The target result set of a \emph{marking-sensitive projection node} $[\pi_Q]^\Phi$ with a single marking-sensitive dependency $n^\Phi_\alpha$ and $Q = \query{[\pi_Q]^\Phi}$ is given by
	$\resultslocal{[\pi_Q]^\Phi}{N^\Phi}{H}{H_p}{\mathcal{C}^\Phi} = \{(m|_Q, \phi_{max}) \mid (m, \phi_{max}) \in \mathcal{C}^\Phi(n^\Phi_\alpha) \wedge \phi_{max} = max(\{ \phi' \mid (m', \phi') \in \mathcal{C}^\Phi(n^\Phi_\alpha) \wedge m'|_Q = m|_Q\})\}$.
	\item The target result set of a \emph{marking assignment node} $[\phi := i]^\Phi$ with a single marking-sensitive dependency $n^\Phi_\alpha$ is given by
	$\resultslocal{[\phi := i]^\Phi}{N^\Phi}{H}{H_p}{\mathcal{C}^\Phi} = \{(m, i) \mid (m, \phi_\alpha) \in \mathcal{C}^\Phi(n^\Phi_\alpha)\}$.
	\item The target result set of a \emph{marking filter node} $[\phi > i_{min}]^\Phi$ with marking-sensitive dependency $n^\Phi_\alpha$ is given by
	$\resultslocal{[\phi > i_{min}]^\Phi}{N^\Phi}{H}{H_p}{\mathcal{C}^\Phi} = \{(m, \phi_\alpha) \mid (m, \phi_\alpha) \in \mathcal{C}^\Phi(n^\Phi_\alpha) \wedge \phi_\alpha > i_{min}\}$.
	\item The target result set of a \emph{forward navigation node} $[v \rightarrow_n w]^\Phi$ with
	$Q = \query{[v \rightarrow_n w]^\Phi} = (\{v, w\}, \{e\}, \{(e, v)\}, \{(e, w)\})$ and marking-sensitive dependency $n^\Phi_v$ with
	$Q_v = \query{n^\Phi_v} = (\{v\}, \emptyset, \emptyset, \emptyset)$ is given by
	$\resultslocal{[v \rightarrow_n w]^\Phi}{N^\Phi}{H}{H_p}{\mathcal{C}^\Phi} = \{(m, \phi) \mid m \in \allmatches{Q}{H} \wedge (m|_{Q_v}, \phi) \in \mathcal{C}^\Phi(n^\Phi_v)\}$.
	\item The target result set of a \emph{backward navigation node} $[w \leftarrow_n v]^\Phi$ with
	$Q = \query{[w \leftarrow_n v]^\Phi} = (\{v, w\}, \{e\}, \{(e, v)\}, \{(e, w)\})$ and marking-sensitive dependency $n^\Phi_w$ with
	$Q_w = \query{n^\Phi_w} = (\{w\}, \emptyset, \emptyset, \emptyset)$ is given by
	$\resultslocal{[w \leftarrow_n v]^\Phi}{N^\Phi}{H}{H_p}{\mathcal{C}^\Phi} = \{(m, \phi) \mid m \in \allmatches{Q}{H} \wedge (m|_{Q_w}, \phi) \in \mathcal{C}^\Phi(n^\Phi_w)\}$.
	\item The target result set of a \emph{marking-sensitive vertex input node} $[v]^\Phi$ with $Q = \query{[v]^\Phi} = (\{v\}, \emptyset, \emptyset, \emptyset)$ is given by
	$\resultslocal{[v]^\Phi}{N^\Phi}{H}{H_p}{\mathcal{C}^\Phi} = \{(m, \infty) \mid m \in \allmatches{Q}{H_p}\}$.
\end{itemize}
\end{defi}

To obtain query results in the format of the standard RETE approach, we define the \emph{stripped result set} of a marking-sensitive RETE node $n^\Phi$ for a marking-sensitive configuration $\mathcal{C}^\Phi$ as the set of matches that appear in tuples in the node's current result set in $\mathcal{C}^\Phi$.

\begin{defi}[Stripped Result Sets of Marking-Sensitive RETE Nodes] \label{def:stripped_result_set} 
The \emph{stripped result set} of a marking-sensitive RETE node $n^\Phi$ for a marking-sensitive configuration $\mathcal{C}^\Phi$ is given by $\resultsstripped{n^\Phi}{\mathcal{C}^\Phi} = \{m \mid (m, \phi) \in \mathcal{C}^\Phi(n^\Phi)\}$.
\end{defi}

\subsection{Localized Search with Marking-sensitive RETE} \label{sec:localized_search_with_marking-sensitive_rete}

Based on these adaptations, we introduce a recursive $localize$ procedure, which takes a regular, well-formed\footnote{The restriction to well-formed RETE nets prevents an optimization where redundant computation of matches for isomorphic query subgraphs is avoided via a non-tree-like structure. However, for queries as in \cite{erling2015ldbc}, performance is often primarily determined by the computation of matches for a large query subgraph that cannot be reused.} RETE net $(N, p)$ as input and outputs a marking-sensitive RETE net, which performs a localized search that does not require searching the full model to produce complete results according to Definition~\ref{def:completeness_subgraphs}.

If $p = [v \rightarrow w]$ is an edge input node, the result of localization for $(N, p)$ is given by $localize(N, p) = (LNS(p), [\cup]^\Phi)$. The \emph{local navigation structure} $LNS(p)$ consists of seven RETE nodes as shown in Figure~\ref{fig:localization_structures} (left): (1, 2) Two marking-sensitive vertex input nodes $[v]^\Phi$ and $[w]^\Phi$, (3, 4) two marking-sensitive union nodes $[\cup]^\Phi_v$ and $[\cup]^\Phi_w$ with $[v]^\Phi$ respectively $[w]^\Phi$ as a dependency, (5) a forward navigation node $[v \rightarrow_n w]^\Phi$ with $[\cup]^\Phi_v$ as a dependency, (6) a backward navigation node $[w \leftarrow_n v]^\Phi$ with $[\cup]^\Phi_w$ as a dependency, and (7) a marking-sensitive union node $[\cup]^\Phi$ with dependencies $[v \rightarrow_n w]^\Phi$ and $[w \leftarrow_n v]^\Phi$.

Importantly, the marking-sensitive vertex input nodes of the local navigation structure are executed over the relevant subgraph, whereas the forward and backward navigation nodes are executed over the full model. Intuitively, the local navigation structure thus takes the role of the edge input node, but initially only extracts edges that are adjacent to a vertex in the relevant subgraph. Connecting additional dependencies to the union nodes $[\cup]^\Phi_v$ and $[\cup]^\Phi_w$ in further constructions then allows the extraction of additional edges that may be required to compute complete query results.

If $p$ is a join node, it has two dependencies $p_l$ and $p_r$ with $Q_l = \query{p_l}$ and $Q_r = \query{p_r}$, which are the roots of two RETE subtrees $N_l$ and $N_r$. In this case, $(N, p)$ is localized as $localize(N, p) = (N^\Phi, p^\Phi) = (N^\Phi_{\bowtie} \cup N^\Phi_l \cup N^\Phi_r \cup RPS_l \cup RPS_r, [\bowtie]^\Phi)$, where $(N^\Phi_l, p^\Phi_l) = localize(N_l, p_l)$, $(N^\Phi_r, p^\Phi_r) = localize(N_r, p_r)$, $N^\Phi_{\bowtie}$ consists of the marking-sensitive join $[\bowtie]^\Phi$ with dependencies $p^\Phi_l$ and $p^\Phi_r$, $RPS_l = RPS(p^\Phi_l, N^\Phi_r)$, and $RPS_r = RPS(p^\Phi_r, N^\Phi_l)$.

The \emph{request projection structure} $RPS_l = RPS(p^\Phi_l, N^\Phi_r)$ contains three RETE nodes as displayed in Figure~\ref{fig:localization_structures} (center): (1) A marking filter node $[\phi > h]^\Phi$ with $p^\Phi_l$ as a dependency, (2) a marking-sensitive projection node $[\pi_{Q_v}]^\Phi$ with $[\phi > h]^\Phi$ as a dependency, and (3) a marking assignment node $[\phi := h]^\Phi$ with $[\pi_{Q_v}]^\Phi$ as a dependency. The value of $h$ is given by the height of the join tree of which $p$ is the root and $Q_v$ is a graph consisting of a single vertex $v \in V^{Q_l \cap Q_r}$. The request projection structure is then connected to an arbitrary local navigation structure in $N^\Phi_r$ that has a marking-sensitive vertex input $[v]^\Phi$ with $Q_v = \query{[v]^\Phi}$. Therefore, it also adds a dependency from the marking-sensitive union node $[\cup]^\Phi_v$ that already depends on $[v]^\Phi$ to the marking assignment node $[\phi := h]^\Phi$. The mirrored structure $RPS_r = RPS(p^\Phi_r, N^\Phi_l)$ is constructed analogously.

Via the request projection structures, partial matches from one join dependency can be propagated to the subnet under the other dependency. Intuitively, the inserted request projection structures thereby allow the join's dependencies to request the RETE subnet under the other dependency to fetch and process the model parts required to complement the first dependency's results. The marking of a match then signals up to which height in the join tree complementarity for that match is required. Notably, matches involving elements in the relevant subgraph are marked $\infty$, as complementarity for them is required at the very top to guarantee completeness of the overall result.

The result of applying $localize$ to a RETE net consisting of a single join is displayed in Figure~\ref{fig:localization_structures} (right). It consists of a marking sensitive join and two local navigation structures connected via request projection structures.

\begin{figure} [t]
	\centering
	\begin{subfigure}[b]{0.35\textwidth}
		\begin{center}

\begin{tikzpicture}

\node (union) [ms_retenode] {$\cup$};

\node[below left = 0.35cm and -0.5cm of union] (forward) [ms_retenode] {$v \rightarrow_n w$};

\node[below right = 0.35cm and -0.5cm of union] (backward) [ms_retenode] {$w \leftarrow_n v$};

\node[below = 0.35cm of forward] (unionv) [ms_retenode] {$\cup$};

\node[below = 0.35cm of backward] (unionw) [ms_retenode] {$\cup$};

\node[below = 0.35cm of unionv] (v) [ms_retenode] {$v$};

\node[below = 0.35cm of unionw] (w) [ms_retenode] {$w$};

\draw [-{Latex}] (union) -- (forward);

\draw [-{Latex}] (union) -- (backward);

\draw [-{Latex}] (forward) -- (unionv);

\draw [-{Latex}] (backward) -- (unionw);

\draw [-{Latex}] (unionv) -- (v);

\draw [-{Latex}] (unionw) -- (w);

\end{tikzpicture}
		\end{center}
	\end{subfigure}\quad\vline\quad
	\begin{subfigure}[b]{0.15\textwidth}
		\begin{center}

\begin{tikzpicture}

\node (assignment) [ms_retenode] {$\phi := h$};

\node[below = 0.35cm of assignment] (projection) [ms_retenode] {$\pi_{Q_v}$};

\node[below = 0.35cm of projection] (filter) [ms_retenode] {$\phi > h$};

\draw [-{Latex}] (assignment) -- (projection);

\draw [-{Latex}] (projection) -- (filter);

\end{tikzpicture}
		\end{center}
	\end{subfigure}\quad\vline\quad
	\begin{subfigure}[b]{0.35\textwidth}
		\begin{center}

\begin{tikzpicture}

\node (join) [ms_retenode] {$\bowtie$};

\node[below left = 0.35cm and -0.5cm of join, dashed] (ce) [ms_retenode] {$LNS^{u \rightarrow v}$};

\node[below right = 0.35cm and -0.5cm of join, dashed] (fe) [ms_retenode] {$LNS^{v \rightarrow w}$};

\node[below = 0.35cm of ce, dashed] (rpsp) [ms_retenode] {$RPS_r$};

\node[below = 0.35cm of fe, dashed] (rpsf) [ms_retenode] {$RPS_l$};

\draw [-{Latex}] (join) -- (ce);

\draw [-{Latex}] (join) -- (fe);

\draw [-{Latex}] (ce) -- (rpsp);

\draw [-{Latex}] (fe) -- (rpsf);

\draw [-{Latex}] (rpsp) -- (fe);

\draw [-{Latex}] (rpsf) -- (ce);

\end{tikzpicture}
		\end{center}
	\end{subfigure}
	\caption{LNS (left), RPS (center), and localized RETE net (right)} \label{fig:localization_structures}
\end{figure}

\begin{figure}[t]

\begin{tikzpicture}

\node (join) [ms_retenode] {$\bowtie$};
\node[below = 0.0cm of join] (join_config) [reteconfig] {$(m_1, \infty)$};


\node[below left = 0.35cm and 2.75cm of join_config] (ul) [ms_retenode] {$\cup$};
\node[below = 0.0cm of ul]  (ul_config) [reteconfig] {$(m_{1.1}, \infty)$};

\node[below left = 0.35cm and 0.25cm of ul_config] (fl) [ms_retenode] {$p \rightarrow_n c$};
\node[below = 0.0cm of fl]  (fl_config) [reteconfig] {$(m_{1.1}, \infty)$};

\node[below right = 0.35cm and 0.25cm of ul_config] (bl) [ms_retenode] {$c \leftarrow_n p$};
\node[below = 0.0cm of bl]  (bl_config) [reteconfig] {};

\node [below = 0.35cm of fl_config] (xfl) [ms_retenode] {$\cup$};
\node[below = 0.0cm of xfl]  (xfl_config) [reteconfig] {$(m_{1.1.1}, \infty)$};

\node [below = 0.35cm of bl_config] (xbl) [ms_retenode] {$\cup$};
\node[below = 0.0cm of xbl]  (xbl_config) [reteconfig] {};

\node [below = 0.35cm of xfl_config] (vl) [ms_retenode] {$p$};
\node[below = 0.0cm of vl]  (vl_config) [reteconfig] {$(m_{1.1.1}, \infty)$};

\node [below left = 0.35cm and -0.5cm of xbl_config] (wl) [ms_retenode] {$c$};
\node[below = 0.0cm of wl]  (wl_config) [reteconfig] {};

\draw [-{Latex}] (join_config) -- (ul);

\draw [-{Latex}] (ul_config) -- (fl);

\draw [-{Latex}] (ul_config) -- (bl);

\draw [-{Latex}] (fl_config) -- (xfl);

\draw [-{Latex}] (bl_config) -- (xbl);

\draw [-{Latex}] (xfl_config) -- (vl);

\draw [-{Latex}] (xbl_config) -- (wl);


\node[below right = 0.35cm and 2.75cm of join_config] (ur) [ms_retenode] {$\cup$};
\node[below = 0.0cm of ur]  (ur_config) [reteconfig] {$(m_{1.2}, 1)$};

\node[below left = 0.35cm and 0.25cm of ur_config] (fr) [ms_retenode] {$c \rightarrow_n f$};
\node[below = 0.0cm of fr]  (fr_config) [reteconfig] {$(m_{1.2}, 1)$};

\node[below right = 0.35cm and 0.25cm of ur_config] (br) [ms_retenode] {$f \leftarrow_n c$};
\node[below = 0.0cm of br]  (br_config) [reteconfig] {};

\node [below = 0.35cm of fr_config] (xfr) [ms_retenode] {$\cup$};
\node[below = 0.0cm of xfr]  (xfr_config) [reteconfig] {$(m_{1.2.1}, 1)$};

\node [below = 0.35cm of br_config] (xbr) [ms_retenode] {$\cup$};
\node[below = 0.0cm of xbr]  (xbr_config) [reteconfig] {};

\node [below right = 0.35cm and -0.5cm of xfr_config] (vr) [ms_retenode] {$c$};
\node[below = 0.0cm of vr]  (vr_config) [reteconfig] {};

\node [below = 0.35cm of xbr_config] (wr) [ms_retenode] {$f$};
\node[below = 0.0cm of wr]  (wr_config) [reteconfig] {};

\draw [-{Latex}] (join_config) -- (ur);

\draw [-{Latex}] (ur_config) -- (fr);

\draw [-{Latex}] (ur_config) -- (br);

\draw [-{Latex}] (fr_config) -- (xfr);

\draw [-{Latex}] (br_config) -- (xbr);

\draw [-{Latex}] (xfr_config) -- (vr);

\draw [-{Latex}] (xbr_config) -- (wr);


\node[below right = 0.35cm and -0.5cm of xbl_config] (mar) [ms_retenode] {$\phi := 1$};
\node[below = 0.0cm of mar]  (mar_config) [reteconfig] {};

\node[below = 0.35cm of mar_config] (pr) [ms_retenode] {$\pi_{Q_c}$};
\node[below = 0.0cm of pr]  (pr_config) [reteconfig] {};

\node[below = 0.35cm of pr_config] (mfr) [ms_retenode] {$\phi > 1$};
\node[below = 0.0cm of mfr]  (mfr_config) [reteconfig] {};

\draw [-{Latex}] (xbl_config) -- (mar);

\draw [-{Latex}] (mar_config) -- (pr);

\draw [-{Latex}] (pr_config) -- (mfr);

\draw [-{Latex}] (mfr_config) -- (-1.3,-12.6) -- (0.15, -12.6) -- (0.15, -1.85) -- (ur);


\node[below left = 0.35cm and -0.5cm of xfr_config] (mal) [ms_retenode] {$\phi := 1$};
\node[below = 0.0cm of mal]  (mal_config) [reteconfig] {$(m_{1.2.1}, 1)$};

\node[below = 0.35cm of mal_config] (pl) [ms_retenode] {$\pi_{Q_c}$};
\node[below = 0.0cm of pl]  (pl_config) [reteconfig] {$(m_{1.2.1}, \infty)$};

\node[below = 0.35cm of pl_config] (mfl) [ms_retenode] {$\phi > 1$};
\node[below = 0.0cm of mfl]  (mfl_config) [reteconfig] {$(m_{1.1}, \infty)$};

\draw [-{Latex}] (xfr_config) -- (mal);

\draw [-{Latex}] (mal_config) -- (pl);

\draw [-{Latex}] (pl_config) -- (mfl);

\draw [-{Latex}] (mfl_config) -- (1.3,-12.85) -- (-0.15, -12.85) -- (-0.15, -1.85) -- (ul);

\end{tikzpicture}
	\caption{Localized RETE net for the example query in Figure~\ref{fig:example_query_and_rete} with a consistent configuration for the host graph in Figure~\ref{fig:example_host_graph}} \label{fig:localized_rete_net}
\end{figure}

The localized RETE net for the example query in Figure~\ref{fig:example_query_and_rete} is shown in Figure~\ref{fig:localized_rete_net}. Alongside the RETE net, the figure visualizes a consistent configuration for the host graph in Figure~\ref{fig:example_host_graph}. In the example host graph, the relevant subgraph consists of the node $p_1$, which is indicated by a bold border.

With the relevant subgraph consisting of only a single node, the only intermediate result in any of the net's vertex input nodes is the tuple $(m_{1.1.1}, \infty)$ for the node $[p]^\Phi$. This intermediate result propagates up through the containing local navigation structure on the left, where it triggers the extraction of the match $m_{1.1}$ by the forward navigation node $[p \rightarrow_n c]^\Phi$. Via the connected request projection structure, the restricted match $m_{1.2.1}$ that consists only of the target node $c_1$ of the extracted edge is then propagated into the local navigation structure on the right. As a result, the match $m_{1.2}$ is extracted by the forward navigation node $[c \rightarrow_n f]^\Phi$, which is combined with the match $m_{1.1}$ by the join node $[\bowtie]^\Phi$ to form the complete match $m_1$.

Notably, the host graph contains a second match $m_2$ for the query pattern, which however does not touch the relevant subgraph. In this example, localization allows the RETE net to completely ignore this second match and skip any related computations.

Moreover, the assignment of the marking $1$ to the match $m_{1.2.1}$ in the request projection structure on the right prevents a propagation of any of the derived intermediate results back into the left local navigation structure. Thereby, potentially problematic cyclic requirement relationships of intermediate results are avoided.

As an example, consider the case where the node $p_1$ is removed from the relevant subgraph and this change is to be reflected in an updated configuration. Without the additional control of propagation provided by the marking filter and marking assignment nodes, the removal would not mandate an unloading of all associated intermediate results. Specifically, the result set for the union node below the backward navigation node $[c \leftarrow_n f]^\Phi$ would still contain the match $m_{1.2.1}$ and the backward navigation node would hence extract the match $m_{1.1}$ again, leading to the pollution of the RETE net's configuration with superfluous intermediate results.

\begin{figure}[t]

\begin{tikzpicture}

\node (r1) [vertex] {\textit{r\textsubscript{1}:Proj}};


\node[below left = 1.25cm and 1.25cm of r1, line width = 4pt] (p1) [vertex] {\textbf{\textit{p\textsubscript{1}:Pkg}}};

\node[below = 1.25cm of p1] (c1) [vertex] {\textit{c\textsubscript{1}:Class}};

\node[below = 1.25cm of c1] (f1) [vertex] {\textit{f\textsubscript{1}:Field}};

\draw [-{Latex}] (r1) -- (p1) node [midway, right] {\textit{pe}};

\draw [-{Latex}] (p1) -- (c1) node [midway, right] {\textit{ce}};

\draw [-{Latex}] (p1) -- (c1) node [midway, right] {\textit{ce}};

\draw [-{Latex}] (c1) -- (f1) node [midway, right] {\textit{fe}};


\node[below right = 1.25cm and 1.25cm of r1] (p2) [vertex] {\textit{p\textsubscript{2}:Pkg}};

\node[below = 1.25cm of p2] (c2) [vertex] {\textit{c\textsubscript{2}:Class}};

\node[below = 1.25cm of c2] (f2) [vertex] {\textit{f\textsubscript{2}:Field}};

\draw [-{Latex}] (r1) -- (p2) node [midway, right] {\textit{pe}};

\draw [-{Latex}] (p2) -- (c2) node [midway, right] {\textit{ce}};

\draw [-{Latex}] (c2) -- (f2) node [midway, right] {\textit{fe}};


\node[above left = -1.8cm and -2.75cm of f1, rectangle, draw = black, dashed, minimum width = 4.5cm, minimum height = 9.25cm, text depth = 9cm, text width = 4.25cm] (m1) {\textit{m\textsubscript{1}}};

\node[above left = -1.8cm and -2.25cm of c1, rectangle, draw = orange, dashed, minimum width = 3.75cm, minimum height = 5.75cm, text depth = 5.5cm, text width = 3.5cm] (m11) {\textit{m\textsubscript{1.1}}};

\node[above left = -1.55cm and -2cm of p1, rectangle, draw = orange, dashed, minimum width = 3cm, minimum height = 2cm, text depth = 1.75cm, text width = 2.75cm] (m111) {\textit{m\textsubscript{1.1.1}}};

\node[above left = -1.55cm and -2.5cm of f1, rectangle, draw = blue, dashed, minimum width = 3.75cm, minimum height = 5.5cm, text depth = 5.25cm, text width = 3.5cm] (m12) {\textit{m\textsubscript{1.2}}};

\node[above left = -1.55cm and -2cm of c1, rectangle, draw = blue, dashed, minimum width = 3cm, minimum height = 2cm, text depth = 1.75cm, text width = 2.75cm] (m121) {\textit{m\textsubscript{1.2.1}}};

\node[above left = -1.8cm and -2.75cm of f2, rectangle, draw = black, dashed, minimum width = 4.5cm, minimum height = 9.25cm, text depth = 9cm, text width = 4.25cm, align = right] (m2) {\textit{m\textsubscript{2}}};

\end{tikzpicture}
	\caption{Example host graph with relevant subgraph marked in bold} \label{fig:example_host_graph}
\end{figure}

Formally, a consistent configuration for a localized RETE net then indeed guarantees query results that are complete according to Definition~\ref{def:completeness_subgraphs}:

\begin{thm}[Consistent configurations for RETE nets localized via $localize$ yield complete query results under the relevant subgraph] \label{the:completeness_consistent_configuration} 
Let $H$ be a graph, $H_p \subseteq H$, $(N, p)$ a well-formed RETE net, and $Q = \query{p}$. Furthermore, let $\mathcal{C}^\Phi$ be a consistent configuration for the localized RETE net $(N^\Phi, p^\Phi) = localize(N, p)$. The set of matches from $Q$ into $H$ given by the stripped result set $\resultsstripped{p^\Phi}{\mathcal{C}^\Phi}$ is then complete under $H_p$.
\end{thm}

\begin{proof} (Idea)\footnote{Detailed proofs for theorems in this article are given in Appendix~\ref{app:technical_details} or, if indicated, in \cite{preprint}.}
It can be shown via induction over the height of $N$ that request projection structures ensure the construction of all intermediate results required to guarantee completeness of the overall result under $H_p$. See the proof of Theorem 5 in Appendix D of \cite{preprint} for a full proof.
\end{proof}

Interestingly, exactly those matches that touch the relevant subgraph are marked $\infty$ in a consistent configuration for a localized RETE net, which allows an easy recognition of these matches in the result set of the net's production node:

\begin{thm}[Matches are marked $\infty$ in query results of RETE nets localized via $localize$ iff they touch the relevant subgraph] \label{the:precision_consistent_configuration}
Let $H$ be a graph, $H_p \subseteq H$, $(N, p)$ a well-formed RETE net, and $Q = \query{p}$. Furthermore, let $\mathcal{C}^\Phi$ be a consistent configuration for the localized RETE net $(N^\Phi, p^\Phi) = localize(N, p)$. It then holds that $\forall (m, \phi) \in \mathcal{C}^\Phi(p^\Phi) : m(Q) \cap H_p \neq \emptyset \Leftrightarrow \phi = \infty$.
\end{thm}

\begin{proof}(Idea)
Follows since marking-sensitive joins ultimately combine matches from local navigation structures in $(N^\Phi, p^\Phi)$, where only those matches that touch $H_p$ are marked $\infty$.
\end{proof}

Notably, the insertion of request projection structures creates cycles in the localized RETE net, which prevents execution via a reverse topological sorting. However, the marking filter and marking assignment nodes in the request projection structures effectively prevent cyclic execution at the level of intermediate results: Because matches in the result set of a dependency of some join at height $h$ that are only computed on request from the other dependency are marked $h$, these matches are filtered out in the dependent request projection structure. An execution order for the localized RETE net $(N^\Phi, p^\Phi) = localize(N, p)$ can thus be constructed recursively via an $order$ procedure as follows:

If $p$ is an edge input node, the RETE graph given by $N^\Phi = LNS(p)$ is a tree that can be executed via a reverse topological sorting of the nodes in $LNS(p)$, that is, $order(N^\Phi) = toposort(LNS(p))^{-1}$.

If $p$ is a join, according to the construction, the localized RETE graph is given by $N^\Phi = N^\Phi_{\bowtie} \cup N^\Phi_l \cup N^\Phi_r \cup RPS_l \cup RPS_r$. In this case, an execution order for $(N^\Phi, p^\Phi)$ can be derived via the concatenation $order(N^\Phi) = order(RPS_r) \circ order(N^\Phi_l) \circ order(RPS_l) \circ order(N^\Phi_r) \circ order(RPS_r) \circ order(N^\Phi_l) \circ order(N^\Phi_{\bowtie})$, where $order(RPS_l) = toposort(RPS_l)^{-1}$ and $order(RPS_r) = toposort(RPS_r)^{-1}$ and $order(N^\Phi_{\bowtie}) = [p^\Phi]$, that is, the sequence containing only $p^\Phi$.

Executing a localized RETE net $(N^\Phi, p^\Phi)$ via $order(N^\Phi)$ then guarantees a consistent result configuration for any starting configuration:

\begin{thm}[Execution of localized RETE nets via $order$ yields consistent configurations] \label{the:execution_order_consistency} 
Let $H$ be a graph, $H_p \subseteq H$, $(N, p)$ a well-formed RETE net, and $\mathcal{C}^\Phi_0$ an arbitrary starting configuration. Executing $(N^\Phi, p^\Phi) = localize(N, p)$ via $O = order(N^\Phi)$ then yields a consistent configuration $\mathcal{C}^\Phi = execute(O, N^\Phi, H, H_p, \mathcal{C}^\Phi_0)$.
\end{thm}

\begin{proof} (Idea)
Follows because the inserted marking filter nodes prevent cyclic execution behavior at the level of intermediate results. See the proof of Theorem 6 in Appendix D of \cite{preprint} for a full proof.
\end{proof}

Combined with the result from Theorem~\ref{the:completeness_consistent_configuration}, this means that a localized RETE net can be used to compute complete query results for a relevant subgraph in the sense of Definition~\ref{def:completeness_subgraphs}, as outlined in the following corollary:

\begin{cor}[Execution of RETE nets localized via $localize$ yields complete query results under the relevant subgraph] \label{cor:completeness_execution} 
Let $H$ be a graph, $H_p \subseteq H$, $(N, p)$ a well-formed RETE net, and $Q = \query{p}$. Furthermore, let $\mathcal{C}^\Phi_0$ be an arbitrary starting configuration for the marking-sensitive RETE net $(N^\Phi, p^\Phi) = localize(N, p)$ and $\mathcal{C}^\Phi = execute(order(N^\Phi), N^\Phi, H, H_p, \mathcal{C}^\Phi_0)$. The set of matches from $Q$ into $H$ given by $\resultsstripped{p^\Phi}{\mathcal{C}^\Phi}$ is then complete under $H_p$.
\end{cor}

\begin{proof}
Follows directly from Theorem~\ref{the:completeness_consistent_configuration} and Theorem~\ref{the:execution_order_consistency}.
\end{proof}

\subsection{Performance of Localized RETE Nets} \label{sec:localized_rete_performance}

Performance of a RETE net $(N, p)$ with respect to both execution time and memory consumption is largely determined by the \emph{effective size} of a consistent configuration $\mathcal{C}$ for $(N, p)$, which we define as follows:

\begin{defi}[Effective Size of RETE Net Configurations] \label{def:effective_configuration_size} 
The \emph{effective size} of a configuration $\mathcal{C}$ for a RETE net $(N, p)$ is given by $|\mathcal{C}|_e := \sum_{n \in V^N} \sum_{m \in \mathcal{C}(n)} |m|$, where we define the size of a match $m : Q \rightarrow H$ as $|m| := |m^V| + |m^E| = |V^Q| + |E^Q|$.
\end{defi}

It can then be shown that localization incurs only a constant factor overhead on the effective size of $\mathcal{C}$ for any edge-dominated host graph:\footnote{For non-edge-dominated host graphs, the number of matches for (marking-sensitive) vertex input nodes may exceed the number of matches for related edge-input nodes. If matches for marking-sensitive vertex input nodes make up the bulk of intermediate results, localization then introduces an overhead on effective configuration size that cannot be characterized by a constant factor. However, we expect this situation to be rare in practice, since it requires very sparse host graphs.}

\begin{thm}[RETE net localization via $localize$ introduces only a constant factor overhead on effective configuration size] \label{the:upper_bound_configuration_size}
Let $H$ be an edge-dominated graph, $H_p \subseteq H$, $(N, p)$ a well-formed RETE net with $Q = \query{p}$, $\mathcal{C}$ a consistent configuration for $(N, p)$ for host graph $H$, and $\mathcal{C}^\Phi$ a consistent configuration for the marking-sensitive RETE net $(N^\Phi, p^\Phi) = localize(N, p)$ for host graph $H$ and relevant subgraph $H_p$. It then holds that $\sum_{n^\Phi \in V^{N^\Phi}} \sum_{(m, \phi) \in \mathcal{C}(n^\Phi)} |m| \leq 7 \cdot |\mathcal{C}|_e$.
\end{thm}

\begin{proof} (Idea)
Follows because for each RETE node in $(N, p)$, $(N^\Phi, p^\Phi)$ includes a constant number of nodes whose result sets contain corresponding matches and marking-sensitive result sets contain no duplicate matches.
\end{proof}

By Theorem~\ref{the:upper_bound_configuration_size}, it then follows that localization of a RETE net incurs only a constant factor overhead on memory consumption even in the worst case where the relevant subgraph is equal to the full model:

\begin{cor}[RETE net localization via $localize$ introduces only a constant factor overhead on memory consumption]
Let $H$ be an edge-dominated graph, $H_p \subseteq H$, $(N, p)$ a well-formed RETE net, $\mathcal{C}$ a consistent configuration for $(N, p)$ for host graph $H$, and $\mathcal{C}^\Phi$ a consistent configuration for the localized RETE net $(N^\Phi, p^\Phi) = localize(N, p)$ for host graph $H$ and relevant subgraph $H_p$. Assuming that storing a match $m$ requires an amount of memory in $O(|m|)$ and storing an element from $\overline{\mathbb{N}}$ requires an amount of memory in $O(1)$, storing $\mathcal{C}^\Phi$ requires memory in $O(|\mathcal{C}|_e)$.
\end{cor}

\begin{proof} (Idea)
Follows from Theorem~\ref{the:upper_bound_configuration_size}.
\end{proof}

In the worst case, a localized RETE net would still be required to compute a complete result. In this scenario, the execution of the localized RETE net may essentially require superfluous recomputation of match markings, causing computational overhead. When starting with an empty configuration, the number of such recomputations per match is however limited by the size of the query graph, only resulting in a small increase in computational complexity:

\begin{thm}[Execution time overhead introduced by RETE net localization via $localize$ depends on query graph size compared to average match size] \label{the:complexity_time_localized} 
Let $H$ be an edge-dominated graph, $H_p \subseteq H$, $(N, p)$ a well-formed RETE net for query graph $Q$, $\mathcal{C}$ a consistent configuration for $(N, p)$, and $\mathcal{C}^\Phi_0$ the empty configuration for $(N^\Phi, p^\Phi) = localize(N, p)$. Executing $(N^\Phi, p^\Phi)$ via $execute(order(N^\Phi), N^\Phi, H, H_p, \mathcal{C}^\Phi_0)$ then takes $O(T \cdot (|Q_a| + |Q|))$ steps, with $|Q_a|$ the average size of matches in $\mathcal{C}$ and $T = \sum_{n \in V^N} |\mathcal{C}(n)|$.
\end{thm}

\begin{proof} (Idea)
Follows since the effort for initial construction of matches by the marking-sensitive RETE net is linear in the total size of the constructed matches and the marking of a match changes at most $|Q|$ times. See the proof of Theorem 9 in Appendix D of \cite{preprint} for a full proof.
\end{proof}

Assuming an empty starting configuration, a regular well-formed RETE net $(N, p)$ can be executed in $O(|\mathcal{C}|_e)$ steps, which can also be expressed as $O(T \cdot |Q_a|)$. The overhead of a localized RETE net compared to the original net can thus be characterized by the factor $\frac{|Q|}{|Q_a|}$. Assuming that matches for the larger query subgraphs constitute the bulk of intermediate results, which seems reasonable in many scenarios, $\frac{|Q|}{|Q_a|}$ may be approximated by a constant factor.

For non-empty starting configurations and incremental changes, no sensible guarantees can be made. On the one hand, in a localized RETE net, a host graph modification may trigger the computation of a large number of intermediate results that were previously omitted due to localization. On the other hand, in a standard RETE net, a modification may result in substantial effort for constructing superfluous intermediate results that can be avoided by localization. Depending on the exact host graph structure and starting configuration, execution may thus essentially require a full recomputation for either the localized or standard RETE net but cause almost no effort for the other variant.


\section{Incremental Extended Queries over Subgraphs} \label{sec:incremental_extended_queries_over_subgraphs}

Standard RETE nets can be employed to incrementally compute results for extended graph queries via semi-join and anti-join nodes \cite{barkowsky2023host}. It is thus desirable to extend the localization technique for simple graph queries presented in Section~\ref{sec:incremental_queries_over_persistent_models} to also enable the \emph{localized} incremental execution of extended graph queries.

However, multiple interpretations of local correctness of query results for extended graph queries are conceivable depending on the application scenario. More specifically, it is unclear whether this notion should cover only query results where the main match touches the relevant subgraph or also those where elements from the relevant subgraph play a role regarding the satisfaction of the equipped nested graph condition.

Notably, it has been shown that nested graph conditions are equivalent to first order logic \cite{habel2009correctness,rensink2004representing} and satisfiability for first order logic is known to be undecidable. A precise solution for the second case thus seems infeasible, as discussed in more detail in Section~\ref{sec:localized_rete_satisfaction_changes}. In the following, we will therefore first develop an extension of our localization technique that covers the first case and then proceed by presenting a further extension that covers the second case at a practical level via overapproximation.

\subsection{Localized RETE for Extended Queries} \label{sec:localized_rete_extended_queries}

In the most basic case, a set of matches for an extended graph query $(Q, \psi)$ may be considered correct under a given relevant subgraph if it contains no matches that violate $\psi$ and it contains all matches for $Q$ that satisfy the condition $\psi$ and involve elements from the relevant subgraph:

\begin{defi}[Correctness of Extended Query Results under Subgraphs] \label{def:correctness_subgraphs_ngcs} 
	We say that a set of matches $M$ for the extended graph query $(Q, \psi)$ into a host graph $H$ is correct under a subgraph $H_p \subseteq H$
	iff $M \subseteq \{m \in \allmatches{Q}{H} \mid m \models \psi\}$ and $\{m \in \allmatches{Q}{H} \mid m \models \psi \wedge m(Q) \cap H_p \neq \emptyset\} \subseteq M$.
\end{defi}

To enable a localized execution of extended graph queries that meets this definition via the RETE approach, we first have to introduce marking-sensitive versions of the semi-join and anti-join node. These nodes work similarly to the regular versions, but also preserve the marking of matches from their left dependency:

\begin{defi}[Target Result Sets of Advanced Marking-Sensitive RETE Nodes] \label{def:marking-sensitive_result_sets_ngcs} 
Let $H$ be a host graph with relevant subgraph $H_p \subseteq H$, $(N^\Phi, p^\Phi)$ a containing marking-sensitive RETE net, and $\mathcal{C}^\Phi$ a configuration for $(N^\Phi, p^\Phi)$.
\begin{itemize}
	\item The target result set of a \emph{marking-sensitive semi-join node} $[\ltimes]^\Phi$ with marking-sensitive dependencies $n^\Phi_l$ and $n^\Phi_r$ with $Q_l = \query{n^\Phi_l}$ and $Q_r = \query{n^\Phi_r}$ such that $V^{Q_\cap} \neq \emptyset$ for $Q_\cap = Q_l \cap Q_r$ is given by
	$\resultslocal{[\ltimes]^\Phi}{N^\Phi}{H}{H_p}{\mathcal{C}^\Phi} = \{(m_l, \phi_l) \mid (m_l, \phi_l) \in \mathcal{C}^\Phi(n^\Phi_l) \wedge \exists (m_r, \phi_r) \in \mathcal{C}^\Phi(n^\Phi_r) : m_l|_{Q_\cap} = m_r|_{Q_\cap}\}$.
	\item The target result set of a \emph{marking-sensitive semi-join node} $[\ltimes]^\Phi$ along a non-empty graph morphism $a : Q_l \rightarrow Q_r$ with marking-sensitive dependencies $n^\Phi_l$ and $n^\Phi_r$ with $Q_l = \query{n^\Phi_l}$ and $Q_r = \query{n^\Phi_r}$ is given by
	$\resultslocal{[\ltimes]^\Phi}{N^\Phi}{H}{H_p}{\mathcal{C}^\Phi} = \{(m_l, \phi_l) \mid (m_l, \phi_l) \in \mathcal{C}^\Phi(n^\Phi_l) \wedge \exists (m_r, \phi_r) \in \mathcal{C}^\Phi(n^\Phi_r) : m_l = m_r \circ a\}$.
	\item The target result set of a \emph{marking-sensitive semi-join node} $[\ltimes]^\Phi$ along a non-empty partial graph morphism $a : Q_l \rightarrow Q_r$ from some subgraph $Q_p \subseteq Q_l$ into $Q_r$ with marking-sensitive dependencies $n^\Phi_l$ and $n^\Phi_r$ with $Q_l = \query{n^\Phi_l}$ and $Q_r = \query{n^\Phi_r}$ is given by
	$\resultslocal{[\ltimes]^\Phi}{N^\Phi}{H}{H_p}{\mathcal{C}^\Phi} = \{(m_l, \phi_l) \mid (m_l, \phi_l) \in \mathcal{C}^\Phi(n^\Phi_l) \wedge \exists (m_r, \phi_r) \in \mathcal{C}^\Phi(n^\Phi_r) : m_l|_{Q_p} = m_r \circ a\}$.
	\item The target result set of a \emph{marking-sensitive anti-join node} $[\rhd]^\Phi$ with marking-sensitive dependencies $n^\Phi_l$ and $n^\Phi_r$ with $Q_l = \query{n^\Phi_l}$ and $Q_r = \query{n^\Phi_r}$ such that $V^{Q_\cap} \neq \emptyset$ for $Q_\cap = Q_l \cap Q_r$ is given by
	$\resultslocal{[\rhd]^\Phi}{N^\Phi}{H}{H_p}{\mathcal{C}^\Phi} = \{(m_l, \phi_l) \mid (m_l, \phi_l) \in \mathcal{C}^\Phi(n^\Phi_l) \wedge \nexists (m_r, \phi_r) \in \mathcal{C}^\Phi(n^\Phi_r) : m_l|_{Q_\cap} = m_r|_{Q_\cap}\}$.
\end{itemize}
\end{defi}

Based on the $localize$ procedure for plain graph queries from Section~\ref{sec:localized_search_with_marking-sensitive_rete} and our RETE net construction technique for extended graph queries from \cite{barkowsky2023host}, we can then construct a localization procedure $localize^\Psi$ that can be applied to an extended graph query $(Q, \psi)$ to create a localized RETE net $(N^\Phi, p^\Phi) = localize^\Psi(Q, \psi)$. $localize^\Psi$ is a recursive procedure, which is visualized in Figure~\ref{fig:localize_psi} and works as follows:

\begin{itemize}
\item For a nested graph condition of the form $\psi = true$, the result of $localize^\Psi$ is given by $(N^\Phi, p^\Phi) = (N^\Phi_Q, p^\Phi_Q) = localize(Q)$\footnote{Here, we use $localize(Q)$ as a shorthand for $localize(N, p)$ with an arbitrary well-formed RETE net $(N, p)$ for $Q$.}.

\item For a nested graph condition of the form $\psi = \exists (a: Q \rightarrow Q', \psi')$, $N^\Phi$ consists of the localized RETE net for the plain pattern $Q$, $(N^\Phi_Q, p^\Phi_Q) = localize(Q)$, the localized RETE net for the query $(N^\Phi_{(Q', \psi')}, p^\Phi_{(Q', \psi')}) = localize^\Psi(Q', \psi')$, a request projection structure $RPS^{\infty}_l = RPS^{\infty}(p^\Phi_Q, N^\Phi_{(Q', \psi')})$, and a marking-sensitive semi-join node $p^\Phi = [\ltimes]^\Phi$ along $a$ with left dependency $p^\Phi_Q$ and right dependency $p^\Phi_{(Q', \psi')}$.

\item For a nested graph condition of the form $\psi = \neg\psi'$, $N^\Phi$ consists of the localized RETE net for the plain pattern $Q$, $(N^\Phi_Q, p^\Phi_Q) = localize(Q)$, the localized RETE net for the query $(N^\Phi_{(Q, \psi')}, p^\Phi_{(Q, \psi')}) = localize^\Psi(Q, \psi')$, a request projection structure $RPS^{\infty}_l = RPS^{\infty}(p^\Phi_Q, N^\Phi_{(Q, \psi')})$, and a marking-sensitive anti-join node $p^\Phi = [\rhd]^\Phi$ with left dependency $p^\Phi_Q$ and right dependency $p^\Phi_{(Q, \psi')}$.

\item For a nested graph condition of the form $\psi = \psi_1 \wedge \psi_2$, $N^\Phi$ consists of
the localized RETE net for the plain pattern $Q$, $(N^\Phi_Q, p^\Phi_Q) = localize(Q)$,
the RETE net $(N^\Phi_{(Q, \psi_1)}, p^\Phi_{(Q, \psi_1)}) = localize^\Psi(Q, \psi_1)$, a request projection structure $RPS^{\infty}_1 = RPS^{\infty}(p^\Phi_Q, N^\Phi_{(Q, \psi_1)})$, a marking-sensitive semi-join node $[\ltimes]^\Phi_1$ with left dependency $p^\Phi_Q$ and right dependency $p^\Phi_{(Q, \psi_1)}$,
the RETE net $(N^\Phi_{(Q, \psi_2)}, p^\Phi_{(Q, \psi_2)}) = localize^\Psi(Q, \psi_2)$, a request projection structure $RPS^{\infty}_2 = RPS^{\infty}([\ltimes]^\Phi_1, N^\Phi_{(Q, \psi_2)})$, and a marking-sensitive semi-join node $p^\Phi = [\ltimes]^\Phi_2$ with left dependency $[\ltimes]^\Phi_1$ and right dependency $p^\Phi_{(Q, \psi_2)}$.
\end{itemize}

Note that we include semi-join and anti-join nodes when computing join tree height for determining the value of marking filter nodes in request projection structures. Moreover, by $RPS^{\infty}$, we denote a request projection structure where the marking assignment node assigns a marking of $\infty$. While the marking assignment nodes in these request projection structures are technically redundant when considering the constructions in this article and could be removed, we include the nodes to mirror the structure of the request projection structures from Section~\ref{sec:incremental_queries_over_persistent_models}.

Intuitively, by employing localized RETE nets in the computation of nested graph condition satisfaction, it is no longer necessary to compute related matches over the entire host graph. The insertion of request projection structures via $localize^\Psi$ then ensures that matches that are required to correctly evaluate a nested graph condition are computed on demand for every relevant match for the context pattern.

\begin{figure} [t]
	\centering
	\begin{subfigure}[b]{0.36\textwidth}
		\begin{center}

\begin{tikzpicture}

\node[rectangle,
	minimum width = 4.15cm, 
	minimum height = 3.25cm] (r) at (0,0) {};

\node [below = -3.25cm of r, dashed] [ms_retenode] {$N^\Phi_Q$};

\end{tikzpicture}
		\end{center}
	\end{subfigure}\hspace{4pt}\vline\hspace{4pt}
	\begin{subfigure}[b]{0.55\textwidth}
		\begin{center}

\begin{tikzpicture}

\node[rectangle,
	minimum width = 6.85cm, 
	minimum height = 3.25cm] (r) at (0,0) {};

\node [below = -3.25cm of r] (join) [ms_retenode] {$\ltimes$};

\node[below left = 0.35cm and -0.5cm of join, dashed] (q) [ms_retenode] {$N^\Phi_Q$};

\node[below right = 0.35cm and -0.5cm of join, dashed] (qp_psip) [ms_retenode] {$N^\Phi_{(Q', \psi')}$};

\node[below = 0.35cm of qp_psip, dashed] (rpsl) [ms_retenode] {$RPS^{\infty}_l$};

\draw [-{Latex}] (join) -- (q);

\draw [-{Latex}] (join) -- (qp_psip);

\draw [-{Latex}] (qp_psip) -- (rpsl);

\draw [-{Latex}] (rpsl) -- (q);

\end{tikzpicture}
		\end{center}
	\end{subfigure}

	\hrule

	\begin{subfigure}[b]{0.36\textwidth}
		\begin{center}

\begin{tikzpicture}

\node[rectangle,
	minimum width = 4.15cm, 
	minimum height = 4.3cm] (r) at (0,0) {};

\node [below = -4.1cm of r] (join) [ms_retenode] {$\rhd$};

\node[below left = 0.35cm and -0.5cm of join, dashed] (q) [ms_retenode] {$N^\Phi_Q$};

\node[below right = 0.35cm and -0.5cm of join, dashed] (q_psip) [ms_retenode] {$N^\Phi_{(Q, \psi')}$};

\node[below = 0.35cm of q_psip, dashed] (rpsl) [ms_retenode] {$RPS^{\infty}_l$};

\draw [-{Latex}] (join) -- (q);

\draw [-{Latex}] (join) -- (q_psip);

\draw [-{Latex}] (q_psip) -- (rpsl);

\draw [-{Latex}] (rpsl) -- (q);

\end{tikzpicture}
		\end{center}
	\end{subfigure}\hspace{4pt}\vline\hspace{4pt}
	\begin{subfigure}[b]{0.55\textwidth}
		\begin{center}

\begin{tikzpicture}

\node[rectangle,
	minimum width = 6.85cm, 
	minimum height = 4.3cm] (r) at (0,0) {};

\node [below left = -4.1cm and -4.9cm of r] (join2) [ms_retenode] {$\ltimes_2$};

\node [below left = 0.35cm and 0.25cm of join2] (join1) [ms_retenode] {$\ltimes_1$};

\node[below right = 0.35cm and 0.25cm of join2, dashed] (q_psi2) [ms_retenode] {$N^\Phi_{(Q, \psi_2)}$};

\node[below = 0.35cm of q_psi2, dashed] (rps2) [ms_retenode] {$RPS^{\infty}_2$};

\node[below left = 0.35cm and -0.5cm of join1, dashed] (q) [ms_retenode] {$N^\Phi_Q$};

\node[below right = 0.35cm and -0.5cm of join1, dashed] (q_psi1) [ms_retenode] {$N^\Phi_{(Q, \psi_1)}$};

\node[below = 0.35cm of q_psi1, dashed] (rps1) [ms_retenode] {$RPS^{\infty}_1$};

\draw [-{Latex}] (join2) -- (join1);

\draw [-{Latex}] (join2) -- (q_psi2);

\draw [-{Latex}] (q_psi2) -- (rps2);

\draw [-{Latex}] (rps2) -- (join1);

\draw [-{Latex}] (join1) -- (q);

\draw [-{Latex}] (join1) -- (q_psi1);

\draw [-{Latex}] (q_psi1) -- (rps1);

\draw [-{Latex}] (rps1) -- (q);

\end{tikzpicture}
		\end{center}
	\end{subfigure}
	\caption{Results of applying $localize^\Psi$ to an extended graph query $(Q, \psi)$ with nested graph conditions of the form $\psi = \texttt{true}$ (top left), $\psi = \exists(a : Q \rightarrow Q', \psi')$ (top right), $\psi = \neg\psi'$ (bottom left), and $\psi = \psi_1 \wedge \psi_2$ (bottom right)} \label{fig:localize_psi}
\end{figure}

The localized RETE net for the extended graph query in Figure~\ref{fig:example_query_and_rete_ngc} is visualized in Figure~\ref{fig:localized_rete_net_ngc} along with a consistent configuration for the host graph in Figure~\ref{fig:example_host_graph_ngc}. For local navigation structures and request projection structures, the displayed current result set is the current result set of the topmost union node respectively the marking assignment node of the structure.

Analogously to Figure~\ref{fig:example_query_and_rete_ngc}, the localized RETE net essentially extends the RETE net from Figure~\ref{fig:localized_rete_net} by a marking sensitive semi-join and a local navigation structure for matching the pattern $Q'$. In addition, it contains a request projection structure that controls the computation of matches for $Q'$ based on matches for $Q$ that touch the relevant subgraph.

Via this new request projection structure, the match $m_1$ for $Q$ triggers the extraction of the match $m_{1.3}$ for $Q'$ in the local navigation structure $LNS^{c \rightarrow i}$. This ensures that $m_1$ is included in the result set of the semi-join node at the top of the net, which thereby provides a locally correct query result.

\begin{figure}[t]

\begin{tikzpicture}

\node (semijoin) [ms_retenode] {$\ltimes$};
\node[below = 0.0cm of semijoin]  (semijoin_config) [reteconfig] {$(m_{1}, \infty)$};


\node[below left = 0.35cm and 0.25cm of semijoin_config] (join) [ms_retenode] {$\bowtie$};
\node[below = 0.0cm of join]  (join_config) [reteconfig] {$(m_{1}, \infty)$};

\node[below left = 0.35cm and -0.5cm of join_config, dashed] (ce) [ms_retenode] {$LNS^{p \rightarrow c}$};
\node[below = 0.0cm of ce]  (ce_config) [reteconfig] {$(m_{1.1}, \infty)$};

\node[below right = 0.35cm and -0.5cm of join_config, dashed] (fe) [ms_retenode] {$LNS^{c \rightarrow f}$};
\node[below = 0.0cm of fe]  (fe_config) [reteconfig, minimum width = 1.8cm] {$(m_{1.2}, 1)$};

\node[below = 0.35cm of ce_config, dashed] (rps1) [ms_retenode] {$RPS_c$};
\node[below = 0.0cm of rps1]  (rps1_config) [reteconfig] {$$};

\node[below = 0.35cm of fe_config, dashed] (rps2) [ms_retenode] {$RPS_c$};
\node[below = 0.0cm of rps2]  (rps2_config) [reteconfig] {$(m_{1.2.1}, 1)$};

\draw [-{Latex}] (semijoin_config) -- (join);

\draw [-{Latex}] (join_config) -- (ce);

\draw [-{Latex}] (join_config) -- (fe);

\draw [-{Latex}] (fe_config) -- (rps2);

\draw [-{Latex}] (rps2) -- (ce_config);

\draw [-{Latex}] (ce_config) -- (rps1);

\draw [-{Latex}] (rps1) -- (fe_config);


\node[below right = 0.35cm and 0.25cm of semijoin_config, dashed] (ie) [ms_retenode] {$LNS^{c \rightarrow i}$};
\node[below = 0.0cm of ie]  (ie_config) [reteconfig] {$(m_{1.3}, \infty)$};

\node[below = 0.35cm of ie_config, dashed] (rps3) [ms_retenode] {$RPS^{\infty}_c$};
\node[below = 0.0cm of rps3]  (rps3_config) [reteconfig] {$(m_{1.2.1}, \infty)$};

\draw [-{Latex}] (semijoin_config) -- (ie);

\draw [-{Latex}] (ie_config) -- (rps3);

\draw [-{Latex}] (rps3) -- (join_config);

\end{tikzpicture}
	\caption{Localized RETE net for the example query in Figure~\ref{fig:example_query_and_rete_ngc} with a consistent configuration for the host graph in Figure~\ref{fig:example_host_graph_ngc}} \label{fig:localized_rete_net_ngc}
\end{figure}

\begin{figure}[t]

\begin{tikzpicture}

\node (r1) [vertex] {\textit{r\textsubscript{1}:Proj}};


\node[below left = 1.25cm and 1.25cm of r1, line width = 4pt] (p1) [vertex] {\textbf{\textit{p\textsubscript{1}:Pkg}}};

\node[below = 1.25cm of p1] (c1) [vertex] {\textit{c\textsubscript{1}:Class}};

\node[below = 1.25cm of c1] (f1) [vertex] {\textit{f\textsubscript{1}:Field}};

\node[left = 1.75cm of c1] (i1) [vertex] {\textit{i\textsubscript{1}:Intf}};

\draw [-{Latex}] (r1) -- (p1) node [midway, right] {\textit{pe}};

\draw [-{Latex}] (p1) -- (c1) node [midway, right] {\textit{ce}};

\draw [-{Latex}] (p1) -- (c1) node [midway, right] {\textit{ce}};

\draw [-{Latex}] (c1) -- (f1) node [midway, right] {\textit{fe}};

\draw [-{Latex}] (c1) -- (i1) node [midway, above] {\textit{ie}};


\node[below right = 1.25cm and 1.25cm of r1] (p2) [vertex] {\textit{p\textsubscript{2}:Pkg}};

\node[below = 1.25cm of p2] (c2) [vertex] {\textit{c\textsubscript{2}:Class}};

\node[below = 1.25cm of c2] (f2) [vertex] {\textit{f\textsubscript{2}:Field}};

\draw [-{Latex}] (r1) -- (p2) node [midway, right] {\textit{pe}};

\draw [-{Latex}] (p2) -- (c2) node [midway, right] {\textit{ce}};

\draw [-{Latex}] (c2) -- (f2) node [midway, right] {\textit{fe}};


\node[above left = -1.8cm and -2.75cm of f1, rectangle, draw = black, dashed, minimum width = 4.5cm, minimum height = 9.25cm, text depth = 9cm, text width = 4.25cm] (m1) {\textit{m\textsubscript{1}}};

\node[above left = -1.8cm and -2.25cm of c1, rectangle, draw = orange, dashed, minimum width = 3.75cm, minimum height = 5.75cm, text depth = 5.5cm, text width = 3.5cm] (m11) {\textit{m\textsubscript{1.1}}};

\node[above left = -1.55cm and -2cm of p1, rectangle, draw = orange, dashed, minimum width = 3cm, minimum height = 2cm, text depth = 1.75cm, text width = 2.75cm] (m111) {\textit{m\textsubscript{1.1.1}}};

\node[above left = -1.55cm and -2.5cm of f1, rectangle, draw = blue, dashed, minimum width = 3.75cm, minimum height = 5.5cm, text depth = 5.25cm, text width = 3.5cm] (m12) {\textit{m\textsubscript{1.2}}};

\node[above left = -1.55cm and -2cm of c1, rectangle, draw = blue, dashed, minimum width = 3cm, minimum height = 2cm, text depth = 1.75cm, text width = 2.75cm] (m121) {\textit{m\textsubscript{1.2.1}}};

\node[above left = -1.8cm and -2.75cm of f2, rectangle, draw = black, dashed, minimum width = 4.5cm, minimum height = 9.25cm, text depth = 9cm, text width = 4.25cm, align = right] (m2) {\textit{m\textsubscript{2}}};

\node[above left = -1.65cm and -2.125cm of c1, rectangle, draw = darkgreen, dashed, minimum width = 6cm, minimum height = 2.75cm, text depth = 2.5cm, text width = 5.75cm] (m13) {\textit{m\textsubscript{1.3}}};

\end{tikzpicture}
	\caption{Example host graph with relevant subgraph marked in bold} \label{fig:example_host_graph_ngc}
\end{figure}

Formally, the following theorem states the correctness of localization via $localize^\Psi$ with respect to Definition~\ref{def:correctness_subgraphs_ngcs}:

\begin{thm}[Consistent configurations for RETE nets localized via $localize^\Psi$ yield correct query results under the relevant subgraph] \label{the:localized_rete_ngcs_correctness}
Let $H$ be a graph, $H_p \subseteq H$, and $(Q, \psi)$ an extended graph query. Furthermore, let $\mathcal{C}^\Phi$ be a consistent configuration for the localized RETE net $(N^\Phi, p^\Phi) = localize^\Psi(Q, \psi)$. The set of matches given by the stripped result set $\resultsstripped{p^\Phi}{\mathcal{C}^\Phi}$ is then correct under $H_p$.
\end{thm}

\begin{proof} (Idea)
It can be shown via induction over $\psi$ that the inserted request projection structures trigger the computation of the required intermediate results for nested queries of $(Q, \psi)$.
\end{proof}

An execution order for a localized RETE net $(N^\Phi, p^\Phi) = localize^{\Psi}(Q, \psi)$ can then be computed via the recursive procedure $order^{\Psi}$. This procedure works as follows:

\begin{itemize}
\item For a nested graph condition of the form $\psi = true$, the result of $order^\Psi$ is simply given by $order^\Psi(N^\Phi) = order(N^\Phi)$.

\item For a nested graph condition of the form $\psi = \exists (a: Q \rightarrow Q', \psi')$, the result of $order^\Psi$ is given by $order^\Psi(N^\Phi) = order(N^\Phi_Q) \circ toposort(RPS^{\infty}_l)^{-1} \circ order^\Psi(N^\Phi_{(Q', \psi')}) \circ [\ltimes]^\Phi$.

\item For a nested graph condition of the form $\psi = \neg\psi'$, the result of $order^\Psi$ is given by $order^\Psi(N^\Phi) = order(N^\Phi_Q) \circ toposort(RPS^{\infty}_l)^{-1} \circ order^\Psi(N^\Phi_{(Q, \psi')}) \circ [\rhd]^\Phi$.

\item Finally, for a nested graph condition of the form $\psi = \psi_1 \wedge \psi_2$, the result of $order^\Psi$ is given by $order^\Psi(N^\Phi) = order(N^\Phi_Q) \circ toposort(RPS^{\infty}_1)^{-1} \circ order^\Psi(N^\Phi_{(Q, \psi_1)}) \circ [\ltimes]^\Phi_1 \circ toposort(RPS^{\infty}_2)^{-1} \circ order^\Psi(N^\Phi_{(Q, \psi_2)}) \circ [\ltimes]^\Phi_2$.
\end{itemize}

This execution order then guarantees a consistent resulting configuration:

\begin{thm}[Execution of localized RETE nets via $order^{\Psi}$ yields consistent configurations] \label{the:localized_rete_ngcs_execution} 
Let $H$ be a graph, $H_p \subseteq H$, $(Q, \psi)$ a graph query, and $\mathcal{C}^\Phi_0$ an arbitrary starting configuration for the marking-sensitive RETE net $(N^\Phi, p^\Phi) = localize^{\Psi}(Q, \psi)$. Executing $(N^\Phi, p^\Phi)$ via $O = order^{\Psi}(N^\Phi)$ then yields a consistent configuration $\mathcal{C}^\Phi = execute(O, N^\Phi, H, H_p, \mathcal{C}^\Phi_0)$.
\end{thm}

\begin{proof} (Idea)
Follows because all cyclical structures in $N^\Phi$ are encapsulated in localized RETE subnets for plain graph queries, for which $order$ creates an admissible ordering.
\end{proof}

It thus follows that a RETE net localized via $localize^\Psi$ can be executed to produce a set of matches that is correct under a given subgraph:

\begin{cor}[Execution of RETE nets localized via $localize^\Psi$ yields correct query results under the relevant subgraph]
Let $H$ be a graph, $H_p \subseteq H$, and $(Q, \psi)$ a graph query. Furthermore, let $\mathcal{C}^\Phi_0$ be an arbitrary configuration for the marking-sensitive RETE net $(N^\Phi, p^\Phi) = localize^\Psi(Q, \psi)$ and $\mathcal{C}^\Phi = execute(order^\Psi(N^\Phi), N^\Phi, H, H_p, \mathcal{C}^\Phi_0)$. The set of matches from $Q$ into $H$ given by $\resultsstripped{p^\Phi}{\mathcal{C}^\Phi}$ is then correct under $H_p$.
\end{cor}

\begin{proof}
Follows directly from Theorems~\ref{the:localized_rete_ngcs_execution} and~\ref{the:localized_rete_ngcs_correctness}.
\end{proof}

As $localize^\Psi$ essentially mirrors the construction procedure for RETE nets for extended graph queries from \cite{barkowsky2023host}, each RETE net $(N^\Phi, p^\Phi) = localize^\Psi(Q, \psi)$ for some extended query $(Q, \psi)$ has a corresponding regular RETE net $(N, p)$ created via the procedure described in \cite{barkowsky2023host} that uses the same RETE net structure for $Q$ and all patterns in $\psi$. As for $localize$ and well-formed RETE nets for simple queries, we can therefore show that a consistent configuration for $(N^\Phi, p^\Phi)$ only exhibits a constant-factor increase in configuration size compared to $(N, p)$ in the worst case:

\begin{thm}[RETE net localization via $localize^\Psi$ introduces only a constant factor overhead on effective configuration size] \label{the:upper_bound_configuration_size_ngc}
Let $H$ be an edge-dominated graph, $H_p \subseteq H$, $(N, p)$ a RETE net created via the procedure described in \cite{barkowsky2023host} for the extended graph query $(Q, \psi)$, $\mathcal{C}$ a consistent configuration for $(N, p)$ for host graph $H$, and $\mathcal{C}^\Phi$ a consistent configuration for the marking-sensitive RETE net $(N^\Phi, p^\Phi) = localize^\Psi(Q, \psi)$ corresponding to $(N, p)$ for host graph $H$ and relevant subgraph $H_p$. It then holds that $\sum_{n^\Phi \in V^{N^\Phi}} \sum_{(m, \phi) \in \mathcal{C}^\Phi(n^\Phi)} |m| \leq 7 \cdot |\mathcal{C}|_e$.
\end{thm}

\begin{proof} (Idea)
Follows because each node in $N$ has only a constant number of associated nodes with result sets containing corresponding matches in $N^\Phi$.
\end{proof}

Consequently, localization via $localize^\Psi$ then also incurs at most a constant-factor overhead on memory consumption:

\begin{cor}[RETE net localization via $localize^\Psi$ introduces only a constant factor overhead on memory consumption]
Let $H$ be an edge-dominated graph, $H_p \subseteq H$, $(N, p)$ a RETE net created via the procedure described in \cite{barkowsky2023host} for the extended graph query $(Q, \psi)$, $\mathcal{C}$ a consistent configuration for $(N, p)$ for host graph $H$, and $\mathcal{C}^\Phi$ a consistent configuration for the marking-sensitive RETE net $(N^\Phi, p^\Phi) = localize^\Psi(Q, \psi)$ corresponding to $(N, p)$ for host graph $H$ and relevant subgraph $H_p$. Assuming that storing a match $m$ requires an amount of memory in $O(|m|)$ and storing an element from $\overline{\mathbb{N}}$ requires an amount of memory in $O(1)$, storing $\mathcal{C}^\Phi$ requires memory in $O(|\mathcal{C}|_e)$.
\end{cor}

\begin{proof}
Follows directly from Theorem~\ref{the:upper_bound_configuration_size_ngc} and the assumptions.
\end{proof}

Finally, by a similar argumentation as for $localize$, $localize^\Psi$ also only leads to an increase in computational complexity by factor $\frac{|Q|}{|Q_a|}$, with $|Q_a|$ the average size of matches in a consistent configuration for $(N, p)$:

\begin{thm}[Execution time overhead introduced by RETE net localization via $localize^\Psi$ depends on query graph size compared to average match size]
Let $H$ be an edge-dominated graph, $H_p \subseteq H$, $(N, p)$ a RETE net created via the procedure described in \cite{barkowsky2023host} for the extended graph query $(Q, \psi)$, $\mathcal{C}$ a consistent configuration for $(N, p)$ for host graph $H$, and $\mathcal{C}^\Phi_0$ the empty configuration for the marking-sensitive RETE net $(N^\Phi, p^\Phi) = localize^\Psi(Q, \psi)$ corresponding to $(N, p)$. Executing $(N^\Phi, p^\Phi)$ via $execute(order^\Psi(N^\Phi), N^\Phi, H, H_p, \mathcal{C}^\Phi_0)$ then takes $O(T \cdot |Q_{max}|)$ steps, with $T = \sum_{n \in V^N} |\mathcal{C}(n)|$ and $|Q_{max}|$ the size of the largest graph among $Q$ and graphs in $\psi$.
\end{thm}

\begin{proof} (Idea)
Follows from Theorem~\ref{the:complexity_time_localized} and the fact that all semi-joins, anti-joins, and request projection structures not encapsulated in RETE subnets for plain graph queries in $N^\Phi$ are only executed once.
\end{proof}

Interestingly, only a subset of the vertex input nodes in a localized RETE net for an extended graph query are required for the net to produce correct results with respect to Definition~\ref{def:correctness_subgraphs_ngcs}. Specifically, all vertex input nodes except those in the local navigation structures of the RETE net for the query pattern created in the topmost recursive call of $localize^\Psi$ are ultimately superfluous. By removing these superfluous vertex input nodes, localization via $localize^\Psi$ may then actually improve performance even in cases where the relevant subgraph coincides with the full host graph.

Essentially, the resulting RETE net no longer necessarily computes all matches for all subpatterns of a query's nested graph condition. Instead, match computation for these subpatterns is actually guided by requests for complementarity resulting from matches for the base pattern. Restrictive base patterns with few matches may then create a similar effect to localized execution, with only a subset of matches for subpatterns in nested graph conditions being computed in order to complement existing matches for the base pattern. This is similar to what can be achieved with local-search-based querying techniques but with the advantage of enabling incremental execution. As with localization, the effectiveness however depends on the exact query and host graph.

\subsection{Localized Detection of NGC Satisfaction Changes} \label{sec:localized_rete_satisfaction_changes}

While RETE nets constructed via $localize^\Psi$ can be used to locally and incrementally compute matches for an extended query that touch a relevant subgraph, the notion of completeness in Definition~\ref{def:correctness_subgraphs_ngcs} no longer covers the use case of correctly tracking the full effect of local changes on a global query result. This is due to the fact that for extended queries, local changes may result in changes to the satisfaction of a nested graph condition for matches that do not touch the relevant subgraph, effectively manipulating query results that do not fall under Definition~\ref{def:correctness_subgraphs_ngcs}.

To track the effect of local host graph changes on the global set of results for an extended graph query, we thus also have to maintain the set of matches for which a modification of the relevant subgraph may change the nested graph condition's satisfaction. We ideally also want to precompute this set of matches before the host graph is modified to enable incremental maintenance. We thus have no prior knowledge regarding the modifications we may have to consider except for the fact that they concern the relevant subgraph.

This proves problematic when considering a case where the relevant subgraph constitutes the entire host graph as an example. In this scenario, checking whether an arbitrary modification to the relevant subgraph may change the satisfaction of a nested graph condition for some context match corresponds to checking satisfiability of the nested graph condition.

Notably however, nested graph conditions have been shown to be equivalent to first-order logic over graphs \cite{habel2009correctness,rensink2004representing}. The satisfiability problem for first-order logic being undecidable thus makes a direct computation of matches for which NGC satisfaction somehow depends on the relevant subgraph impossible in the general case.

In order to still enable a localized tracking of changes to the results of an extended graph query, we therefore employ the following overapproximation for the set of matches where NGC satisfaction may change as a result of local modifications:

\begin{defi}[Subgraph Satisfaction Dependence] \label{def:subgraph_satisfaction_dependence} 
We say that a match $m : Q \rightarrow H$ for an extended graph query $(Q, \psi)$ into a host graph $H$ with relevant subgraph $H_p \subseteq H$ is \emph{subgraph satisfaction dependent} if

\begin{itemize}
\item $\psi = \exists (a: Q \rightarrow Q', \psi')$ and there exists a graph morphism $m' : Q' \rightarrow H$ such that $m = m' \circ a$ and $m'(Q') \cap H_p \neq \emptyset \text{ or } m' \text{ is subgraph satisfaction dependent}$
\item $\psi = \neg\psi'$ and $m$ is subgraph satisfaction dependent for $(Q, \psi')$
\item $\psi = \psi_1 \wedge \psi_2$ and $m$ is subgraph satisfaction dependent for $(Q, \psi_1)$ or $(Q, \psi_2)$
\end{itemize}

If $\psi = \texttt{true}$, $m$ is never subgraph satisfaction dependent.
\end{defi}

We furthermore define a subgraph-restricted, that is, local graph modification based on the basic notion of graph modifications introduced in Section~\ref{sec:preliminaries_graphs}. In a subgraph-restricted graph modification, changes are limited to relevant subgraphs of source and target graph and relevant subgraphs are mapped onto each other:

\begin{defi}[Subgraph-restricted Graph Modification] \label{def:subgraph_restricted_graph_modification} 
We say that a graph modification $H \xleftarrow{f} K \xrightarrow{g} H'$ from a graph $H$ with relevant subgraph $H_p$ into a modified graph $H'$ with relevant subgraph $H'_p$ is \emph{subgraph-restricted} if $iso_s = g \circ f^{-1}|_{H_s} = (f \circ g^{-1}|_{H'_s})^{-1}$ is an isomorphism between $H_s = \{V^{H} \setminus (V^{H_p} \setminus V^{H_c}), E^{H} \setminus E^{H_p}, s^{H}|_{E^{H} \setminus E^{H_p}}, t^{H}|_{E^{H} \setminus E^{H_p}}\}$ and $H'_s = \{V^{H'} \setminus (V^{H'_p} \setminus V^{H'_c}), E^{H'} \setminus E^{H'_p}, s^{H'}|_{E^{H'} \setminus E^{H'_p}}, t^{H'}|_{E^{H'} \setminus E^{H'_p}}\}$, where $V^{H_c} = \{v_c \in V^{H_p} \mid e \in E^{H} \setminus E^{H_p} : s^H(e) = v_c \vee t^{H}(e) = v_c\}$ and with $V^{H'_c}$ defined analogously. Furthermore, it must hold that $(g \circ f^{-1})(H_p) \subseteq H'_p \wedge (f \circ g^{-1})(H'_p) \subseteq H_p$.
\end{defi}

\begin{figure}
\begin{center}

\begin{tikzpicture}

\node (h) [graph] {$H$};

\node [right=3cm of h] (k) [graph] {$K$};

\node [right=3cm of k] (h_prime) [graph] {$H'$};

\node[below right = 1.25cm and 0.5cm of h] (h_p) [graph] {$H_p$};

\node[below left = 1.25cm and 0.5cm of h] (h_s) [graph] {$H_s$};

\node[below left = 1.25cm and 0.5cm of h_prime] (h_p_prime) [graph] {$H'_p$};

\node[below right = 1.25cm and 0.5cm of h_prime] (h_s_prime) [graph] {$H'_s$};

\draw [arrows={Hooks[left]-Stealth}] (k) -- (h) node[midway,auto,swap] {f};

\draw [arrows={Hooks[right]-Stealth}] (k) -- (h_prime) node[midway,auto] {g};

\draw (h) -- (h_s) node[midway,inner sep=0] (j1) {};

\draw (h) -- (h_p) node[midway,inner sep=0] (j2) {};

\draw[bend right] (j1) to node [midway,auto] (union) {$\cup$} (j2);

\draw (h_prime) -- (h_p_prime) node[midway,inner sep=0] (j1_prime) {};

\draw (h_prime) -- (h_s_prime) node[midway,inner sep=0] (j2_prime) {};

\draw [bend right] (j1_prime) to node [midway,auto] (union_prime) {$\cup$} (j2_prime);

\node[below right = 1.25cm and 0.5cm of h_s] (h_c) [graph] {$H_c$};

\draw (h_s) -- (h_c) node[midway,inner sep=0] (j3) {};

\draw (h_p) -- (h_c) node[midway,inner sep=0] (j4) {};

\draw[bend right] (j3) to node [midway,auto] (intersection) {$\cap$} (j4);

\node[below left = 1.25cm and 0.5cm of h_s_prime] (h_c_prime) [graph] {$H'_c$};

\draw (h_s_prime) -- (h_c_prime) node[midway,inner sep=0] (j3_prime) {};

\draw (h_p_prime) -- (h_c_prime) node[midway,inner sep=0] (j4_prime) {};

\draw[bend left] (j3_prime) to node [midway,auto,swap] (intersection_prime) {$\cap$} (j4_prime);


\draw [bend left,arrows={-Stealth[left]}] (h_p) to node[midway,auto] {$(g \circ f^{-1})|_{H_p}$} (h_p_prime);

\draw [bend left,arrows={-Stealth[left]}] (h_p_prime) to node[midway,auto,swap] {$(f \circ g^{-1})|_{H'_p}$} (h_p);

\draw [bend right,out=-90,in=-90,arrows={Stealth-Stealth}] (h_s) to node[midway,auto,swap] {$iso_s = g \circ f^{-1}|_{H_s} = (f \circ g^{-1}|_{H'_s})^{-1}$} (h_s_prime);

\end{tikzpicture}
\end{center}
\caption{Relationship between graphs in Definition~\ref{def:subgraph_restricted_graph_modification}} \label{fig:subgraph_restricted_graph_modification}
\end{figure}

The relationship between the graphs in the above definition is visualized in the diagram in Figure~\ref{fig:subgraph_restricted_graph_modification}. $H$ is the source graph of the modification, which is modified into $H'$ via the intermediate graph $K$. $H_p$ and $H'_p$ are the relevant subgraphs of $H$ respectively $H'$, which have to be mapped onto each other via $f \circ g^{-1}$ and $g \circ f^{-1}$. $H_s$ and $H'_s$ constitute the remainder of $H$ and $H'$ and are required to be isomorphic, which means that the subgraph-restricted graph modification cannot make changes outside the relevant subgraph. Note that $H_s$ and $H'_s$ may share nodes with $H_p$ respectively $H'_p$ if they are adjacent to edges in $H \setminus H_p$ respectively $H' \setminus H'_p$ to avoid dangling edges. $H_c$ and $H'_c$ then are the discrete graphs consisting of these shared nodes $V^{H_c}$ respectively $V^{H'_c}$ and no edges.

The subgraph satisfaction dependent matches for an extended graph query are then a superset of all matches for which the satisfaction of the query's nested graph condition may change as a result of a subgraph-restricted graph modification. Essentially, a change in nested graph condition satisfaction requires a match to be subgraph satisfaction dependent before or after the host graph modification:

\begin{thm}[Subgraph-restricted graph modifications can only change NGC satisfaction for subgraph satisfaction dependent matches] \label{the:subgraph_satisfaction_dependence_necessary} 
Let $(Q, \psi)$ be an extended graph query and $H \xleftarrow{f} K \xrightarrow{g} H'$ a subgraph-restricted graph modification from host graph $H$ with relevant subgraph $H_p$ into the modified host graph $H'$ with relevant subgraph $H'_p$. It then holds for any graph morphisms $m_K : Q \rightarrow K$ that $(f \circ m_K \models \psi \wedge g \circ m_K \not\models \psi) \vee (f \circ m_K \not\models \psi \wedge g \circ m_K \models \psi) \Rightarrow f \circ m_K \text{ is subgraph satisfaction dependent } \text{ or } g \circ m_K \text{ is subgraph satisfaction dependent}$.
\end{thm}

\begin{proof} (Idea)
Can be shown via induction over $\psi$, exploiting the fact that host graph modifications can only directly change the satisfaction of existential subconditions.
\end{proof}

A RETE net for locally computing all subgraph satisfaction dependent matches for an extended graph query $(Q, \psi)$ can then be constructed recursively via the procedure $localize^{sat}$, which is visualized in Figure~\ref{fig:localize_sat} and works as follows:

\begin{itemize}
\item For a nested graph condition of the form $\psi = true$, $(N^{sat}, p^{sat}) = localize^{sat}(Q, \psi)$ consists of a single dummy RETE node $p^{sat} = [\emptyset]^\Phi$, whose result set is always empty.

\item For a nested graph condition of the form $\psi = \exists (a: Q \rightarrow Q', \psi')$, $(N^{sat}, p^{sat}) = localize^{sat}(Q, \psi)$ consists of the localized RETE net for the plain pattern $Q$, $(N^\Phi_Q, p^\Phi_Q) = localize(Q)$, the localized RETE net for the plain pattern $Q'$, $(N^\Phi_{Q'}, p^\Phi_{Q'}) = localize(Q')$, the RETE net $(N^{sat}_{(Q', \psi')}, p^{sat}_{(Q', \psi')}) = localize^{sat}(Q', \psi')$, a marking-sensitive union node $[\cup]^\Phi$ with dependencies $p^\Phi_{Q'}$ and $p^{sat}_{(Q', \psi')}$, a request projection structure given by $RPS^{\infty}_r = RPS^{\infty}([\cup]^\Phi, N^\Phi_Q)$, and a marking-sensitive semi-join node $p^{sat} = [\ltimes]^\Phi$ along $a$ with left dependency $p^\Phi_Q$ and right dependency $[\cup]^\Phi$.

\item For a nested graph condition of the form $\psi = \neg\psi'$, $(N^{sat}, p^{sat}) = localize^{sat}(Q, \psi) = localize^{sat}(Q, \psi')$.

\item For a nested graph condition of the form $\psi = \psi_1 \wedge \psi_2$, $(N^{sat}, p^{sat}) = localize^{sat}(Q, \psi)$ consists of the RETE nets $(N^{sat}_{(Q, \psi_1)}, p^{sat}_{(Q, \psi_1)}) = localize^{sat}(Q, \psi_1)$ and $(N^{sat}_{(Q, \psi_2)}, p^{sat}_{(Q, \psi_2)}) = localize^{sat}(Q, \psi_2)$ and a marking-sensitive union node $p^{sat} = [\cup]^\Phi$ with dependencies $p^{sat}_{(Q, \psi_1)}$ and $p^{sat}_{(Q, \psi_2)}$.
\end{itemize}

\begin{figure} [t]
	\centering
	\begin{subfigure}[b]{0.36\textwidth}
		\begin{center}

\begin{tikzpicture}

\node[rectangle,
	minimum width = 4.15cm, 
	minimum height = 3.25cm] (r) at (0,0) {};

\node [below = -3.25cm of r] [ms_retenode] {$\emptyset$};

\end{tikzpicture}
		\end{center}
	\end{subfigure}\hspace{4pt}\vline\hspace{4pt}
	\begin{subfigure}[b]{0.55\textwidth}
		\begin{center}

\begin{tikzpicture}

\node[rectangle,
	minimum width = 6.85cm, 
	minimum height = 3.25cm] (r) at (0,0) {};

\node [below left = -3.25cm and -3.65cm of r] (join) [ms_retenode] {$\ltimes$};

\node [below left = 0.35cm and 0.25cm of join, dashed] (q) [ms_retenode] {$N^{\Phi}_{Q}$};

\node [below right = 0.35cm and 0.25cm of join] (union) [ms_retenode] {$\cup$};

\node [below = 0.35cm of q, dashed] (rpsr) [ms_retenode] {$RPS^{\infty}_r$};

\node [below left = 0.35cm and -0.5cm of union, dashed] (qp) [ms_retenode] {$N^{\Phi}_{Q'}$};

\node [below right = 0.35cm and -0.5cm of union, dashed] (qp_sat) [ms_retenode] {$N^{sat}_{(Q', \psi')}$};

\draw [-{Latex}] (join) -- (q);

\draw [-{Latex}] (join) -- (union);

\draw [-{Latex}] (union) -- (qp);

\draw [-{Latex}] (union) -- (qp_sat);

\draw [-{Latex}] (rpsr) -- (union);

\draw [-{Latex}] (q) -- (rpsr);

\end{tikzpicture}
		\end{center}
	\end{subfigure}

	\hrule

	\begin{subfigure}[b]{0.36\textwidth}
		\begin{center}

\begin{tikzpicture}

\node[rectangle,
	minimum width = 4.15cm, 
	minimum height = 2.2cm] (r) at (0,0) {};

\node [below = -1.95cm of r, dashed] [ms_retenode] {$N^{sat}_{(Q, \psi')}$};

\end{tikzpicture}
		\end{center}
	\end{subfigure}\hspace{4pt}\vline\hspace{4pt}
	\begin{subfigure}[b]{0.55\textwidth}
		\begin{center}

\begin{tikzpicture}

\node[rectangle,
	minimum width = 6.85cm, 
	minimum height = 2.2cm] (r) at (0,0) {};

\node [below = -1.95cm of r] (union) [ms_retenode] {$\cup$};

\node [below left = 0.35cm and 0.25cm of union, dashed] (q_psi1) [ms_retenode] {$N^{sat}_{(Q, \psi_1)}$};

\node [below right = 0.35cm and 0.25cm of union, dashed] (q_psi2) [ms_retenode] {$N^{sat}_{(Q, \psi_2)}$};

\draw [-{Latex}] (union) -- (q_psi1);

\draw [-{Latex}] (union) -- (q_psi2);

\end{tikzpicture}
		\end{center}
	\end{subfigure}
	\caption{Results of applying $localize^{sat}$ to an extended graph query $(Q, \psi)$ with nested graph conditions of the form $\psi = \texttt{true}$ (top left), $\psi = \exists(a : Q \rightarrow Q', \psi')$ (top right), $\psi = \neg\psi'$ (bottom left), and $\psi = \psi_1 \wedge \psi_2$ (bottom right)} \label{fig:localize_sat}
\end{figure}

A localized RETE net for computing subgraph-satisfaction dependent matches for the example query in Figure~\ref{fig:example_query_and_rete_ngc} is displayed in Figure~\ref{fig:localized_rete_net_sat} alongside a consistent configuration for the host graph in Figure~\ref{fig:example_host_graph_sat}. It consists of the localized RETE net in Figure~\ref{fig:localized_rete_net} and the structure created for handling the nested graph condition according to Figure~\ref{fig:localize_sat}.

Notably, there exist no matches for the base pattern $Q$ that touch the relevant subgraph, which only consists of the interface node $i_1$. Instead, the computation of matches for $Q$ is controlled via the request projection structure that feeds matches for the right dependency of the semi-join node into the localized RETE net for $Q$.

Compared to Figure~\ref{fig:localized_rete_net_ngc}, propagation of matches is essentially reversed: First, the local navigation structure $LNS^{c \rightarrow i}$ extracts all matches for the pattern $Q'$ that touch the relevant subgraph. The request projection structure connected to the union node then causes the local navigation structure $LNS^{c \rightarrow f}$ to extract the match $m_{1.2}$. The leftmost request projection structure finally ensures that the required complementary match $m_{1.1}$ is extracted by the local navigation structure $LNS^{p \rightarrow c}$.

\begin{figure}[t]

\begin{tikzpicture}

\node (semijoin) [ms_retenode] {$\ltimes$};
\node[below = 0.0cm of semijoin]  (semijoin_config) [reteconfig] {$(m_{1}, \infty)$};


\node[below left = 0.35cm and 2.75cm of semijoin_config] (join) [ms_retenode] {$\bowtie$};
\node[below = 0.0cm of join]  (join_config) [reteconfig] {$(m_{1}, \infty)$};

\node[below left = 0.35cm and 0.25cm of join_config, dashed] (ce) [ms_retenode] {$LNS^{p \rightarrow c}$};
\node[below = 0.0cm of ce]  (ce_config) [reteconfig] {$(m_{1.1}, 1)$};

\node[below right = 0.35cm and 0.25cm of join_config, dashed] (fe) [ms_retenode, minimum width = 1.8cm] {$LNS^{c \rightarrow f}$};
\node[below = 0.0cm of fe]  (fe_config) [reteconfig, minimum width = 1.8cm] {$(m_{1.2}, \infty)$};

\node[below = 0.35cm of ce_config, dashed] (rps1) [ms_retenode] {$RPS_c$};
\node[below = 0.0cm of rps1]  (rps1_config) [reteconfig] {$(m_{1.2.1}, 1)$};

\node[below left = 0.35cm and -0.5cm of fe_config, dashed] (rps2) [ms_retenode] {$RPS_c$};
\node[below = 0.0cm of rps2]  (rps2_config) [reteconfig] {$$};

\node[below right = 0.35cm and -0.5cm of fe_config, dashed] (rps3) [ms_retenode] {$RPS^{\infty}_c$};
\node[below = 0.0cm of rps3]  (rps3_config) [reteconfig] {$(m_{1.2.1}, \infty)$};

\draw [-{Latex}] (semijoin_config) -- (join);

\draw [-{Latex}] (join_config) -- (ce);

\draw [-{Latex}] (join_config) -- (fe);

\draw [-{Latex}] (fe_config) -- (rps2);

\draw [-{Latex}] (rps2) -- (ce_config);

\draw [-{Latex}] (ce_config) -- (rps1);

\draw [-{Latex}] (rps1) -- (fe_config);


\node[below right = 0.35cm and 2.75cm of semijoin_config] (union) [ms_retenode] {$\cup$};
\node[below = 0.0cm of union]  (union_config) [reteconfig] {$(m_{1.3}, \infty)$};

\node[below left = 0.35cm and -0.5cm of union_config, dashed] (ie) [ms_retenode] {$LNS^{c \rightarrow i}$};
\node[below = 0.0cm of ie]  (ie_config) [reteconfig] {$(m_{1.3}, \infty)$};

\node[below right = 0.35cm and -0.5cm of union_config] (empty) [ms_retenode] {$\emptyset$};
\node[below = 0.0cm of empty]  (empty_config) [reteconfig] {$$};

\draw [-{Latex}] (semijoin_config) -- (union);

\draw [-{Latex}] (union_config) -- (ie);

\draw [-{Latex}] (union_config) -- (empty);

\draw [-{Latex}] (fe_config) -- (rps3);

\draw [-{Latex}] (rps3) -- (-1.2, -1.85) -- (union);

\end{tikzpicture}
	\caption{Localized RETE net for computing subgraph-satisfaction dependent matches for the example query in Figure~\ref{fig:example_query_and_rete_ngc} with a consistent configuration for the host graph in Figure~\ref{fig:example_host_graph_sat}} \label{fig:localized_rete_net_sat}
\end{figure}

\begin{figure}[t]

\begin{tikzpicture}

\node (r1) [vertex] {\textit{r\textsubscript{1}:Proj}};


\node[below left = 1.25cm and 1.25cm of r1] (p1) [vertex] {\textit{p\textsubscript{1}:Pkg}};

\node[below = 1.25cm of p1] (c1) [vertex] {\textit{c\textsubscript{1}:Class}};

\node[below = 1.25cm of c1] (f1) [vertex] {\textit{f\textsubscript{1}:Field}};

\node[left = 1.75cm of c1, line width = 4pt] (i1) [vertex] {\textbf{\textit{i\textsubscript{1}:Intf}}};

\draw [-{Latex}] (r1) -- (p1) node [midway, right] {\textit{pe}};

\draw [-{Latex}] (p1) -- (c1) node [midway, right] {\textit{ce}};

\draw [-{Latex}] (p1) -- (c1) node [midway, right] {\textit{ce}};

\draw [-{Latex}] (c1) -- (f1) node [midway, right] {\textit{fe}};

\draw [-{Latex}] (c1) -- (i1) node [midway, above] {\textit{ie}};


\node[below right = 1.25cm and 1.25cm of r1] (p2) [vertex] {\textit{p\textsubscript{2}:Pkg}};

\node[below = 1.25cm of p2] (c2) [vertex] {\textit{c\textsubscript{2}:Class}};

\node[below = 1.25cm of c2] (f2) [vertex] {\textit{f\textsubscript{2}:Field}};

\draw [-{Latex}] (r1) -- (p2) node [midway, right] {\textit{pe}};

\draw [-{Latex}] (p2) -- (c2) node [midway, right] {\textit{ce}};

\draw [-{Latex}] (c2) -- (f2) node [midway, right] {\textit{fe}};


\node[above left = -1.8cm and -2.75cm of f1, rectangle, draw = black, dashed, minimum width = 4.5cm, minimum height = 9.25cm, text depth = 9cm, text width = 4.25cm] (m1) {\textit{m\textsubscript{1}}};

\node[above left = -1.8cm and -2.25cm of c1, rectangle, draw = orange, dashed, minimum width = 3.75cm, minimum height = 5.75cm, text depth = 5.5cm, text width = 3.5cm] (m11) {\textit{m\textsubscript{1.1}}};

\node[above left = -1.55cm and -2cm of p1, rectangle, draw = orange, dashed, minimum width = 3cm, minimum height = 2cm, text depth = 1.75cm, text width = 2.75cm] (m111) {\textit{m\textsubscript{1.1.1}}};

\node[above left = -1.55cm and -2.5cm of f1, rectangle, draw = blue, dashed, minimum width = 3.75cm, minimum height = 5.5cm, text depth = 5.25cm, text width = 3.5cm] (m12) {\textit{m\textsubscript{1.2}}};

\node[above left = -1.55cm and -2cm of c1, rectangle, draw = blue, dashed, minimum width = 3cm, minimum height = 2cm, text depth = 1.75cm, text width = 2.75cm] (m121) {\textit{m\textsubscript{1.2.1}}};

\node[above left = -1.8cm and -2.75cm of f2, rectangle, draw = black, dashed, minimum width = 4.5cm, minimum height = 9.25cm, text depth = 9cm, text width = 4.25cm, align = right] (m2) {\textit{m\textsubscript{2}}};

\node[above left = -1.65cm and -2.125cm of c1, rectangle, draw = darkgreen, dashed, minimum width = 6cm, minimum height = 2.75cm, text depth = 2.5cm, text width = 5.75cm] (m13) {\textit{m\textsubscript{1.3}}};

\end{tikzpicture}
	\caption{Example host graph with relevant subgraph marked in bold} \label{fig:example_host_graph_sat}
\end{figure}

A consistent configuration for the resulting RETE net then indeed guarantees a complete query result with respect to Definition~\ref{def:subgraph_satisfaction_dependence}:

\begin{thm}[Consistent configurations for RETE nets localized via $localize^{sat}$ yield all subgraph satisfaction dependent matches] \label{the:completeness_satisfaction_dependency} 
Let $H$ be a graph, $H_p \subseteq H$, $(Q, \psi)$ an extended graph query, and $\mathcal{C}^\Phi$ a consistent configuration for the RETE net $(N^{sat}, p^{sat}) = localize^{sat}(Q, \psi)$. It then holds that $\forall m \in \allmatches{Q}{H} : m \text{ is subgraph satisfaction dependent } \Rightarrow m \in \resultsstripped{p^{sat}}{\mathcal{C}^\Phi}$.
\end{thm}

\begin{proof} (Idea)
Can be shown via induction over $\psi$, exploiting the close alignment of $localize^{sat}$ with Definition~\ref{def:subgraph_satisfaction_dependence}.
\end{proof}

Given an extended graph query $(Q, \psi)$ and the RETE net $(N^{sat}, p^{sat}) = localize^{sat}(Q, \psi)$, an execution order for $(N^{sat}, p^{sat})$ can be constructed in a straightforward manner via the procedure $order^{sat}$:

\begin{itemize}
\item For a nested graph condition of the form $\psi = \texttt{true}$, the result of $order^{sat}$ is given by $order^{sat}(N^{sat}) = [\emptyset]^\Phi$.

\item For a nested graph condition of the form $\psi = \exists (a: Q \rightarrow Q', \psi')$, the result of $order^{sat}$ is given by $order^{sat}(N^{sat}) = order(N^\Phi_{Q'}) \circ order^{sat}(N^{sat}_{(Q', \psi')}) \circ [\cup]^\Phi \circ toposort(RPS^{\infty}_r)^{-1} \circ order(N^\Phi_Q) \circ [\ltimes]^\Phi$.

\item For a nested graph condition of the form $\psi = \neg\psi'$, $(N^{sat}, p^{sat}) = localize^{sat}(Q, \psi')$ and thus the ordering via $order^{sat}(N^{sat})$ is covered by one of the other cases.

\item For a nested graph condition of the form $\psi = \psi_1 \wedge \psi_2$, the result of $order^{sat}$ is given by $order^{sat}(N^{sat}, p^{sat}) = order^{sat}(N^{sat}_{(Q, \psi_1)}) \circ order^{sat}(N^{sat}_{(Q, \psi_2)}) \circ [\cup]^\Phi$.
\end{itemize}

This execution order then guarantees consistency of the resulting configuration:

\begin{thm}[Execution of localized RETE nets via $order^{sat}$ yields consistent configurations] \label{the:satisfaction_rete_execution} 
Let $H$ be a graph, $H_p \subseteq H$, $(Q, \psi)$ a graph query, and $\mathcal{C}^\Phi_0$ an arbitrary starting configuration for the marking-sensitive RETE net $(N^{sat}, p^{sat}) = localize^{sat}(Q, \psi)$. Executing $(N^{sat}, p^{sat})$ via $O = order^{sat}(N^{sat})$ then yields a consistent configuration $\mathcal{C}^\Phi = execute(O, N^{sat}, H, H_p, \mathcal{C}^\Phi_0)$.
\end{thm}

\begin{proof} (Idea)
Follows because all cyclical structures in $N^{sat}$ are encapsulated in localized RETE subnets for plain graph queries, for which $order$ creates an admissible ordering.
\end{proof}

We hence again get a corollary stating the correctness of RETE net execution via $localize^{sat}$:

\begin{cor}[Execution of RETE nets localized via $localize^{sat}$ yields all subgraph satisfaction dependent matches]
Let $H$ be a graph, $H_p \subseteq H$, and $(Q, \psi)$ a graph query. Furthermore, let $\mathcal{C}^\Phi_0$ be an arbitrary starting configuration for the marking-sensitive RETE net $(N^{sat}, p^{sat}) = localize^{sat}(Q, \psi)$ and $\mathcal{C}^\Phi = execute(order^{sat}(N^{sat}), N^{sat}, H, H_p, \mathcal{C}^\Phi_0)$. It then holds that $\forall m \in \allmatches{Q}{H} : m \text{ is subgraph satisfaction dependent } \Rightarrow m \in \resultsstripped{p^{sat}}{\mathcal{C}^\Phi}$.
\end{cor}

\begin{proof}
Follows directly from Theorems~\ref{the:completeness_satisfaction_dependency} and~\ref{the:satisfaction_rete_execution}.
\end{proof}

Technically, it can be shown that a RETE net for an extended graph query $(Q, \psi)$ constructed via $(N^{sat}, p^{sat}) = localize^{sat}(Q, \psi)$ has similar performance characteristics as a RETE net constructed via $localize^\Psi$ compared to a regular RETE net $(N, p)$ for $(Q, \psi)$ constructed via the basic procedure described in \cite{barkowsky2023host}. Specifically, the size of a consistent configuration for $(N^{sat}, p^{sat})$ is only increased by a constant factor compared to a consistent configuration for $(N, p)$, resulting in only a constant-factor increase in memory consumption and a minor increase in runtime complexity.\footnote{See Theorem~\ref{the:upper_bound_configuration_size_satisfaction_appendix}, Corollary~\ref{cor:memory_consumption_satisfaction_rete_appendix}, and Theorem~\ref{the:execution_time_satisfaction_rete_appendix} in Appendix~\ref{app:technical_details}.}

However, it is important to note that $(N, p)$ and $(N^{sat}, p^{sat})$ potentially compute very different sets of matches. By Definition~\ref{def:subgraph_satisfaction_dependence}, RETE nets constructed via $localize^{sat}$ may have to compute matches for $Q$ that do not actually satisfy $\psi$ if $\psi$ involves negations or conjunctions, whereas such matches have to be excluded from the results of $(N, p)$.

Relating the performance of $(N^{sat}, p^{sat})$ to the performance of $(N, p)$ therefore relies on redundancy in the basic construction in \cite{barkowsky2023host}. Many optimizations to $(N, p)$, such as the elimination of redundancy or the downward propagation of nested graph conditions in the RETE net described in \cite{barkowsky2023host}, are applicable to localized RETE nets constructed via $localize^\Psi$ and likely yield similar benefits. In contrast, due to the different construction procedure and desired computation results, such optimizations can have a very different effect on $(N^{sat}, p^{sat})$ or may simply not be applicable at all. Consequently, executing $(N^{sat}, p^{sat})$ as part of any execution of $(Q, \psi)$ may create substantial overhead on both memory consumption and execution time in unfavorable cases if further optimizations are involved.

Still, using the $localize^{sat}$ procedure, we can now construct a RETE net that allows the localized computation of changes to the results of an extended graph query $(Q, \psi)$ caused by a subgraph-restricted graph modification $H \xleftarrow{f} K \xrightarrow{g} H'$, which may be beneficial in practice.

The result $(N^\Delta, p^\Delta) = localize^\Delta(Q, \psi)$ of this construction via the procedure $localize^\Delta$ is visualized in Figure~\ref{fig:localize_delta}. The resulting RETE net consists of the localized version $(N^\Phi_Q, p^\Phi_Q) = localize(N^Q, p^Q)$ of a regular RETE net $(N^Q, p^Q)$ with height $h$ for the plain graph query $Q$, a marking filter node $[\phi > h]^\Phi$ with dependency $p^\Phi_Q$, the RETE net $(N^{sat}_{(Q, \psi)}, p^{sat}_{(Q, \psi)}) = localize^{sat}(Q, \psi)$, and a marking-sensitive union node $[\cup]^\Phi$ with dependencies $[\phi > h]^\Phi$ and $p^{sat}_{(Q, \psi)}$. In addition, $N^\Delta$ comprises the RETE net $(N^\Phi_{(Q, \psi)}, p^\Phi_{(Q, \psi)}) = localize^{\Psi}(Q, \psi)$, a request projection structure $RPS^{\infty}_l = RPS^{\infty}([\cup]^\Phi, N^\Phi_{(Q, \psi)})$, and a marking-sensitive semi-join $[\ltimes]^\Phi$ with left dependency $[\cup]^\Phi$ and right dependency $p^\Phi_{(Q, \psi)}$.

Additionally, another copy of the RETE net for computing subgraph satisfaction dependent matches, $(N^{sat'}_{(Q, \psi)}, p^{sat'}_{(Q, \psi)}) = localize^{sat}(Q, \psi)$, is created for computing subgraph satisfaction dependent matches in $H'$. These matches then have to be translated into matches into $H$ along the graph morphisms $f$ and $g$ via a dedicated \emph{marking-sensitive translation node} $[\xrightarrow{f \circ g^{-1}}]^\Phi$ with dependency $p^{sat'}_{(Q, \psi)}$, which is also added to $N^\Delta$. The target result set of this marking-sensitive translation node is given by $\resultslocal{[\xrightarrow{f \circ g^{-1}}]^\Phi}{N^\Delta}{H}{H_p}{\mathcal{C}^\Phi} = \{(m, \phi) | m \in \allmatches{Q}{H} \wedge \exists (m', \phi) \in \mathcal{C}^\Phi(p^{sat'}) : f \circ g^{-1} \circ m' = m\}$. The node $[\xrightarrow{f \circ g^{-1}}]^\Phi$ can then be added as an additional dependency to $[\cup]^\Phi$. 

\begin{figure}
\centering

\begin{tikzpicture}

\node[rectangle,
	minimum width = 9.1cm, 
	minimum height = 4.1cm] (r) at (0,0) {};

\node [below right= -4.11cm and -4.44cm of r] (join) [ms_retenode] {$\ltimes$};

\node [below left = 0.35cm and 1cm of join] (union) [ms_retenode] {$\cup$};

\node [below right = 0.35cm and 1cm of join, dashed] (q_psi) [ms_retenode] {$N^\Phi_{(Q, \psi)}$};

\node [below left = 0.35cm and 0.25cm of union] (filter) [ms_retenode] {$\phi > h$};

\node [below = 0.35cm of filter, dashed] (q) [ms_retenode] {$N^\Phi_{Q}$};

\node [below = 1.45cm of union, dashed] (q_sat) [ms_retenode] {$N^{sat}_{(Q, \psi)}$};

\node [below right = 0.35cm and 0.25cm of union] (translate) [ms_retenode] {$\xrightarrow{f \circ g^{-1}}$};

\node [below = 0.35cm of translate, dashed] (q_satp) [ms_retenode] {$N^{sat'}_{(Q, \psi)}$};

\node [below = 0.35cm of q_psi, dashed] (rpsl) [ms_retenode] {$RPS^{\infty}_l$};

\draw [-{Latex}] (join) -- (union);

\draw [-{Latex}] (join) -- (q_psi);

\draw [-{Latex}] (union) -- (filter);

\draw [-{Latex}] (filter) -- (q);

\draw [-{Latex}] (union) -- (q_sat);

\draw [-{Latex}] (union) -- (translate);

\draw [-{Latex}] (translate) -- (q_satp);

\draw [-{Latex}] (rpsl) -- (union);

\draw [-{Latex}] (q_psi) -- (rpsl);

\end{tikzpicture}
\caption{Result of applying $localize^\Delta$ to an extended graph query $(Q, \psi)$} \label{fig:localize_delta}
\end{figure}

The created net thereby computes the matches into $H$ that satisfy $\psi$ and touch the relevant subgraph or are subgraph satisfaction dependent in $H$ or $H'$:

\begin{thm}[For a subgraph-restricted graph modification, consistent configurations for RETE nets localized via $localize^\Delta$ yield all matches that touch the relevant subgraph in the source or are subgraph satisfaction dependent in the source or target of the modification] \label{the:ngc_delta_rete_completeness} 
Let $(Q, \psi)$ be an extended graph query and $H \xleftarrow{f} K \xrightarrow{g} H'$ a subgraph-restricted graph modification modifying the graph $H$ with relevant subgraph $H_p \subseteq H$ into the graph $H'$ with relevant subgraph $H'_p \subseteq H'$. Furthermore, let $\mathcal{C}^\Phi$ be a configuration that is consistent for $(N^\Delta, p^\Delta) = localize^\Delta(Q, \psi)$ for host graph $H$ and relevant subgraph $H_p$ and consistent for $(N^{sat'}_{(Q, \psi)}, p^{sat'}_{(Q, \psi)}) = localize^{sat}(Q, \psi)$ for host graph $H'$ and relevant subgraph $H'_p$. It must then hold that $\forall m \in \allmatches{Q}{H} : m \models \psi \wedge (m(Q) \cap H_p \neq \emptyset \text{ or } m \text{ is subgraph satisfaction dependent } \text{ or } \exists m' \in \allmatches{Q}{H'} : f \circ g^{-1} \circ m' = m \wedge m' \text{ is subgraph satisfaction dependent}) \Rightarrow m \in \resultsstripped{p^\Delta}{\mathcal{C}^\Phi}$.
\end{thm}

\begin{proof} (Idea)
Follows from Theorems~\ref{the:completeness_consistent_configuration},~\ref{the:localized_rete_ngcs_correctness} and~\ref{the:completeness_satisfaction_dependency}.
\end{proof}

An execution order for a RETE net $(N^\Delta, p^\Delta) = localize^\Delta(Q, \psi)$ can simply be constructed as $order^\Delta(N^\Delta) = order(N^\Phi_Q) \circ [\phi > h]^\Phi \circ order^{sat}(N^{sat}_{(Q, \psi)}) \circ [\xrightarrow{f \circ g^{-1}}]^\Phi \circ [\cup]^\Phi \circ toposort(RPS^{\infty}_l)^{-1} \circ order^\Psi(N^\Phi_{(Q, \psi)}) \circ [\ltimes]^\Phi$:

\begin{thm}[Execution of localized RETE nets via $order^\Delta$ yields consistent configurations] \label{the:ngc_delta_rete_execution} 
Let $(Q, \psi)$ be an extended graph query and $H$ a host graph with relevant subgraph $H_p$ and $\mathcal{C}^\Phi_0$ an arbitrary starting configuration for the marking-sensitive RETE net $(N^\Delta, p^\Delta) = localize^\Delta(Q, \psi)$. Executing $(N^\Delta, p^\Delta)$ via $O = order^\Delta(N^\Delta)$ then yields a consistent configuration $\mathcal{C}^\Phi = execute(O, N^\Delta, H, H_p, \mathcal{C}^\Phi_0)$.
\end{thm}

\begin{proof} (Idea)
Follows because all cyclical structures in $N^\Delta$ are encapsulated in localized RETE subnets for plain graph queries, for which $order$ creates an admissible ordering.
\end{proof}

As usual, we thus get a corollary stating the correctness of RETE net execution via $order^\Delta$:

\begin{cor}[For a subgraph-restricted graph modification, execution of RETE nets localized via $localize^\Delta$ yields all matches that touch the relevant subgraph in the source or are subgraph satisfaction dependent in the source or target graph of the modification]
Let $(Q, \psi)$ be an extended graph query and $H \xleftarrow{f} K \xrightarrow{g} H'$ a subgraph-restricted graph modification modifying the graph $H$ with relevant subgraph $H_p \subseteq H$ into the graph $H'$ with relevant subgraph $H'_p \subseteq H'$. Furthermore, let $\mathcal{C}^\Phi_0$ be a configuration for $(N^\Delta, p^\Delta) = localize^\Delta(Q, \psi)$ and $(N^{sat'}_{(Q, \psi)}, p^{sat'}_{(Q, \psi)}) = localize^{sat}(Q, \psi)$ that is consistent for $(N^{sat'}_{(Q, \psi)}, p^{sat'}_{(Q, \psi)})$ for host graph $H'$ and relevant subgraph $H'_p$. It then holds for the configuration $\mathcal{C}^\Phi = execute(O, N^\Delta, H, H_p, \mathcal{C}^\Phi_0)$ that $\forall m \in \allmatches{Q}{H} : m \models \psi \wedge (m(Q) \cap H_p \neq \emptyset \text{ or } m \text{ is subgraph satisfaction dependent } \text{or } \exists m' \in \allmatches{Q}{H'} : f \circ g^{-1} \circ m' = m \wedge m' \text{ is subgraph satisfaction dependent}) \Rightarrow m \in \resultsstripped{p^\Delta}{\mathcal{C}^\Phi}$.
\end{cor}

\begin{proof}
Follows directly from Theorem~\ref{the:ngc_delta_rete_completeness} and Theorem~\ref{the:ngc_delta_rete_execution}.
\end{proof}

Changes to the results of an extended graph query $(Q, \psi)$ resulting from a subgraph-restricted graph modification $H \xleftarrow{f} K \xrightarrow{g} H'$ for relevant subgraphs $H_p \subseteq H$ and $H'_p \subseteq H'$ can then be computed by mirroring the construction with $localize^\Delta$ for $H$ and $H'$. This entails creating RETE nets $(N^\Delta, p^\Delta) = localize^\Delta(Q, \psi)$ and $(N^{sat'}_{(Q, \psi)}, p^{sat'}_{(Q, \psi)}) = localize^{sat}(Q, \psi)$ as well as duplicates $(N^{\Delta'}, p^{\Delta'}) = localize^\Delta(Q, \psi)$ and $(N^{sat}_{(Q, \psi)}, p^{sat}_{(Q, \psi)}) = localize^{sat}(Q, \psi)$. $(N^\Delta, p^\Delta)$ and $(N^{sat}_{(Q, \psi)}, p^{sat}_{(Q, \psi)})$ are executed over $H$, whereas $(N^{\Delta'}, p^{\Delta'})$ and $(N^{sat'}_{(Q, \psi)}, p^{sat'}_{(Q, \psi)})$ are executed over $H'$. After appropriate translation of matches, the set of matches that are deleted by the modification or for which the satisfaction of $\psi$ changes from $\texttt{true}$ to $\texttt{false}$ can be computed as the difference between results for $p^\Delta$ and $p^{\Delta'}$:

\begin{thm}[Removal of matches from the results of an extended graph query caused by a subgraph-restricted graph modification can be detected using RETE nets localized via $localize^\Delta$]
Let $(Q, \psi)$ be an extended graph query and $H \xleftarrow{f} K \xrightarrow{g} H'$ a subgraph-restricted graph modification modifying the graph $H$ with relevant subgraph $H_p \subseteq H$ into the graph $H'$ with relevant subgraph $H'_p \subseteq H'$.  Furthermore, let $\mathcal{C}^\Phi$ be a configuration that is consistent for $(N^\Delta, p^\Delta) = localize^\Delta(Q, \psi)$ and $(N^{sat}_{(Q, \psi)}, p^{sat}_{(Q, \psi)}) = localize^{sat}(Q, \psi)$ for host graph $H$ and relevant subgraph $H_p$ and consistent for $(N^{\Delta'}, p^{\Delta'}) = localize^\Delta(Q, \psi)$ and $(N^{sat'}_{(Q, \psi)}, p^{sat'}_{(Q, \psi)}) = localize^{sat}(Q, \psi)$ for host graph $H'$ and relevant subgraph $H'_p$. It then holds that $\{m \in \allmatches{Q}{H} \mid m \models \psi \wedge \nexists m' \in \allmatches{Q}{H'} : m = f \circ g^{-1} \circ m' \wedge m' \models \psi\} = \resultsstripped{p^\Delta}{\mathcal{C}^\Phi} \setminus \{m \in \allmatches{Q}{H} \mid \exists m' \in \resultsstripped{p^{\Delta'}}{\mathcal{C}^{\Phi'}} : m = f \circ g^{-1} \circ m'\}$.
\end{thm}

\begin{proof} (Idea)
Inclusion in both directions can be shown by exploiting the fact that ultimately, both $(N^\Delta, p^\Delta)$ and $(N^{\Delta'}, p^{\Delta'})$ check the satisfaction of $\psi$ for sets of matches that can only differ in matches touching the relevant subgraph.
\end{proof}

Analogously, matches created by the modification or for which the satisfaction of $\psi$ changes from $\texttt{false}$ to $\texttt{true}$ can be computed as the difference between $p^{\Delta'}$ and $p^\Delta$. The required translation of matches can be realized in a straightforward manner via marking-sensitive translation nodes and difference computation can be performed via marking-sensitive anti-joins, as displayed in Figure~\ref{fig:localize_delta_combined}.

\begin{figure}
\centering

\begin{tikzpicture}

\node[rectangle,
	minimum width = 7cm, 
	minimum height = 3.25cm] (r) at (0,0) {};

\node [below left= -3.27cm and -1.71cm of r] (join) [ms_retenode] {$\rhd$};

\node [below = 1.75cm of join, dashed] (delta_q) [ms_retenode] {$N^\Delta_{(Q, \psi)}$};

\node [below right = 0.5cm and -0.5cm of join] (translate) [ms_retenode] {$\xrightarrow{g \circ f^{-1}}$};

\node [below right= -3.27cm and -1.71cm of r] (joinp) [ms_retenode] {$\rhd'$};

\node [below = 1.75cm of joinp, dashed] (delta_qp) [ms_retenode] {$N^{\Delta'}_{(Q, \psi)}$};

\node [below left = 0.5cm and -0.5cm of joinp] (translatep) [ms_retenode] {$\xrightarrow{f \circ g^{-1}}$};

\draw [-{Latex}] (join) -- (delta_q);

\draw [-{Latex}] (join) -- (translatep);

\draw [-{Latex}] (translate) -- (delta_q);

\draw [-{Latex}] (joinp) -- (delta_qp);

\draw [-{Latex}] (joinp) -- (translate);

\draw [-{Latex}] (translatep) -- (delta_qp);

\end{tikzpicture}
\caption{RETE construction for locally computing changes in results for an extended graph query $(Q, \psi)$} \label{fig:localize_delta_combined}
\end{figure}


\section{Evaluation} \label{sec:evaluation}

We aim to investigate whether RETE net localization can improve performance of query execution in scenarios where the relevant subgraph constitutes only a fraction of the full model, considering initial query execution time, execution time for incrementally processing model updates, and memory consumption as performance indicators. We experiment\footnote{Experiments were executed on a Linux SMP Debian 4.19.67-2 computer with Intel Xeon E5-2630 CPU (2.3\,GHz clock rate) and 386\,GB system memory running OpenJDK 11.0.6. Reported measurements correspond to the arithmetic mean of measurements for 10 runs. Memory measurements were obtained using the Java Runtime class. To represent graph data, all experiments use EMF \cite{emf} object graphs that enable reverse navigation of edges. Our implementation is available under \cite{implementation}. More details and query visualizations can be found in \cite{preprint} and \cite{implementation} or in Appendix~\ref{app:queries}.} with the following querying techniques:

\begin{itemize}
\item \textbf{STANDARD}: Our own implementation of a regular RETE net with global execution semantics \cite{barkowsky2023host}.
\item \textbf{LOCALIZED}: Our own implementation of the RETE net used for STANDARD, localized according to the description in Section~\ref{sec:localized_search_with_marking-sensitive_rete} and Section~\ref{sec:localized_rete_extended_queries}.
\item \textbf{DELTA}: Our own implementation of the RETE net used for STANDARD, localized for detecting changes to the results of extended graph queries according to the description in Section~\ref{sec:localized_rete_satisfaction_changes}.
\item \textbf{VIATRA}: The external RETE-based VIATRA tool \cite{varro2016}.
\item \textbf{SDM*}: Our own local-search-based Story Diagram Interpreter tool \cite{giese2009improved}.
\end{itemize}

Note that STANDARD and LOCALIZED implement the additional optimizations for elimination of redundancy and downward propagation of nested graph conditions described in \cite{barkowsky2023host}. To the extent possible, DELTA employs these optimizations as well.

For SDM*, we only consider searching for new query results. We thus underapproximate the time and memory required for a full solution with this strategy, which would also require maintaining previously found results.

\subsection{Plain Graph Queries over Synthetic Abstract Syntax Graphs} \label{sec:evaluation_java_scalability}

We first attempt a systematic evaluation via a synthetic experiment, which emulates a developer loading part of a large model into their workspace and monitoring some well-formedness constraints as they modify the loaded part, that is, relevant subgraph, without simultaneous changes to other model parts.

We therefore generate Java abstract syntax graphs with 1, 10, 100, 1000, and 10000 packages, with each package containing 10 classes with 10 fields referencing other classes in the same or a different package. As relevant subgraph, we consider a single package and its contents. We then execute a plain graph query searching for paths consisting of a package and four classes connected via fields.

After the initial query execution, we modify the relevant subgraph by creating a class in the considered package along with 10 fields referencing other existing classes in the relevant subgraph, that is, the same package. We then perform an incremental update of query results. This step of modifying the relevant subgraph and updating query results is repeated 10 times.

Due to the query being a simple path, the effort for query execution largely depends on the number of overall query results. This number is in turn directly determined by the number of classes considered during execution and their number of connections to other classes. By fixing the former for localized execution via the selection of the relevant subgraph and fixing the latter for both localized and global execution via the construction of the host graph, we thus aim to isolate the effect of localization on the performance of query execution.

Figure~\ref{fig:java_scalability_bars} (left) displays the execution times for the initial execution of the query. The execution time of LOCALIZED remains around 120\,ms for all model sizes. The execution time for SDM* slowly grows from around 350\,ms to 1025\,ms due to indexing effort that is necessary for observing model changes. In contrast, the execution time for the other RETE-based strategies clearly scales with model size, with the execution time for STANDARD growing from around 13\,ms for the smallest model to more than 184\,000\,ms for the largest model. On the one hand, localization thus incurs a noticeable overhead in initial execution time for the smallest model, where even localized query execution is essentially global. On the other hand, it significantly improves execution time for the larger models and even achieves better scalability than the local-search-based tool in this scenario.

The average times for processing a model update are displayed in Figure~\ref{fig:java_scalability_bars} (center). Here, all strategies achieve execution times mostly independent of model size. While the measurements for STANDARD fluctuate, likely due to the slightly unpredictable behavior of hash-based indexing structures, average execution times remain low overall and below 10\,ms for LOCALIZED. Still, localization incurs a noticeable overhead up to factor 6 compared to STANDARD and VIATRA. This overhead is expected, since in this scenario, all considered updates affect the relevant subgraph and thus impact the results of the localized RETE net similarly to the results of the standard RETE nets. Consequently, localization does not reduce computational effort, but causes overhead instead.

Finally, Figure~\ref{fig:java_scalability_bars} (right) shows the memory measurements for all strategies and models after the final update. Here, LOCALIZED again achieves a substantial improvement in scalability compared to the other RETE-based strategies, with a slightly higher memory consumption for the smallest model and an improvement by factor 120 over STANDARD for the largest model. This is a result of the localized RETE net producing the same number of intermediate results for all model sizes, with the slight growth in memory consumption likely a product of the growing size of the model itself. SDM*, not storing any matches, performs better for all but the largest model, where memory consumption surpasses the measurement for LOCALIZED. This surprising result can probably be explained by the fact that SDM* has to index the full model to observe modifications, causing an overhead in memory consumption.

\begin{figure}
\centering
\includegraphics[width=\textwidth]{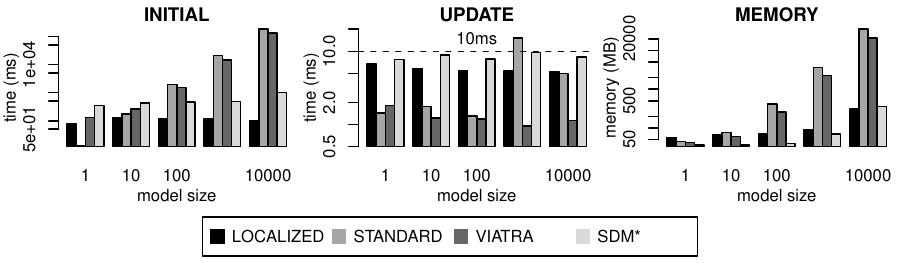}
\caption{Measurements for the synthetic abstract syntax graph scenario (log scale)}\label{fig:java_scalability_bars}
\end{figure}

In addition to these experiments, where the full model is always stored in main memory, we also experimented with a model initially stored on disk via the persistence layer CDO \cite{cdo}. Measurements mostly mirror those for the main-memory-based experiment. Notably though, in conjunction with CDO the LOCALIZED strategy achieves almost ideal scalability regarding memory consumption for this scenario, with measurements around 70\,MB for all model sizes.\footnote{See \cite{preprint} for measurement results.}

\subsection{Plain Graph Queries over Real Abstract Syntax Graphs} \label{sec:evaluation_java_incremental}

To evaluate our approach in the context of a more realistic application scenario, we perform a similar experiment using real data and plain graph queries inspired by object-oriented design patterns \cite{gamma1993design}. In contrast to the synthetic scenario, this experiment emulates a situation where modifications may concern not only the relevant subgraph but the entire model, for instance when multiple developers are simulatenously working on different model parts.

We therefore extract a history of real Java abstract syntax graphs with about 16\,000 vertices and 45\,000 edges from a software repository using the MoDisco tool \cite{bruneliere2010modisco,bruneliere2014modisco}. After executing the queries over the initial commit, we replay the history and perform incremental updates of query results after each commit. The employed queries are visualized in Appendix~\ref{app:queries}. As relevant subgraph, we again consider a single package and its contents.

Figure~\ref{fig:java_incremental_time} displays the aggregate execution time for processing the commits one after another for the queries where LOCALIZED performed best and worst compared to STANDARD, with the measurement at $x = 0$ indicating the initial execution time for the starting model. Initial execution times are similarly low due to a small starting model and in fact slightly higher for LOCALIZED. However, on aggregate LOCALIZED outperforms STANDARD with an improvement between factor 5 and 18 due to significantly lower update times, which are summarized in Figure~\ref{fig:update_clouds} (left).

In this case, the improvement mostly stems from the more precise monitoring of the model for modifications: The RETE nets of both STANDARD and LOCALIZED remain small due to strong filtering effects in the query graphs. However, while STANDARD spends significant effort on processing model change notifications due to observing all appropriately typed model elements, this effort is substantially reduced for LOCALIZED, which only monitors elements relevant to query results required for completeness under the relevant subgraph. The execution times of SDM* can be explained by the same effect. Interestingly, VIATRA seems to implement an improved handling of such notifications, achieving improved execution times for particularly small updates even compared to LOCALIZED, but requiring more time if an update triggers changes to the RETE net. Combined with a higher RETE net initialization time, this results in LOCALIZED also outperforming VIATRA for all considered queries.

Regarding memory consumption, all strategies perform very similarly, which is mostly a result of the size of the model itself dominating the measurement and hiding the memory impact of the rather small RETE nets.

\begin{figure}
\centering
\includegraphics[width=\textwidth]{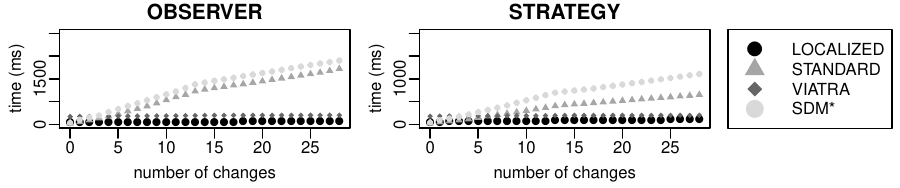}
\caption{Execution times for the real abstract syntax graph scenario}\label{fig:java_incremental_time}
\end{figure}

\begin{figure}
\centering
\includegraphics[width=\textwidth]{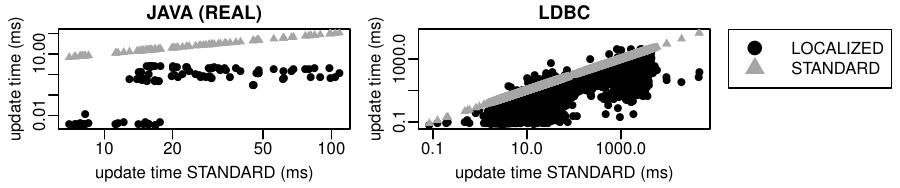}
\caption{Summary of update processing times for plain graph queries in different scenarios (log scale)}\label{fig:update_clouds}
\end{figure}

\subsection{LDBC Social Network Benchmark} \label{sec:evaluation_ldbc}

Finally, we also perform experiments inspired by the independent LDBC Social Network Benchmark \cite{erling2015ldbc,angles2024ldbc}, simulating a case where a user of a social network wants to incrementally track query results relating to them personally.\footnote{Our experiments are not to be confused with an official run of the LDBC Social Network Benchmark. The benchmark specification and data generator only serve as a source of plausible data and queries.}

We therefore generate a synthetic social network consisting of around 850\,000 vertices, including about 1700 persons, and 5\,500\,000 edges using the benchmark's data generator and the predefined scale factor 0.1. We subsequently transform this dataset into a sequence of element creations and deletions based on the timestamps included in the data. We then create a starting graph by replaying the first half of the sequence and perform an initial execution of adapted versions of twelve benchmark queries consisting of plain graph patterns, with a person with a close-to-average number of contacts in the final social network designated as relevant subgraph. After the initial query execution, we replay the remaining changes, incrementally updating the query results after each change. To evaluate our approach for localizing execution of extended graph queries via the RETE mechanism, we also perform analogous experiments with versions of two queries that are equipped with nested graph conditions according to their original formulation in the benchmark. Visualizations of the employed queries can be found in Appendix~\ref{app:queries}.

\subsubsection{Plain Graph Queries}

The resulting execution times for the plain queries where LOCALIZED performed best and worst compared to STANDARD are displayed in Figure~\ref{fig:ldbc_time}. A summary of all update time measurements for LOCALIZED in comparison with STANDARD is also displayed in Figure~\ref{fig:update_clouds} (right). For all queries, LOCALIZED ultimately outperforms the other approaches by a substantial margin, as the localized RETE version forgoes the computation of a large number of irrelevant intermediate results due to the small relevant subgraph on the one hand and avoids redundant computations on the other hand.

\begin{figure}
\centering
\includegraphics[width=\textwidth]{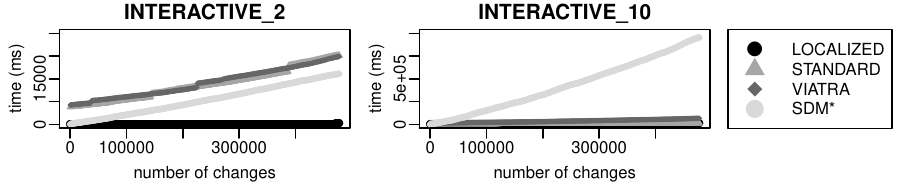}
\caption{Execution times for the LDBC scenario}\label{fig:ldbc_time}
\end{figure}

\begin{figure}
\centering
\includegraphics[width=\textwidth]{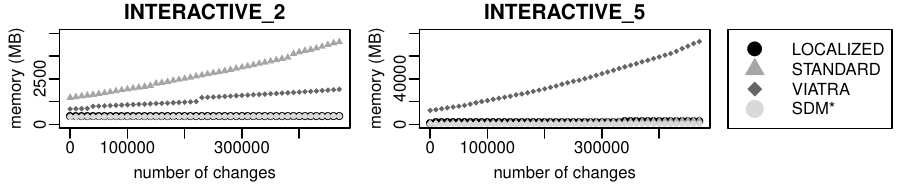}
\caption{Memory measurements for the LDBC scenario}\label{fig:ldbc_memory}
\end{figure}

The memory measurements in Figure~\ref{fig:ldbc_memory} mostly mirror execution times for RETE-based approaches, with the memory consumption for LOCALIZED always lower than for STANDARD and VIATRA except for a period at the beginning of the execution of the query INTERACTIVE\_5, where STANDARD outperforms LOCALIZED. The weaker performance for INTERACTIVE\_5 and INTERACTIVE\_10 likely stems from the fact that, disregarding edge direction, the associated query graphs contain cycles that act as strong filters for subsequent (intermediate) results. These filters then achieve a somewhat similar effect as localization. The weaker performance of VIATRA for INTERACTIVE\_5 is a product of the usage of a suboptimal RETE net. As expected, memory consumption is lowest for SDM* for all queries.

\subsubsection{Extended Graph Queries}

The resulting execution times for the considered extended graph queries are displayed in Figure~\ref{fig:ldbc_ngc_time}. A visual comparison between update times for LOCALIZED and DELTA with STANDARD is shown in Figure~\ref{fig:update_clouds_ngc}.

For both queries, LOCALIZED outperforms STANDARD and VIATRA significantly with respect to both initial execution time and update times, as it benefits from the same localized computation of matches for the base pattern as in the plain query case. Moreover, checking effort for the equipped nested graph conditions is also reduced, with checks only performed locally in the context of relevant base pattern matches rather than globally.

Compared to LOCALIZED, DELTA requires noticeably more initial execution time and time for processing updates for both queries. The deterioration in performance is likely a product of several factors. First, the required duplication of the RETE nets for localized match computation and checking of nested graph condition creates obvious redundancy. Second, DELTA introduces additional RETE nets for computing subgraph satisfaction dependent matches for the base pattern. Third, effectively enabling the execution of queries over two versions of the host graph at the same time comes with additional indexing effort required for the DELTA strategy.

The deterioration in performance is slightly more pronounced for the query INTERACTIVE\_4\_NGC. This is likely due to the fact that the typing of vertices in the nested graph conditions of INTERACTIVE\_3\_NGC prevents changes to the relevant subgraph from impacting the conditions' satisfaction, which leads to empty RETE subnets for the computation of subgraph satisfaction dependent matches.

\begin{figure}
\centering
\includegraphics[width=\textwidth]{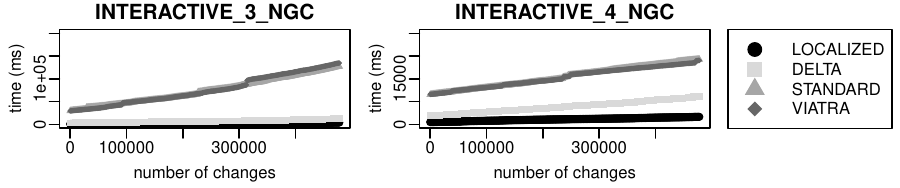}
\caption{Execution times for the LDBC scenario}\label{fig:ldbc_ngc_time}
\end{figure}

\begin{figure}
\centering
\includegraphics[width=\textwidth]{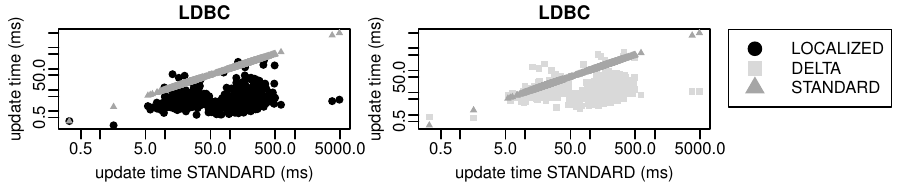}
\caption{Summary of update processing times for extended graph queries in the LDBC scenario (log scale)}\label{fig:update_clouds_ngc}
\end{figure}

\begin{figure}
\centering
\includegraphics[width=\textwidth]{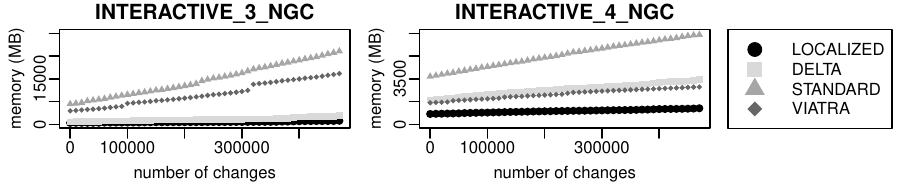}
\caption{Memory measurements for the LDBC scenario}\label{fig:ldbc_ngc_memory}
\end{figure}

Most observations related to execution time again also apply to memory consumption due to the close link between RETE net execution time and configuration size. The most notable exception is the lower memory consumption measured for VIATRA and the query INTERACTIVE\_4\_NGC compared to DELTA. This is likely a product of a generally more memory efficient RETE implementation in VIATRA in comparison with our own, which is also illustrated by the difference in memory consumption between VIATRA and STANDARD despite very similar execution times.

Note that SDM* is excluded from our experiments with extended graph queries. This is due to the fact that it is unclear how to implement change-based search of matches for extended graph queries purely with local search and without implementing some kind of auxiliary data structure for storing information about the search process.

\subsection{Discussion}

On the one hand, our experimental results demonstrate that in situations where the relevant subgraph constitutes only a fraction of the full model, RETE net localization can improve the performance of incremental query execution compared to both the standard RETE approach and a solution based on local search. In such scenarios, localization can improve scalability with respect to initial query execution time and memory consumption, as demonstrated in Sections~\ref{sec:evaluation_java_scalability} and~\ref{sec:evaluation_ldbc} and, if changes are not restricted to the relevant subgraph, also update processing time, as shown in Sections~\ref{sec:evaluation_java_incremental} and~\ref{sec:evaluation_ldbc}.

On the other hand, as demonstrated in Section~\ref{sec:evaluation_java_scalability}, localization incurs an overhead on update processing time if changes are only made to the relevant subgraph and on initial execution time and memory consumption if the relevant subgraph contains most of the elements in the full model. While this overhead will essentially be limited to a constant factor in many scenarios, as analyzed in Section~\ref{sec:localized_rete_performance} and~\ref{sec:localized_rete_extended_queries}, the standard RETE approach remains preferable for query execution with global semantics or if modifications are restricted to the relevant subgraph and initial query execution time and memory consumption are irrelevant.

These observations apply to both plain and extended graph queries when it comes to the application scenario of tracking all matches for the base query pattern that include elements from the relevant subgraph. In the case of plain graph queries, the application scenario of localized tracking of changes to a global query result coincides with the first. However, for extended graph queries, localization requires noticeably more effort in the second scenario and can lead to increased theoretical complexity if further optimizations are involved, as discussed in Section~\ref{sec:localized_rete_satisfaction_changes}. However, our evaluation results in in Section~\ref{sec:evaluation_ldbc} indicate that it may still prove beneficial in practice.

\subsection{Threats to Validity}

To mitigate internal threats to the validity of our results resulting from unexpected JVM behavior, we have performed 10 runs of all experiments. However, with reliable memory measurements a known pain point of Java-based experiments, the reported memory consumption values are still not necessarily accurate and can only serve as an indicator. To minimize the impact of the implementation on conceptual observations, we compare the prototypical implementation of our approach to a regular RETE implementation \cite{barkowsky2023host}, which shares a large portion of the involved code, and to two existing tools \cite{varro2016,giese2009improved}.

We have attempted to address external threats to validity via experiments accounting for different application domains and a combination of synthetic and real-world queries and data, including a setting from an established, independent benchmark. Still, our results cannot be generalized and do not support quantitative claims, but serve to demonstrate conceptual advantages and disadvantages of the presented approach.


\section{Related Work} \label{sec:related_work}

With graph query execution forming the foundation of many applications, there already exists an extensive body of research regarding its optimization.

Techniques based on local search \cite{cordella2004sub,geiss2006grgen,giese2009improved,arendt2010henshin,han2013turboiso,bi2016efficient,juttner2018vf2++} constitute one family of graph querying approaches. While they are designed to exploit locality in the host graph to improve execution time, repeated query execution leads to redundant computations that are only avoided by fully incremental techniques.

In \cite{egyed2006instant}, Egyed proposes a scoping mechanism for local search to support incremental query execution, only recomputing query results when a graph element touched during query execution changes. While this approach offers some degree of incrementality, it is limited to queries with designated root elements that serve as an anchor in the host graph and may still result in redundant computations, since query reevaluation is only controlled at the granularity of root elements.

A second family of solutions is based on discrimination networks \cite{hanson2002trigger,varro2013rete,varro2016,beyhl2018framework}, the most prominent example of which are RETE nets.

VIATRA \cite{varro2016} is a mature tool for incremental graph query execution based on the RETE algorithm \cite{forgy1989rete}, which supports advanced concepts for query specification and optimization not considered in this article. Notably, VIATRA allows reuse of matches for isomorphic query subgraphs within a single RETE net. This is achieved via RETE structures not covered by the rather restrictive definition of well-formedness used in this article, which points to a possible direction for future work. However, while VIATRA also has a local search option for query execution, it does not integrate local search with the incremental query engine but rather offers it as an alternative.

Beyhl \cite{beyhl2018framework} presents an incremental querying technique based on a generalized version of RETE nets, called Generalized Discrimination Networks (GDNs) \cite{hanson2002trigger}. The main difference compared to the RETE algorithm is the lack of join nodes. Instead, more complex nodes that directly compute complex matches using local search are employed. The approach however represents more of a means of controlling the trade-off between local search and RETE rather than an integration and still requires a global computation of matches for the entire host graph.

In previous work \cite{barkowsky2021improving}, we have made a first step in the direction of localizing RETE-based query execution. While this earlier technique already allowed anchoring the execution of a RETE net to certain host graph vertices, this anchoring was based on typing information and its results did not meet the definition of completeness introduced in this article.

Model repositories such as CDO \cite{cdo} and NeoEMF \cite{daniel2017neoemf} provide support for query execution over partial models via lazy loading. As persistence layers, these solutions however focus on implementing an interface of atomic model access operations in order to be agnostic regarding the employed query mechanism.

The Mogwa\"{\i} tool \cite{daniel2018scalable} aims to improve query execution over persistence layers like CDO and NeoEMF by mapping model queries to native queries for the underlying database system instead of using the atomic model access operations provided by the layer's API, avoiding loading the entire model into main memory. The tool however does not consider incremental query execution.

Jahanbin et al. propose an approach for querying partially loaded models stored via persistence layers \cite{jahanbin2022partial} or as XMI files \cite{jahanbin2023towards}. In contrast to the solution presented in this article, their approach still aims to always provide complete query results for the full model and is thus based on static analysis and typing information rather than dynamic exploitation of locality.

Query optimization for relational databases is a research topic that has been under intense study for decades \cite{krishnamurthy1986optimization,lee2001optimizing,leis2015good}. Generally, many of the techniques from this field are applicable to RETE nets, which are ultimately based on relational algebra and related to materialized views in relational databases \cite{gupta1995maintenance}. However, relational databases lack the notion of locality inherent to graph-based encodings and are hence not tailored to exploit local navigation.

This shortcoming has given rise to a number of graph databases \cite{angles2012comparison}, which employ a graph-based data representation instead of a relational encoding and form the basis of some model persistance layers like NeoEMF \cite{daniel2017neoemf}. While these database systems are designed to accommodate local navigation, to the best of our knowledge, support for incremental query execution is still lacking.


\section{Conclusion} \label{sec:conclusion}

In this article, we have presented a relaxed notion of completeness for query results that lifts the requirement of strict completeness of results for graph queries and thereby the need for necessarily global query execution. Based on this relaxed notion of completeness, we have developed an extension of the RETE approach that allows local, yet fully incremental execution of both plain and extended graph queries. An initial evaluation demonstrates that the approach can improve scalability in scenarios with small relevant subgraphs, but causes an overhead in unfavorable cases.

In future work, we want to further evaluate the performance of the presented approach. In particular, we are interested in exploring whether the RETE net localization technique for extended graph queries can improve performance in certain scenarios even if host graph and relevant subgraph coincide. We also plan to investigate whether the proposed solution can be adapted to support bulk loading of partial models in order to reduce overhead caused by lazy loading strategies employed by model persistence layers.

\section*{Acknowledgments}
\noindent
This work was developed in the course of the project modular and incremental Global Model Management (project number 336677879) funded by the DFG.


\bibliographystyle{alphaurl}
\bibliography{localized_rete_for_incremental_graph_queries}

\newcommand{\etalchar}[1]{$^{#1}$}
\begin{thebibliography}{EALP{\etalchar{+}}15}

\bibitem[AAA{\etalchar{+}}24]{angles2024ldbc}
Renzo Angles, János~Benjamin Antal, Alex Averbuch, Altan Birler, Peter Boncz, Márton Búr, Orri Erling, Andrey Gubichev, Vlad Haprian, Moritz Kaufmann, Josep Lluís~Larriba Pey, Norbert Martínez, József Marton, Marcus Paradies, Minh-Duc Pham, Arnau Prat-Pérez, David Püroja, Mirko Spasić, Benjamin~A. Steer, Dávid Szakállas, Gábor Szárnyas, Jack Waudby, Mingxi Wu, and Yuchen Zhang.
\newblock {The LDBC Social Network Benchmark}, 2024.
\newblock \href {https://arxiv.org/abs/2001.02299} {\path{arXiv:2001.02299}}, \href {https://doi.org/10.48550/arXiv.2001.02299} {\path{doi:10.48550/arXiv.2001.02299}}.

\bibitem[ABJ{\etalchar{+}}10]{arendt2010henshin}
Thorsten Arendt, Enrico Biermann, Stefan Jurack, Christian Krause, and Gabriele Taentzer.
\newblock {Henshin: advanced concepts and tools for in-place EMF model transformations}.
\newblock In {\em International Conference on Model Driven Engineering Languages and Systems}, pages 121--135. Springer, 2010.
\newblock \href {https://doi.org/10.1007/978-3-642-16145-2_9} {\path{doi:10.1007/978-3-642-16145-2_9}}.

\bibitem[Ang12]{angles2012comparison}
Renzo Angles.
\newblock A comparison of current graph database models.
\newblock In {\em International Conference on Data Engineering Workshops}, pages 171--177. IEEE, 2012.
\newblock \href {https://doi.org/10.1109/ICDEW.2012.31} {\path{doi:10.1109/ICDEW.2012.31}}.

\bibitem[BBG21]{barkowsky2021improving}
Matthias Barkowsky, Thomas Brand, and Holger Giese.
\newblock Improving adaptive monitoring with incremental runtime model queries.
\newblock In {\em Symposium on Software Engineering for Adaptive and Self-Managing Systems}, pages 71--77. IEEE, 2021.
\newblock \href {https://doi.org/10.1109/SEAMS51251.2021.00019} {\path{doi:10.1109/SEAMS51251.2021.00019}}.

\bibitem[BCDM14]{bruneliere2014modisco}
Hugo Bruneliere, Jordi Cabot, Gr{\'e}goire Dup{\'e}, and Fr{\'e}d{\'e}ric Madiot.
\newblock Modisco: A model driven reverse engineering framework.
\newblock {\em Information and Software Technology}, 56(8):1012--1032, 2014.
\newblock \href {https://doi.org/10.1016/j.infsof.2014.04.007} {\path{doi:10.1016/j.infsof.2014.04.007}}.

\bibitem[BCJM10]{bruneliere2010modisco}
Hugo Bruneliere, Jordi Cabot, Fr{\'e}d{\'e}ric Jouault, and Fr{\'e}d{\'e}ric Madiot.
\newblock {MoDisco: a generic and extensible framework for model driven reverse engineering}.
\newblock In {\em IEEE/ACM International Conference on Automated Software Engineering}, pages 173--174, 2010.
\newblock \href {https://doi.org/10.1145/1858996.1859032} {\path{doi:10.1145/1858996.1859032}}.

\bibitem[BCL{\etalchar{+}}16]{bi2016efficient}
Fei Bi, Lijun Chang, Xuemin Lin, Lu~Qin, and Wenjie Zhang.
\newblock {Efficient subgraph matching by postponing cartesian products}.
\newblock In {\em International Conference on Management of Data}, pages 1199--1214. ACM, 2016.
\newblock \href {https://doi.org/10.1145/2882903.2915236} {\path{doi:10.1145/2882903.2915236}}.

\bibitem[Bey18]{beyhl2018framework}
Thomas Beyhl.
\newblock {\em {A framework for incremental view graph maintenance}}.
\newblock PhD thesis, Hasso Plattner Institute at the University of Potsdam, 2018.

\bibitem[BG23]{barkowsky2023host}
Matthias Barkowsky and Holger Giese.
\newblock {Host-graph-sensitive RETE nets for incremental graph pattern matching with nested graph conditions}.
\newblock {\em Journal of Logical and Algebraic Methods in Programming}, 131, 2023.
\newblock \href {https://doi.org/10.1016/j.jlamp.2022.100841} {\path{doi:10.1016/j.jlamp.2022.100841}}.

\bibitem[BG24a]{barkowsky2024localized}
Matthias Barkowsky and Holger Giese.
\newblock {Localized RETE for incremental graph queries}.
\newblock In {\em International Conference on Graph Transformation}, pages 118--137. Springer, 2024.
\newblock \href {https://doi.org/10.1007/978-3-031-64285-2_7} {\path{doi:10.1007/978-3-031-64285-2_7}}.

\bibitem[BG24b]{preprint}
Matthias Barkowsky and Holger Giese.
\newblock {Localized RETE for incremental graph queries}.
\newblock {\em arXiv preprint}, 2024.
\newblock \href {https://doi.org/10.48550/arXiv.2405.01145} {\path{doi:10.48550/arXiv.2405.01145}}.

\bibitem[BG25]{implementation}
Matthias Barkowsky and Holger Giese.
\newblock {Localized RETE for Incremental Graph Queries with Nested Graph Conditions Evaluation Artifacts}, 2025.
\newblock Zenodo.
\newblock \href {https://doi.org/10.5281/zenodo.15754399} {\path{doi:10.5281/zenodo.15754399}}.

\bibitem[cdo]{cdo}
{Eclipse CDO Model Repository}.
\newblock \url{https://projects.eclipse.org/projects/modeling.emf.cdo}.
\newblock Last accessed 13 December 2024.

\bibitem[CFSV04]{cordella2004sub}
Luigi~P Cordella, Pasquale Foggia, Carlo Sansone, and Mario Vento.
\newblock {A (sub) graph isomorphism algorithm for matching large graphs}.
\newblock {\em IEEE Transactions on Pattern Analysis and Machine Intelligence}, 26(10):1367--1372, 2004.
\newblock \href {https://doi.org/10.1109/TPAMI.2004.75} {\path{doi:10.1109/TPAMI.2004.75}}.

\bibitem[Cod70]{codd1970relational}
Edgar~F Codd.
\newblock {A relational model of data for large shared data banks}.
\newblock {\em Communications of the ACM}, 13(6):377--387, 1970.
\newblock \href {https://doi.org/10.1145/362384.362685} {\path{doi:10.1145/362384.362685}}.

\bibitem[DSB{\etalchar{+}}17]{daniel2017neoemf}
Gwendal Daniel, Gerson Suny{\'e}, Amine Benelallam, Massimo Tisi, Yoann Vernageau, Abel G{\'o}mez, and Jordi Cabot.
\newblock {NeoEMF: A multi-database model persistence framework for very large models}.
\newblock {\em Science of Computer Programming}, 149:9--14, 2017.
\newblock \href {https://doi.org/10.1016/j.scico.2017.08.002} {\path{doi:10.1016/j.scico.2017.08.002}}.

\bibitem[DSC18]{daniel2018scalable}
Gwendal Daniel, Gerson Suny{\'e}, and Jordi Cabot.
\newblock {Scalable queries and model transformations with the Mogwa\"i tool}.
\newblock In {\em Theory and Practice of Model Transformation}, pages 175--183. Springer, 2018.
\newblock \href {https://doi.org/10.1007/978-3-319-93317-7_9} {\path{doi:10.1007/978-3-319-93317-7_9}}.

\bibitem[EALP{\etalchar{+}}15]{erling2015ldbc}
Orri Erling, Alex Averbuch, Josep Larriba-Pey, Hassan Chafi, Andrey Gubichev, Arnau Prat, Minh-Duc Pham, and Peter Boncz.
\newblock {The LDBC social network benchmark: Interactive workload}.
\newblock In {\em ACM SIGMOD International Conference on Management of Data}, pages 619--630. ACM, 2015.
\newblock \href {https://doi.org/10.1145/2723372.2742786} {\path{doi:10.1145/2723372.2742786}}.

\bibitem[EEPT06]{Ehrig+2006}
Hartmut Ehrig, Karsten Ehrig, Ulrike Prange, and Gabriele Taentzer.
\newblock {\em {Fundamentals of algebraic graph transformation}}.
\newblock Springer, 2006.
\newblock \href {https://doi.org/10.1007/3-540-31188-2} {\path{doi:10.1007/3-540-31188-2}}.

\bibitem[Egy06]{egyed2006instant}
Alexander Egyed.
\newblock {Instant consistency checking for the UML}.
\newblock In {\em International Conference on Software Engineering}, pages 381--390, 2006.
\newblock \href {https://doi.org/10.1145/1134285.1134339} {\path{doi:10.1145/1134285.1134339}}.

\bibitem[emf]{emf}
{EMF: Eclipse Modeling Framework}.
\newblock \url{https://www.eclipse.org/modeling/emf/}.
\newblock Last accessed 13 December 2024.

\bibitem[For89]{forgy1989rete}
Charles~L Forgy.
\newblock {Rete: A fast algorithm for the many pattern/many object pattern match problem}.
\newblock In {\em Readings in Artificial Intelligence and Databases}, pages 547--559. Elsevier, 1989.
\newblock \href {https://doi.org/10.1016/0004-3702(82)90020-0} {\path{doi:10.1016/0004-3702(82)90020-0}}.

\bibitem[GBG{\etalchar{+}}06]{geiss2006grgen}
Rubino Gei{\ss}, Gernot~Veit Batz, Daniel Grund, Sebastian Hack, and Adam Szalkowski.
\newblock {GrGen: A fast SPO-based graph rewriting tool}.
\newblock In {\em International Conference on Graph Transformation}, pages 383--397. Springer, 2006.
\newblock \href {https://doi.org/10.1007/11841883_27} {\path{doi:10.1007/11841883_27}}.

\bibitem[GHJV93]{gamma1993design}
Erich Gamma, Richard Helm, Ralph Johnson, and John Vlissides.
\newblock {Design patterns: Abstraction and reuse of object-oriented design}.
\newblock In {\em European Conference on Object-Oriented Programming}, pages 406--431. Springer, 1993.
\newblock \href {https://doi.org/10.1007/3-540-47910-4_21} {\path{doi:10.1007/3-540-47910-4_21}}.

\bibitem[GHS09]{giese2009improved}
Holger Giese, Stephan Hildebrandt, and Andreas Seibel.
\newblock {Improved flexibility and scalability by interpreting story diagrams}.
\newblock {\em Electronic Communications of the EASST}, 18, 2009.
\newblock \href {https://doi.org/10.14279/tuj.eceasst.18.268} {\path{doi:10.14279/tuj.eceasst.18.268}}.

\bibitem[GM99]{gupta1995maintenance}
Ashish Gupta and Inderpal~Singh Mumick.
\newblock {Maintenance of materialized views: problems, techniques, and applications}.
\newblock In {\em Materialized Views: Techniques, Implementations, and Applications}. The MIT Press, 05 1999.
\newblock \href {https://doi.org/10.7551/mitpress/4472.003.0016} {\path{doi:10.7551/mitpress/4472.003.0016}}.

\bibitem[HBC02]{hanson2002trigger}
Eric~N Hanson, Sreenath Bodagala, and Ullas Chadaga.
\newblock {Trigger condition testing and view maintenance using optimized discrimination networks}.
\newblock {\em IEEE Transactions on Knowledge and Data Engineering}, 14(2):261--280, 2002.
\newblock \href {https://doi.org/10.1109/69.991716} {\path{doi:10.1109/69.991716}}.

\bibitem[HLL13]{han2013turboiso}
Wook-Shin Han, Jinsoo Lee, and Jeong-Hoon Lee.
\newblock {Turboiso: towards ultrafast and robust subgraph isomorphism search in large graph databases}.
\newblock In {\em ACM SIGMOD International Conference on Management of Data}, pages 337--348. ACM, 2013.
\newblock \href {https://doi.org/10.1145/2463676.2465300} {\path{doi:10.1145/2463676.2465300}}.

\bibitem[HP09]{habel2009correctness}
Annegret Habel and Karl-Heinz Pennemann.
\newblock {Correctness of high-level transformation systems relative to nested conditions}.
\newblock {\em Mathematical Structures in Computer Science}, 19(2):245--296, 2009.
\newblock \href {https://doi.org/10.1017/S0960129508007202} {\path{doi:10.1017/S0960129508007202}}.

\bibitem[JKG23]{jahanbin2023towards}
Sorour Jahanbin, Dimitris Kolovos, and Simos Gerasimou.
\newblock {Towards memory-efficient validation of large XMI models}.
\newblock In {\em ACM/IEEE International Conference on Model Driven Engineering Languages and Systems Companion (MODELS-C)}, pages 241--250. IEEE, 2023.
\newblock \href {https://doi.org/10.1109/MODELS-C59198.2023.00053} {\path{doi:10.1109/MODELS-C59198.2023.00053}}.

\bibitem[JKGS22]{jahanbin2022partial}
Sorour Jahanbin, Dimitris Kolovos, Simos Gerasimou, and Gerson Suny{\'e}.
\newblock {Partial loading of repository-based models through static analysis}.
\newblock In {\em ACM SIGPLAN International Conference on Software Language Engineering}, pages 266--278, 2022.
\newblock \href {https://doi.org/10.1145/3567512.3567535} {\path{doi:10.1145/3567512.3567535}}.

\bibitem[JM18]{juttner2018vf2++}
Alp{\'a}r J{\"u}ttner and P{\'e}ter Madarasi.
\newblock {VF2++—An improved subgraph isomorphism algorithm}.
\newblock {\em Discrete Applied Mathematics}, 242:69--81, 2018.
\newblock \href {https://doi.org/10.1016/j.dam.2018.02.018} {\path{doi:10.1016/j.dam.2018.02.018}}.

\bibitem[KBZ86]{krishnamurthy1986optimization}
Ravi Krishnamurthy, Haran Boral, and Carlo Zaniolo.
\newblock {Optimization of nonrecursive queries.}
\newblock In {\em International Conference on Very Large Data Bases}, volume~86, pages 128--137, 1986.

\bibitem[Ken02]{kent2002model}
Stuart Kent.
\newblock {Model driven engineering}.
\newblock In {\em International Conference on Integrated Formal Methods}, pages 286--298. Springer, 2002.
\newblock \href {https://doi.org/10.1007/3-540-47884-1_16} {\path{doi:10.1007/3-540-47884-1_16}}.

\bibitem[LGM{\etalchar{+}}15]{leis2015good}
Viktor Leis, Andrey Gubichev, Atanas Mirchev, Peter Boncz, Alfons Kemper, and Thomas Neumann.
\newblock {How good are query optimizers, really?}
\newblock {\em Proceedings of the VLDB Endowment}, 9(3):204--215, 2015.
\newblock \href {https://doi.org/10.14778/2850583.2850594} {\path{doi:10.14778/2850583.2850594}}.

\bibitem[LSC01]{lee2001optimizing}
Chiang Lee, Chi-Sheng Shih, and Yaw-Huei Chen.
\newblock {Optimizing large join queries using a graph-based approach}.
\newblock {\em IEEE Transactions on Knowledge and Data Engineering}, 13(2):298--315, 2001.
\newblock \href {https://doi.org/10.1109/69.917567} {\path{doi:10.1109/69.917567}}.

\bibitem[Ren04]{rensink2004representing}
Arend Rensink.
\newblock {Representing first-order logic using graphs}.
\newblock In {\em International Conference on Graph Transformation}, pages 319--335. Springer, 2004.
\newblock \href {https://doi.org/10.1007/978-3-540-30203-2_23} {\path{doi:10.1007/978-3-540-30203-2_23}}.

\bibitem[Rib99]{ribeiro1999parallel}
Leila Ribeiro.
\newblock Parallel composition of graph grammars.
\newblock {\em Applied categorical structures}, 7(4):405--430, 1999.
\newblock \href {https://doi.org/10.1023/A:1008691205954} {\path{doi:10.1023/A:1008691205954}}.

\bibitem[SIRV18]{szarnyas2018train}
G{\'a}bor Sz{\'a}rnyas, Benedek Izs{\'o}, Istv{\'a}n R{\'a}th, and D{\'a}niel Varr{\'o}.
\newblock {The Train Benchmark: cross-technology performance evaluation of continuous model queries}.
\newblock {\em Software \& Systems Modeling}, 17:1365--1393, 2018.
\newblock \href {https://doi.org/10.1007/s10270-016-0571-8} {\path{doi:10.1007/s10270-016-0571-8}}.

\bibitem[TELW14]{taentzer2014fundamental}
Gabriele Taentzer, Claudia Ermel, Philip Langer, and Manuel Wimmer.
\newblock {A fundamental approach to model versioning based on graph modifications: from theory to implementation}.
\newblock {\em Software \& Systems Modeling}, 13(1):239--272, 2014.
\newblock \href {https://doi.org/10.1007/s10270-012-0248-x} {\path{doi:10.1007/s10270-012-0248-x}}.

\bibitem[VBH{\etalchar{+}}16]{varro2016}
D{\'a}niel Varr{\'o}, G{\'a}bor Bergmann, {\'A}bel Heged{\"u}s, {\'A}kos Horv{\'a}th, Istv{\'a}n R{\'a}th, and Zolt{\'a}n Ujhelyi.
\newblock {Road to a reactive and incremental model transformation platform: three generations of the VIATRA framework}.
\newblock {\em Software \& Systems Modeling}, 15(3):609--629, 2016.
\newblock \href {https://doi.org/10.1007/s10270-016-0530-4} {\path{doi:10.1007/s10270-016-0530-4}}.

\bibitem[VD13]{varro2013rete}
Gergely Varr{\'o} and Frederik Deckwerth.
\newblock {A Rete network construction algorithm for incremental pattern matching}.
\newblock In {\em International Conference on Theory and Practice of Model Transformations}, pages 125--140. Springer, 2013.
\newblock \href {https://doi.org/10.1007/978-3-642-38883-5_13} {\path{doi:10.1007/978-3-642-38883-5_13}}.

\end{thebibliography}


\appendix

\clearpage

\section{Technical Details} \label{app:technical_details}

We first introduce a few technical definitions required only for proofs in this appendix.

Given a local navigation structure $LNS(p)$ containing the marking-sensitive vertex input nodes $[v]^\Phi$ and $[w]^\Phi$, we call $\chi(LNS(p)) = \{[\cup]^\Phi_v, [\cup]^\Phi_w\}$ the \emph{extension points} of $LNS(p)$.

Let $(N, p)$ be a well-formed RETE net with $p$ a join node. For the localized RETE net $(N^\Phi, p^\Phi) = localize(N, p)$ with $N^\Phi = N^\Phi_{\bowtie} \cup N^\Phi_l \cup N^\Phi_r \cup RPS_l \cup RPS_r$, the extension points of $N^\Phi$ are given by $\chi(N^\Phi) = \chi(N^\Phi_l) \cup \chi(N^\Phi_r)$.

Any RETE net $(N^\Phi, p^\Phi) = localize^\Psi(Q, \psi)$ directly contains a RETE subnet $(N^\Phi_Q, p^\Phi_Q) = localize(Q)$ regardless of the form of the nested graph condition $\psi$. The extension points of $(N^\Phi, p^\Phi)$ are then given by $\chi(N^\Phi) = \chi(N^\Phi_Q)$.

We say that a marking-sensitive RETE net $(X^\Phi, p^\Phi_X)$ is a \emph{modular extension} of $(N^\Phi, p^\Phi)$ if $N^\Phi \subseteq X^\Phi$ and $\forall e \in E^{X^\Phi} : s^{X^\Phi}(e) \in V^{N^\Phi} \wedge t^{X^\Phi}(e) \in V^{N^\Phi} \Rightarrow e \in E^{N^\Phi}$.

\subsection{Theorems in Section~\ref{sec:localized_search_with_marking-sensitive_rete}}

\setcounter{thmappcount}{3}

\setcounter{thmapp}{5}
\begin{thmapp}[Matches are marked $\infty$ in query results of RETE nets localized via $localize$ iff they touch the relevant subgraph] \label{the:precision_consistent_configuration_appendix} 
Let $H$ be a graph, $H_p \subseteq H$, $(N, p)$ a well-formed RETE net, and $Q = \query{p}$. Furthermore, let $\mathcal{C}^\Phi$ be a consistent configuration for the localized RETE net $(N^\Phi, p^\Phi) = localize(N, p)$. It then holds that $\forall (m, \phi) \in \mathcal{C}^\Phi(p^\Phi) : m(Q) \cap H_p \neq \emptyset \Leftrightarrow \phi = \infty$.
\end{thmapp}

\begin{proof}
Since each extension point $x \in \chi(N^\Phi)$ is only preceded by marking assignment nodes that assign a marking other than $\infty$ and some marking-sensitive vertex input node $[v]^\Phi$ that only extracts matches $m$ with $m(v) \in H_p$, it must hold for the union node $[\cup]^\Phi$ at the top of each local navigation structure that $\forall (m, \phi) \in \mathcal{C}^\Phi([\cup]^\Phi) : m(Q) \cap H_p \neq \emptyset \Leftrightarrow \phi = \infty$.

All joins in $N^\Phi$ ultimately only combine matches from the top union nodes of local navigation structures, assigning the maximum marking in the process. Consequently, for any marking-sensitive join node $[\bowtie]^\Phi \in N^\Phi$ and a tuple $(m, \infty) \in \mathcal{C}^\Phi([\bowtie]^\Phi)$, there must be a tuple $(m', \infty) \in \mathcal{C}^\Phi([\cup]^\Phi)$ for some local navigation structure with top union node $[\cup]^\Phi$ with $Q' = query([\cup]^\Phi)$ such that $m(Q') = m'(Q')$. Thus, it must hold that $m(Q) \cap H_p \neq \emptyset$.

Conversely, by the semantics of the marking-sensitive join node, which always assigns the maximum marking of related dependent matches to the resulting matches, a marking of $\infty$ ultimately always propagates to the top of the RETE net. By Theorem~\ref{the:completeness_consistent_configuration} and because the construction of a match that touches the relevant subgraph via joins must always involve at least one join with a match that is marked $\infty$ in a local navigation structure, it thus follows that $\forall m \in \allmatches{Q}{H} : (m, \infty) \in \mathcal{C}^\Phi([\bowtie]^\Phi)$.

Since $p^\Phi$ is either a marking-sensitive join node or the top union node of a local navigation structure, it follows that $\forall (m, \infty) \in \mathcal{C}^\Phi(p^\Phi) : \exists v \in V^Q : m(v) \in V^{H_p}$.
\end{proof}

\subsection{Theorems in Section~\ref{sec:localized_rete_performance}}

\setcounter{thmapp}{9}
\begin{thmapp}[RETE net localization via $localize$ introduces only a constant factor overhead on effective configuration size] \label{the:upper_bound_configuration_size_appendix} 
Let $H$ be an edge-dominated graph, $H_p \subseteq H$, $(N, p)$ a well-formed RETE net with $Q = \query{p}$, $\mathcal{C}$ a consistent configuration for $(N, p)$ for host graph $H$, and $\mathcal{C}^\Phi$ a consistent configuration for the marking-sensitive RETE net $(N^\Phi, p^\Phi) = localize(N, p)$ for host graph $H$ and relevant subgraph $H_p$. It then holds that $\sum_{n^\Phi \in V^{N^\Phi}} \sum_{(m, \phi) \in \mathcal{C}(n^\Phi)} |m| \leq 7 \cdot |\mathcal{C}|_e$.
\end{thmapp}

\begin{proof}
Follows directly from Lemma~\ref{lem:upper_bound_configuration_size}.
\end{proof}

\subsection{Theorems in Section~\ref{sec:localized_rete_extended_queries}}

\stepcounter{thmappcount}

\setcounter{thmapp}{2}
\begin{thmapp}[Consistent configurations for RETE nets localized via $localize^\Psi$ yield correct query results under the relevant subgraph] \label{the:localized_rete_ngcs_correctness_appendix} 
Let $H$ be a graph, $H_p \subseteq H$, and $(Q, \psi)$ an extended graph query. Furthermore, let $\mathcal{C}^\Phi$ be a consistent configuration for the localized RETE net $(N^\Phi, p^\Phi) = localize^\Psi(Q, \psi)$. The set of matches given by the stripped result set $\resultsstripped{p^\Phi}{\mathcal{C}^\Phi}$ is then correct under $H_p$.
\end{thmapp}

\begin{proof}
Regardless of what $\psi$ looks like, there is a matching marking-sensitive vertex input node connected to each $x \in \chi(N^\Phi)$. Because there is at least one extension point $x \in \chi(N^\Phi)$ with $query(x) = (\{v\}, \emptyset, \emptyset, \emptyset)$ for each $v \in V^Q$, it follows from Lemma~\ref{lem:localized_rete_condition_completeness} that $\forall m \in \resultsstripped{p^\Phi}{\mathcal{C}^\Phi} : m \models \psi \wedge \forall m \in \allmatches{Q}{H}: (m' \models \psi \wedge m'(Q) \cap H_p \neq \emptyset) \Rightarrow m' \in \resultsstripped{p^\Phi}{\mathcal{C}^\Phi}$. Thus, the theorem holds.
\end{proof}

\begin{thmapp}[Execution of localized RETE nets via $order^{\Psi}$ yields consistent configurations] \label{the:localized_rete_ngcs_execution_appendix} 
Let $H$ be a graph, $H_p \subseteq H$, $(Q, \psi)$ a graph query, and $\mathcal{C}^\Phi_0$ an arbitrary starting configuration for the marking-sensitive RETE net $(N^\Phi, p^\Phi) = localize^{\Psi}(Q, \psi)$. Executing $(N^\Phi, p^\Phi)$ via $O = order^{\Psi}(N^\Phi)$ then yields a consistent configuration $\mathcal{C}^\Phi = execute(O, N^\Phi, H, H_p, \mathcal{C}^\Phi_0)$.
\end{thmapp}

\begin{proof}
Follows directly from Lemma~\ref{lem:localized_rete_condition_execution}.
\end{proof}

\setcounter{thmapp}{5}
\begin{thmapp}[RETE net localization via $localize^\Psi$ introduces only a constant factor overhead on effective configuration size] \label{the:upper_bound_configuration_size_ngc_appendix}
Let $H$ be an edge-dominated graph, $H_p \subseteq H$, $(N, p)$ a RETE net created via the procedure described in \cite{barkowsky2023host} for the extended graph query $(Q, \psi)$, $\mathcal{C}$ a consistent configuration for $(N, p)$ for host graph $H$, and $\mathcal{C}^\Phi$ a consistent configuration for the marking-sensitive RETE net $(N^\Phi, p^\Phi) = localize^\Psi(Q, \psi)$ corresponding to $(N, p)$ for host graph $H$ and relevant subgraph $H_p$. It then holds that $\sum_{n^\Phi \in V^{N^\Phi}} \sum_{(m, \phi) \in \mathcal{C}^\Phi(n^\Phi)} |m| \leq 7 \cdot |\mathcal{C}|_e$.
\end{thmapp}

\begin{proof}
We show the correctness of the theorem via structural induction over the nested graph condition $\psi$.

In the base case where $\psi = true$, it holds that $localize^\Psi(Q, \psi) = localize(N, p)$. From Theorem~\ref{the:upper_bound_configuration_size}, it then immediately follows that $\sum_{n^\Phi \in V^{N^\Phi}} \sum_{(m, \phi) \in \mathcal{C}^\Phi(n^\Phi)} |m| \leq 7 \cdot |\mathcal{C}|_e$.

We now proceed to showing that, under the induction hypothesis that the lemma holds for a nested graph condition of nesting depth $d$, the lemma holds for any nested graph condition of nesting depth $d+1$.

For a nested condition of the form $\psi = \exists (a: Q \rightarrow Q', \psi')$, $N^\Phi$ consists of the localized RETE net for $Q$, $(N^\Phi_Q, p^\Phi_Q) = localize(Q)$, the RETE net $(N^\Phi_{(Q', \psi')}, p^\Phi_{(Q', \psi')}) = localize^\Psi(Q', \psi')$, a request projection structure $RPS^{\infty}_l = RPS^{\infty}(p^\Phi_Q, N^\Phi_{(Q', \psi')})$, and a marking-sensitive semi-join node $p^\Phi = [\ltimes]^\Phi$.

By the construction described in \cite{barkowsky2023host}, $(N, p)$ then consists of a RETE net for the plain graph query $Q$, $(N_Q, p_Q)$, the RETE net for the extended graph query $(Q', \psi')$, $(N', p')$, and a semi-join $p = [\ltimes]$ with dependencies $p_Q$ and $p'$.

Lemma~\ref{lem:upper_bound_configuration_size} entails that $\sum_{n^\Phi \in V^{N^\Phi_Q}} \sum_{(m, \phi) \in \mathcal{C}^\Phi(n^\Phi)} |m| \leq 7 \sum_{n \in V^{N_Q} \setminus \{p_Q\}} \sum_{m \in \mathcal{C}(n)} |m| + 5 \cdot \sum_{m \in \mathcal{C}|_{\{p_Q\}}} |m|$, since $localize(Q) = localize(N_Q, p_Q)$.

In the worst case, the extension points of $(N^\Phi, p^\Phi)$ already contain all possible single-vertex matches into $H$ even without the additional input via $RPS^{\infty}_l$. Therefore, by the induction hypothesis, it must hold that $\sum_{n^\Phi \in V^{N^\Phi_{(Q', \psi')}}} \sum_{(m, \phi) \in \mathcal{C}^\Phi(n^\Phi)} |m| \leq 7 \cdot |\mathcal{C}|_{V^{N'}}|_e$.

From Lemma 12 in \cite{preprint}, it follows that $\sum_{n^\Phi_Q \in V^{RPS^{\infty}_l}} \sum_{(m, \phi) \in \mathcal{C}^\Phi(n^\Phi_Q)} |m| \leq 2 \cdot |\mathcal{C}|_{\{p_Q\}}|_e$. Finally, from Theorem~\ref{the:localized_rete_ngcs_correctness} and Lemma 12 in \cite{preprint}, it follows that $\sum_{(m, \phi) \in \mathcal{C}^\Phi(p^\Phi_{(Q, \psi)}} |m| \leq |\mathcal{C}|_{\{[\ltimes]\}}|_e$.

Consequently, it must hold that $\sum_{n^\Phi \in V^{N^\Phi}} \sum_{(m, \phi) \in \mathcal{C}^\Phi(n^\Phi)} |m| \leq 7 \cdot |\mathcal{C}|_{V^{N_Q} \setminus \{p_Q\}}|_e + 7 \cdot |\mathcal{C}|_{\{p_Q\}}|_e + 7 \cdot |\mathcal{C}|_{V^{N'}}|_e + |\mathcal{C}|_{\{[\ltimes]\}}|_e$ and thus $\sum_{n^\Phi \in V^{N^\Phi}} \sum_{(m, \phi) \in \mathcal{C}^\Phi(n^\Phi)} |m| \leq 7 \cdot |\mathcal{C}|_e$.

For a nested condition of the form $\psi = \neg\psi'$, $N^\Phi$ consists of the localized RETE net for the plain pattern $Q$, $(N^\Phi_Q, p^\Phi_Q) = localize(Q)$, the localized RETE net $(N^\Phi_{(Q, \psi')}, p^\Phi_{(Q, \psi')}) = localize^\Psi(Q, \psi')$, a request projection structure $RPS^{\infty}_l = RPS^{\infty}(p^\Phi_Q, N^\Phi_{(Q, \psi')})$, and marking-sensitive anti-join node $p^\Phi = [\rhd]^\Phi$.

By the construction described in \cite{barkowsky2023host}, $(N, p)$ then consists of a RETE net for the plain graph query $Q$, $(N_Q, p_Q)$, the RETE net for the extended graph query $(Q, \psi')$, $(N', p')$, and an anti-join $p = [\rhd]$ with dependencies $p_Q$ and $p'$.

By analogous argumentation as for the existential case, it must then also hold that $\sum_{n^\Phi \in V^{N^\Phi}} \sum_{(m, \phi) \in \mathcal{C}^\Phi(n^\Phi)} |m| \leq 7 \cdot |\mathcal{C}|_e$.

For a nested condition of the form $\psi = \psi_1 \wedge \psi_2$, $N^\Phi$ consists of
the localized RETE net for the plain pattern $Q$, $(N^\Phi_Q, p^\Phi_Q) = localize(Q)$,
the localized RETE net $(N^\Phi_{(Q, \psi_1)}, p^\Phi_{(Q, \psi_1)}) = localize^\Psi(Q, \psi_1)$, a request projection structure $RPS^{\infty}_1 = RPS^{\infty}(p^\Phi_Q, N^\Phi_{(Q, \psi_1)})$, a marking-sensitive semi-join node $[\ltimes]^\Phi_1$,
the localized RETE net $(N^\Phi_{(Q, \psi_2)}, p^\Phi_{(Q, \psi_2)}) = localize^\Psi(Q, \psi_2)$, a request projection structure $RPS^{\infty}_2 = RPS^{\infty}([\ltimes]^\Phi_1, N^\Phi_{(Q, \psi_2)})$, and a marking-sensitive semi-join node $p^\Phi = [\ltimes]^\Phi_2$.

By the construction described in \cite{barkowsky2023host}, $(N, p)$ then consists of a RETE net for the plain graph query $Q$, $(N_Q, p_Q)$, the RETE net for the extended graph query $(Q, \psi_1)$, $(N_1, p_1)$, the RETE net for the extended graph query $(Q, \psi_2)$, $(N_2, p_2)$, a semi-join $[\ltimes]_1$ with dependencies $p_Q$ and $p_1$, and a semi-join $p = [\ltimes]_2$ with dependencies $[\ltimes]_1$ and $p_2$.

As in the existential case, it then also holds that $\sum_{n^\Phi \in V^{N^\Phi_Q}} \sum_{(m, \phi) \in \mathcal{C}^\Phi(n^\Phi)} |m| \leq 7 \sum_{n \in V^{N_Q} \setminus \{p_Q\}} \sum_{m \in \mathcal{C}(n)} |m| + 5 \cdot \sum_{m \in \mathcal{C}(p_Q)} |m|$. By the induction hypothesis, it holds that $\sum_{n^\Phi \in V^{N^\Phi_{(Q, \psi_1)}}} \sum_{(m, \phi) \in \mathcal{C}^\Phi(n^\Phi)} |m| \leq 7 \cdot |\mathcal{C}|_{V^{N_1}}|_e$ and $\sum_{n^\Phi \in V^{N^\Phi_{(Q, \psi_2)}}} \sum_{(m, \phi) \in \mathcal{C}^\Phi(n^\Phi)} |m| \leq 7 \cdot |\mathcal{C}|_{V^{N_2}}|_e$.

From Lemma 12 in \cite{preprint}, it follows that $\sum_{n^\Phi_Q \in V^{RPS^{\infty}_1}} \sum_{(m, \phi) \in \mathcal{C}^\Phi(n^\Phi_Q)} |m| \leq 2 \cdot |\mathcal{C}|_{\{p_Q\}}|_e$.

We furthermore know by Theorem~\ref{the:localized_rete_ngcs_correctness} that $\mathcal{C}^\Phi([\ltimes]^\Phi_1)$ can only contain matches for $Q$ that satisfy $\phi_1$ and $\mathcal{C}([\ltimes]_1)$ contains all such matches. From Lemma 12 in \cite{preprint}, it thus follows that $\sum_{(m, \phi) \in \mathcal{C}^\Phi([\ltimes]^\Phi_1)} |m| \leq |\mathcal{C}|_{\{[\ltimes]_1\}}|_e$. Consequently, it must hold that $\sum_{n^\Phi_Q \in V^{RPS^{\infty}_2}} \sum_{(m, \phi) \in \mathcal{C}^\Phi(n^\Phi_Q)} |m| \leq 2 \cdot |\mathcal{C}|_{\{[\ltimes]_1\}}|_e$. By Theorem~\ref{the:localized_rete_ngcs_correctness} and Lemma 12 in \cite{preprint}, it also follows that $\sum_{(m, \phi) \in \mathcal{C}^\Phi([\ltimes]^\Phi_2)} |m| \leq |\mathcal{C}|_{\{[\ltimes]_2\}}|_e$.

We thus know that $\sum_{n^\Phi \in V^{N^\Phi}} \sum_{(m, \phi) \in \mathcal{C}^\Phi(n^\Phi)} |m| \leq 7 \cdot |\mathcal{C}|_{V^{N_Q} \setminus \{p_Q\}}|_e + 7 \cdot |\mathcal{C}|_{\{p_Q\}}|_e + 7 \cdot |\mathcal{C}|_{V^{N_1}}|_e + 7 \cdot |\mathcal{C}|_{V^{N_2}}|_e + 3 \cdot |\mathcal{C}|_{\{[\ltimes]_1\}}|_e + |\mathcal{C}|_{\{[\ltimes]_2\}}|_e$ and thus $\sum_{n^\Phi \in V^{N^\Phi}} \sum_{(m, \phi) \in \mathcal{C}^\Phi(n^\Phi)} |m| \leq 7 \cdot |\mathcal{C}|_e$.

From the correctness of base case and induction step then follows the correctness of the theorem.
\end{proof}

\setcounter{thmapp}{7}
\begin{thmapp}[Execution time overhead introduced by RETE net localization via $localize^\Psi$ depends on query graph size compared to average match size]
Let $H$ be an edge-dominated graph, $H_p \subseteq H$, $(N, p)$ a RETE net created via the procedure described in \cite{barkowsky2023host} for the extended graph query $(Q, \psi)$, $\mathcal{C}$ a consistent configuration for $(N, p)$ for host graph $H$, and $\mathcal{C}^\Phi_0$ the empty configuration for the marking-sensitive RETE net $(N^\Phi, p^\Phi) = localize^\Psi(Q, \psi)$ corresponding to $(N, p)$. Executing $(N^\Phi, p^\Phi)$ via $execute(order^\Psi(N^\Phi), N^\Phi, H, H_p, \mathcal{C}^\Phi_0)$ then takes $O(T \cdot |Q_{max}|)$ steps, with $T = \sum_{n \in V^N} |\mathcal{C}(n)|$ and $|Q_{max}|$ the size of the largest graph among $Q$ and graphs in $\psi$.
\end{thmapp}

\begin{proof}
By the construction of $order^\Psi(N^\Phi)$, each node in $N^\Phi$ is either executed only once and none of its dependencies are executed after it or is part of some subnet $(N^\Phi_{Q'}, p^\Phi_{Q'}) = localize(N^{Q'}, p^{Q'})$ for some subnet $(N^{Q'}, p^{Q'})$ of $(N, p)$ and only executed as part of a single execution of $order(N^\Phi_{Q'})$.

By Theorem~\ref{the:complexity_time_localized}, the execution of $order(N^\Phi_{Q'})$ for $(N^\Phi_{Q'}, p^\Phi_{Q'})$ in isolation is in $O(T_{Q'} \cdot |Q'|)$, with $T_{Q'} = \sum_{n \in V^{N^{Q'}}} |\mathcal{C}(n)|$. Each subnet $(N^\Phi_{Q'}, p^\Phi_{Q'})$ has at most one union node in one of its local navigation structures that has the marking assignment node of a request projection structure inserted by $localize^\Psi$ as an additional dependency. By Lemma 17 in \cite{preprint}, the time for executing all nodes in $N^\Phi$ that are part of such a subnet $(N^\Phi_{Q'}, p^\Phi_{Q'})$ is in $O((T^\Phi_Q + T^\Phi_{RPS}) \cdot Q_{max})$, with $T^\Phi_Q = \sum_{n \in N^\Phi_Q} |\mathcal{C}(n)|$ and $T^\Phi_{RPS} = \sum_{n \in N^\Phi_{RPS}} |\mathcal{C}(n)|$, where $N^\Phi_Q$ denotes the set of all RETE nodes in such subnets $(N^\Phi_{Q'}, p^\Phi_{Q'})$ and $N^\Phi_{RPS}$ denotes the set of all RETE nodes in request projection structures inserted by $localize^\Psi$.

All remaining nodes in $(N^\Phi, p^\Phi)$ are then marking-sensitive semi-joins, anti-joins or projections, marking assignments, or marking filters that are only executed once. Moreover, each node in $(N^\Phi, p^\Phi)$ can be the dependency of at most one of these semi-joins or anti-joins and at most one of these marking filters. By Lemmata 16, 18, and 19 in \cite{preprint} and Lemmata~\ref{lem:execution_time_semi_join_batch} and~\ref{lem:execution_time_anti_join_batch}, the execution of these nodes is thus in $O((T^\Phi_\Psi + T^\Phi_{RPS} + T^\Phi_Q) \cdot Q_{max})$, with $T^\Phi_\Psi = \sum_{n \in N^\Phi_\Psi} |\mathcal{C}(n)|$, where $N^\Phi_\Psi$ denotes the set of all semi-join and anti-join nodes inserted by $localize^\Psi$.

Since $T^\Phi_\Psi + T^\Phi_{RPS} + T^\Phi_Q \leq 7 \cdot T$ by Lemma~\ref{lem:localized_rete_ngc_match_number}, it then follows that executing $(N^\Phi, p^\Phi)$ via $execute(order^\Psi(N^\Phi), N^\Phi, H, H_p, \mathcal{C}^\Phi_0)$ takes $O(T \cdot |Q_{max}|)$ steps.
\end{proof}

\subsection{Theorems in Section~\ref{sec:localized_rete_satisfaction_changes}}

\setcounter{thmapp}{10}
\begin{thmapp}[Subgraph-restricted graph modifications can only change NGC satisfaction for subgraph satisfaction dependent matches] \label{the:subgraph_satisfaction_dependence_necessary_appendix}
Let $(Q, \psi)$ be an extended graph query and $H \xleftarrow{f} K \xrightarrow{g} H'$ a subgraph-restricted graph modification from host graph $H$ with relevant subgraph $H_p$ into the modified host graph $H'$ with relevant subgraph $H'_p$. It then holds for any graph morphisms $m_K : Q \rightarrow K$ that $(f \circ m_K \models \psi \wedge g \circ m_K \not\models \psi) \vee (f \circ m_K \not\models \psi \wedge g \circ m_K \models \psi) \Rightarrow f \circ m_K \text{ is subgraph satisfaction dependent } \text{ or } g \circ m_K \text{ is subgraph satisfaction dependent}$.
\end{thmapp}

\begin{proof}
We show the correctness of the theorem via structural induction over the nested graph condition $\psi$.

In the base case of $\psi = \texttt{true}$, the theorem is trivially satisfied, since any match always satisfies $\texttt{true}$.

We now proceed to showing that, under the induction hypothesis that the lemma holds for a nested graph condition of nesting depth $d$, the lemma holds for any nested graph condition of nesting depth $d+1$. Let therefore $m_K : Q \rightarrow K$ be a graph morphism such that $(f \circ m_K \models \psi \wedge g \circ m_K \not\models \psi) \vee (f \circ m_K \not\models \psi \wedge g \circ m_K \models \psi)$.

For a nested condition of the form $\psi = \neg\psi'$, it follows that $(f \circ m_K \not\models \psi' \wedge g \circ m_K \models \psi') \vee (f \circ m_K \models \psi' \wedge g \circ m_K \not\models \psi')$. From the induction hypothesis and the definition of subgraph satisfaction dependence then follows the correctness of the theorem.

For a nested condition of the form $\psi = \psi_1 \wedge \psi_2$, it follows that $(f \circ m_K \models \psi_1 \wedge g \circ m_K \not\models \psi_1) \vee (f \circ m_K \not\models \psi_1 \wedge g \circ m_K \models \psi_1)$ or $(f \circ m_K \models \psi_2 \wedge g \circ m_K \not\models \psi_2) \vee (f \circ m_K \not\models \psi_2 \wedge g \circ m_K \models \psi_2)$. From the induction hypothesis and the definition of subgraph satisfaction dependence then follows the correctness of the theorem.

For a nested condition of the form $\psi = \exists (a: Q \rightarrow Q', \psi')$, it must hold that (1) $\exists m_a \in \allmatches{Q'}{H} : f \circ m_K = m_a \circ a \wedge m_a \models \psi' \wedge \nexists m_a' \in \allmatches{Q'}{H'} : g \circ m_K = m_a' \circ a \wedge m_a' \models \psi'$ or (2) $\nexists m_a \in \allmatches{Q'}{H} : f \circ m_K = m_a \circ a \wedge m_a \models \psi' \wedge \exists m_a' \in \allmatches{Q'}{H'} : g \circ m_K = m_a' \circ a \wedge m_a' \models \psi'$.

Case 1: Let $m_a : Q' \rightarrow H$ such that $f \circ m_K = m_a \circ a \wedge m_a \models \psi'$. It must then either hold that $m_a' = iso_s \circ m_a$ is (1.1) not a valid graph morphism from $Q'$ into $H'$ or (1.2) $m_a' \not\models \psi'$. In case 1.1, since $iso_s$ is an isomorphism, it then follows that $m_a(Q') \cap H_p \neq \emptyset$ and thus $f \circ m_K$ is subgraph satisfaction dependent. In case 1.2, it follows that there must be some graph morphism $m_{K_a}$ such that $m_a = f \circ m_{K_a}$ and $m_a' = g \circ m_{K_a}$. Consequently, $m_a$ must be subgraph satisfaction dependent for $\psi'$ or $m_a'$ must be subgraph satisfaction dependent for $\psi'$ by the induction hypothesis. Since in this case, $g \circ m_K = m_a' \circ a$, it follows that $f \circ m_K \text{ is subgraph satisfaction dependent } \text{ or } g \circ m_K \text{ is subgraph satisfaction dependent}$.

Case 2: Argumentation works analogously to Case 1.

In the case where $a$ is a partial graph morphism from some subgraph $Q_p \subseteq Q$ into $Q'$, the argumentation works analogously.

From the correctness of the base case and the induction step follows the correctness of the theorem.
\end{proof}

\begin{thmapp}[Consistent configurations for RETE nets localized via $localize^{sat}$ yield all subgraph satisfaction dependent matches] \label{the:completeness_satisfaction_dependency_appendix} 
Let $H$ be a graph, $H_p \subseteq H$, $(Q, \psi)$ an extended graph query, and $\mathcal{C}^\Phi$ a consistent configuration for the RETE net $(N^{sat}, p^{sat}) = localize^{sat}(Q, \psi)$. It then holds that $\forall m \in \allmatches{Q}{H} : m \text{ is subgraph satisfaction dependent } \Rightarrow m \in \resultsstripped{p^{sat}}{\mathcal{C}^\Phi}$.
\end{thmapp}

\begin{proof}
Follows directly from Lemma~\ref{lem:completeness_satisfaction_dependency}.
\end{proof}

\begin{thmapp}[Execution of localized RETE nets via $order^{sat}$ yields consistent configurations] \label{the:satisfaction_rete_execution_appendix}
Let $H$ be a graph, $H_p \subseteq H$, $(Q, \psi)$ a graph query, and $\mathcal{C}^\Phi_0$ an arbitrary starting configuration for the marking-sensitive RETE net $(N^{sat}, p^{sat}) = localize^{sat}(Q, \psi)$. Executing $(N^{sat}, p^{sat})$ via $O = order^{sat}(N^{sat})$ then yields a consistent configuration $\mathcal{C}^\Phi = execute(O, N^{sat}, H, H_p, \mathcal{C}^\Phi_0)$.
\end{thmapp}

\begin{proof}
Follows directly from Lemma~\ref{lem:satisfaction_rete_execution}.
\end{proof}

\setcounter{thmapp}{14}
\begin{thmapp}[For a subgraph-restricted graph modification, consistent configurations for RETE nets localized via $localize^\Delta$ yield all matches that touch the relevant subgraph in the source or are subgraph satisfaction dependent in the source or target of the modification] \label{the:ngc_delta_rete_completeness_appendix} 
Let $(Q, \psi)$ be an extended graph query and $H \xleftarrow{f} K \xrightarrow{g} H'$ a subgraph-restricted graph modification modifying the graph $H$ with relevant subgraph $H_p \subseteq H$ into the graph $H'$ with relevant subgraph $H'_p \subseteq H'$. Furthermore, let $\mathcal{C}^\Phi$ be a configuration that is consistent for $(N^\Delta, p^\Delta) = localize^\Delta(Q, \psi)$ for host graph $H$ and relevant subgraph $H_p$ and consistent for $(N^{sat'}_{(Q, \psi)}, p^{sat'}_{(Q, \psi)}) = localize^{sat}(Q, \psi)$ for host graph $H'$ and relevant subgraph $H'_p$. It must then hold that $\forall m \in \allmatches{Q}{H} : m \models \psi \wedge (m(Q) \cap H_p \neq \emptyset \text{ or } m \text{ is subgraph satisfaction dependent } \text{ or } \exists m' \in \allmatches{Q}{H'} : f \circ g^{-1} \circ m' = m \wedge m' \text{ is subgraph satisfaction dependent}) \Rightarrow m \in \resultsstripped{p^\Delta}{\mathcal{C}^\Phi}$.
\end{thmapp}

\begin{proof}
By the definition of $localize^\Delta$, $N^\Delta$ consists of $(N^\Phi_Q, p^\Phi_Q) = localize(N^Q, p^Q)$ for a regular RETE net $(N^Q, p^Q)$ with height $h$ for the plain graph query $Q$, $[\phi > h]^\Phi$, $(N^{sat}, p^{sat}) = localize^{sat}(Q, \psi)$, and $[\cup]^\Phi$. In addition, $N^\Delta$ comprises $(N^\Phi_{(Q, \psi)}, p^\Phi_{(Q, \psi)}) = localize^{\Psi}(Q, \psi)$, $RPS^{\infty}_l = RPS^{\infty}([\cup]^\Phi, N^\Phi_{(Q, \psi)})$, $[\ltimes]^\Phi$, and $[\xrightarrow{f \circ g^{-1}}]^\Phi$.

Let $m \in \allmatches{Q}{H}$ with $m \models \psi \wedge (m(Q) \cap H_p \neq \emptyset \text{ or } m \text{ is subgraph satisfaction dependent or } \\ \exists m' \in \allmatches{Q}{H'} : f \circ g^{-1} \circ m' = m \wedge m' \text{ is subgraph satisfaction dependent})$. We can then distinguish three cases:

Case 1: If $m(Q) \cap H_p \neq \emptyset$, it follows from Lemma 1 in \cite{preprint} that there must be a tuple $(m, \infty) \in \mathcal{C}^\Phi(p^\Phi_Q)$ and thus also $(m, \infty) \in  \mathcal{C}^\Phi([\phi > h]^\Phi)$.

Case 2: If $m \text{ is subgraph satisfaction dependent}$, it follows from Theorem~\ref{the:completeness_satisfaction_dependency} that there must be a tuple $(m, \phi) \in \mathcal{C}^\Phi(p^{sat})$.

Case 3: If $\exists m' \in \allmatches{Q}{H'} : f \circ g^{-1} \circ m' = m \wedge m' \text{ is subgraph satisfaction dependent}$, it follows from Theorem~\ref{the:completeness_satisfaction_dependency} that there must be a tuple $(m', \phi) \in \mathcal{C}^\Phi(p^{sat'})$ such that $f \circ g^{-1} \circ m' = m$. From the semantics of the marking-sensitive translation node then follows that there must be a tuple $(m, \phi) \in \mathcal{C}^\Phi([\xrightarrow{f \circ g^{-1}}]^\Phi)$.

In all three cases, there must hence be a tuple $(m, \phi) \in \mathcal{C}^\Phi([\cup]^\Phi)$. By Lemma~\ref{lem:localized_rete_condition_completeness} and because $m \models \psi$, it must then hold that there is a tuple $(m, \phi') \in \mathcal{C}^\Phi(p^\Phi_{(Q, \psi)})$ and, by the semantics of the marking-sensitive semi-join, $(m, \phi) \in \resultsstripped{p^\Delta}{\mathcal{C}^\Phi}$.
\end{proof}

\begin{thmapp}[Execution of localized RETE nets via $order^\Delta$ yields consistent configurations] \label{the:ngc_delta_rete_execution_appendix}
Let $(Q, \psi)$ be an extended graph query and $H$ a host graph with relevant subgraph $H_p$ and $\mathcal{C}^\Phi_0$ an arbitrary starting configuration for the marking-sensitive RETE net $(N^\Delta, p^\Delta) = localize^\Delta(Q, \psi)$. Executing $(N^\Delta, p^\Delta)$ via $O = order^\Delta(N^\Delta)$ then yields a consistent configuration $\mathcal{C}^\Phi = execute(O, N^\Delta, H, H_p, \mathcal{C}^\Phi_0)$.
\end{thmapp}

\begin{proof}
Follows directly from Lemma~\ref{lem:ngc_delta_rete_execution}.
\end{proof}

\setcounter{thmapp}{17}
\begin{thmapp}[Removal of matches from the results of an extended graph query caused by a subgraph-restricted graph modification can be detected using RETE nets localized via $localize^\Delta$]
Let $(Q, \psi)$ be an extended graph query and $H \xleftarrow{f} K \xrightarrow{g} H'$ a subgraph-restricted graph modification modifying the graph $H$ with relevant subgraph $H_p \subseteq H$ into the graph $H'$ with relevant subgraph $H'_p \subseteq H'$.  Furthermore, let $\mathcal{C}^\Phi$ be a configuration that is consistent for $(N^\Delta, p^\Delta) = localize^\Delta(Q, \psi)$ and $(N^{sat}_{(Q, \psi)}, p^{sat}_{(Q, \psi)}) = localize^{sat}(Q, \psi)$ for host graph $H$ and relevant subgraph $H_p$ and consistent for $(N^{\Delta'}, p^{\Delta'}) = localize^\Delta(Q, \psi)$ and $(N^{sat'}_{(Q, \psi)}, p^{sat'}_{(Q, \psi)}) = localize^{sat}(Q, \psi)$ for host graph $H'$ and relevant subgraph $H'_p$. It then holds that $\{m \in \allmatches{Q}{H} \mid m \models \psi \wedge \nexists m' \in \allmatches{Q}{H'} : m = f \circ g^{-1} \circ m' \wedge m' \models \psi\} = \resultsstripped{p^\Delta}{\mathcal{C}^\Phi} \setminus \{m \in \allmatches{Q}{H} \mid \exists m' \in \resultsstripped{p^{\Delta'}}{\mathcal{C}^{\Phi'}} : m = f \circ g^{-1} \circ m'\}$.
\end{thmapp}

\begin{proof}
By the definition of $localize^\Delta$, $N^\Delta$ consists of $(N^\Phi_Q, p^\Phi_Q) = localize(N^Q, p^Q)$ for a regular RETE net $(N^Q, p^Q)$ with height $h$ for the plain graph query $Q$, $[\phi > h]^\Phi$, $(N^{sat}, p^{sat}) = localize^{sat}(Q, \psi)$, and $[\cup]^\Phi$, which has a dependency on $p^{sat'}_{(Q, \psi)}$ from $(N^{sat'}_{(Q, \psi)}, p^{sat'}_{(Q, \psi)})$. In addition, $N^\Delta$ comprises $(N^\Phi_{(Q, \psi)}, p^\Phi_{(Q, \psi)}) = localize^{\Psi}(Q, \psi)$, $RPS^{\infty}_l = RPS^{\infty}([\cup]^\Phi, N^\Phi_{(Q, \psi)})$, $[\ltimes]^\Phi$, and $[\xrightarrow{f \circ g^{-1}}]^\Phi$.

Analogously, $N^{\Delta'}$ consists of $(N^{\Phi'}_Q, p^{\Phi'}_Q) = localize(N^Q, p^Q)$, $[\phi > h]^{\Phi'}$, $(N^{sat'}, p^{sat'}) = localize^{sat}(Q, \psi)$, and $[\cup]^{\Phi'}$ which has a dependency on $p^{sat}_{(Q, \psi)}$ from $(N^{sat}_{(Q, \psi)}, p^{sat}_{(Q, \psi)})$. In addition, $N^{\Delta'}$ comprises $(N^{\Phi'}_{(Q, \psi)}, p^{\Phi'}_{(Q, \psi)}) = localize^{\Psi}(Q, \psi)$, $RPS^{\infty'}_l = RPS^{\infty}([\cup]^{\Phi'}, N^{\Phi'}_{(Q, \psi)})$, $[\ltimes]^{\Phi'}$, and $[\xrightarrow{g \circ f^{-1}}]^{\Phi'}$.

We first show that $\{m \in \allmatches{Q}{H} \mid m \models \psi \wedge \nexists m' \in \allmatches{Q}{H'} : m = f \circ g^{-1} \circ m' \wedge m' \models \psi\} \subseteq \resultsstripped{p^\Delta}{\mathcal{C}^\Phi} \setminus \{m \in \allmatches{Q}{H} \mid \exists m' \in \resultsstripped{p^{\Delta'}}{\mathcal{C}^{\Phi'}} : m = f \circ g^{-1} \circ m'\}$. Therefore, let $m \in \allmatches{Q}{H}$ such that $m \models \psi \wedge \nexists m' \in \allmatches{Q}{H'} : m = f \circ g^{-1} \circ m' \wedge m' \models \psi$.

We can the distinguish two cases:

Case 1: If $\nexists m' \in \allmatches{Q}{H'} : m = f \circ g^{-1} \circ m'$, it follows that $m(Q) \cap H_p \neq \emptyset$, since otherwise, $m' = iso_s \circ m$ would violate the assumption that $\nexists m' \in \allmatches{Q}{H'} : m = f \circ g^{-1} \circ m'$. It must then hold by Theorem~\ref{the:completeness_consistent_configuration} that there is a tuple $(m, \infty) \in \mathcal{C}^\Phi(p^\Phi_Q)$ and thus also $m \in \resultsstripped{p^\Delta}{\mathcal{C}^\Phi}$ by Lemma~\ref{lem:localized_rete_condition_completeness}. By the assumption that $\nexists m' \in \allmatches{Q}{H'} : m = f \circ g^{-1} \circ m'$, $\resultsstripped{p^{\Delta'}}{\mathcal{C}^{\Phi'}}$ then cannot contain a match $m'$ such that $m = f \circ g^{-1} \circ m'$ and thereby, $m \in \resultsstripped{p^\Delta}{\mathcal{C}^\Phi} \setminus \{m \in \allmatches{Q}{H} \mid \exists m' \in \resultsstripped{p^{\Delta'}}{\mathcal{C}^{\Phi'}} : m = f \circ g^{-1} \circ m'\}$.

Case 2: If there exists a match $m' \in \allmatches{Q}{H'}$ such that $m = f \circ g^{-1} \circ m'$, it must hold that $m' \not\models \psi$ for this match $m'$. By Theorem~\ref{the:subgraph_satisfaction_dependence_necessary}, we then know that either $m$ or $m'$ must be subgraph satisfaction dependent. By Theorem~\ref{the:ngc_delta_rete_completeness}, it must then hold that $m \in \resultsstripped{p^\Delta}{\mathcal{C}^\Phi}$. By Lemma~\ref{lem:localized_rete_condition_completeness}, it also follows that $m' \notin \resultsstripped{p^{\Delta'}}{\mathcal{C}^{\Phi'}}$. By the injectivity of $f$ and $g$, $m'$ must furthermore be unique. Consequently, $m \in \resultsstripped{p^\Delta}{\mathcal{C}^\Phi} \setminus \{m \in \allmatches{Q}{H} \mid \exists m' \in \resultsstripped{p^{\Delta'}}{\mathcal{C}^{\Phi'}} : m = f \circ g^{-1} \circ m'\}$.

Thus, it follows that $\{m \in \allmatches{Q}{H} \mid m \models \psi \wedge \nexists m' \in \allmatches{Q}{H'} : m = f \circ g^{-1} \circ m' \wedge m' \models \psi\} \subseteq \resultsstripped{p^\Delta}{\mathcal{C}^\Phi} \setminus \{m \in \allmatches{Q}{H} \mid \exists m' \in \resultsstripped{p^{\Delta'}}{\mathcal{C}^{\Phi'}} : m = f \circ g^{-1} \circ m'\}$.

We now show that $\{m \in \allmatches{Q}{H} \mid m \models \psi \wedge \nexists m' \in \allmatches{Q}{H'} : m = f \circ g^{-1} \circ m' \wedge m' \models \psi\} \supseteq \resultsstripped{p^\Delta}{\mathcal{C}^\Phi} \setminus \{m \in \allmatches{Q}{H} \mid \exists m' \in \resultsstripped{p^{\Delta'}}{\mathcal{C}^{\Phi'}} : m = f \circ g^{-1} \circ m'\}$. Therefore, let $m \in \resultsstripped{p^\Delta}{\mathcal{C}^\Phi} \setminus \{m \in \allmatches{Q}{H} \mid \exists m' \in \resultsstripped{p^{\Delta'}}{\mathcal{C}^{\Phi'}} : m = f \circ g^{-1} \circ m'\}$.

It then directly follows from Lemma~\ref{lem:localized_rete_condition_completeness} that $m \models \psi$. We will now show that there cannot exist a corresponding match $m'$ from $Q$ into $H'$ with $m = f \circ g^{-1} \circ m'$ and $m' \models \psi$.

From the construction of $N^\Delta$, it follows that there must exist a tuple $(m, \infty) \in \mathcal{C}^\Phi(p^\Phi_Q)$, a tuple $(m, \phi) \in \mathcal{C}^\Phi(p^{sat})$, or a tuple $(m', \phi) \in \mathcal{C}^{\Phi'}(p^{sat'})$.

Case 1: If $(m, \infty) \in \mathcal{C}^\Phi(p^\Phi_Q)$, it must hold by Theorem~\ref{the:precision_consistent_configuration} that $m(Q) \cap H_p \neq \emptyset$. Since there must be some element $e$ in $m(Q) \cap H_p$ for any match $m'$ with $m = f \circ g^{-1} \circ m'$ and because $(g \circ f^{-1})(H_p) \subseteq H'_p$, it must hold that $m'(Q) \cap H'_p \neq \emptyset$. Consequently, $(m', \infty) \in \mathcal{C}^{\Phi'}(p^{\Phi'}_Q)$ by Theorem~\ref{the:completeness_consistent_configuration}. Since $m \notin \{m \in \allmatches{Q}{H} \mid \exists m' \in \resultsstripped{p^{\Delta'}}{\mathcal{C}^{\Phi'}} : m = f \circ g^{-1} \circ m'\}$ and thus $m' \notin \resultsstripped{p^{\Delta'}}{\mathcal{C}^{\Phi'}}$, it must follow from Lemma~\ref{lem:localized_rete_condition_completeness} that $m' \not\models \psi$, since otherwise, it would hold that $\{m \in \allmatches{Q}{H} \mid \exists m' \in \resultsstripped{p^{\Delta'}}{\mathcal{C}^{\Phi'}} : m = f \circ g^{-1} \circ m'\}$.

Case 2: If $(m, \phi) \in \mathcal{C}^\Phi(p^{sat})$, for any match $m'$ with $m = f \circ g^{-1} \circ m'$, it holds that $m' = g \circ f^{-1} \circ m'$ due to the injectivity of $f$ and $g$. Therefore, it must hold that there is some tuple $(m', \phi) \in \mathcal{C}^{\Phi'}([\xrightarrow{g \circ f^{-1}}]^{\Phi'})$. Hence, it must hold that $m' \not\models \psi$, since otherwise, it would hold that $\{m \in \allmatches{Q}{H} \mid \exists m' \in \resultsstripped{p^{\Delta'}}{\mathcal{C}^{\Phi'}} : m = f \circ g^{-1} \circ m'\}$.

Case 3: If $(m', \phi) \in \mathcal{C}^{\Phi'}(p^{sat'})$ for some $m'$ with $m = f \circ g^{-1} \circ m'$, it immediately follows that $m' \not\models \psi$, since otherwise, it would hold that $\{m \in \allmatches{Q}{H} \mid \exists m' \in \resultsstripped{p^{\Delta'}}{\mathcal{C}^{\Phi'}} : m = f \circ g^{-1} \circ m'\}$.

Thus, it follows that $\{m \in \allmatches{Q}{H} \mid m \models \psi \wedge \nexists m' \in \allmatches{Q}{H'} : m = f \circ g^{-1} \circ m' \wedge m' \models \psi\} \supseteq \resultsstripped{p^\Delta}{\mathcal{C}^\Phi} \setminus \{m \in \allmatches{Q}{H} \mid \exists m' \in \resultsstripped{p^{\Delta'}}{\mathcal{C}^{\Phi'}} : m = f \circ g^{-1} \circ m'\}$.

With inclusions in both directions, it holds that $\{m \in \allmatches{Q}{H} \mid m \models \psi \wedge \nexists m' \in \allmatches{Q}{H'} : m = f \circ g^{-1} \circ m' \wedge m' \models \psi\} = \resultsstripped{p^\Delta}{\mathcal{C}^\Phi} \setminus \{m \in \allmatches{Q}{H} \mid \exists m' \in \resultsstripped{p^{\Delta'}}{\mathcal{C}^{\Phi'}} : m = f \circ g^{-1} \circ m'\}$.
\end{proof}

\begin{thm}[RETE net localization via $localize^{sat}$ introduces only a constant factor overhead on effective configuration size] \label{the:upper_bound_configuration_size_satisfaction_appendix}
Let $H$ be an edge-dominated graph, $H_p \subseteq H$, $(N, p)$ a RETE net created via the procedure described in \cite{barkowsky2023host} for the extended graph query $(Q, \psi)$, $\mathcal{C}$ a consistent configuration for $(N, p)$ for host graph $H$, and $\mathcal{C}^\Phi$ a consistent configuration for the marking-sensitive RETE net $(N^{sat}, p^{sat}) = localize^{sat}(Q, \psi)$ for host graph $H$ and relevant subgraph $H_p$ corresponding to $(N, p)$. It then holds that $\sum_{n^{sat} \in V^{N^{sat}}} \sum_{(m, \phi) \in \mathcal{C}^\Phi(n^{sat})} |m| \leq 18 \cdot |\mathcal{C}|_e$.
\end{thm}

\begin{proof}
For nested queries of $(Q, \psi)$ with the form $(Q', \texttt{true})$, $localize^{sat}$ only introduces a dummy node to $N^{sat}$. The result set for any dummy node in $N^{sat}$ is always empty. Hence, the result sets of these nodes do not contribute to configuration size. Moreover, $localize^{sat}$ does not introduce any additional nodes for nested queries with the form $(Q', \neg\psi')$.

For each nested query with the form $(Q', \psi_1 \wedge \psi_2)$, $N^{sat}$ contains one marking-sensitive union node $[\cup]^\Phi$, whereas $N$ contains a RETE subnet $(N_{Q'}, p_{Q'})$ computing matches for $Q'$ as well as two semi-joins. Since the result set for $[\cup]^\Phi$ can only contain at most one tuple for each match for $Q'$ and the result set for $p^{Q'}$ contains all such matches, it follows that $\sum_{(m, \phi) \in \mathcal{C}^\Phi([\cup]^\Phi)} |m| \leq \sum_{m \in \mathcal{C}(p_{Q'})} |m|$.

For each nested query with the form $(Q', \exists(a: Q' \rightarrow Q'', \psi'))$, $N^{sat}$ contains the subnets $(N^\Phi_{Q'}, p^\Phi_{Q'}) = localize(Q')$ and $(N^\Phi_{Q''}, p^\Phi_{Q''}) = localize(Q'')$, a marking-sensitive union node $[\cup]^\Phi$, a marking-sensitive semi-join $[\ltimes]^\Phi$, as well as the request projection structure $RPS^{\infty}_l = RPS^{\infty}([\cup]^\Phi, N^\Phi_{Q'})$. $N$ contains a RETE subnet $(N_{Q'}, p_{Q'})$ computing matches for $Q'$ and a semi-join $[\ltimes]$.

By Theorem~\ref{the:upper_bound_configuration_size} and because the additional dependency on $RPS^{\infty}_l$ does not affect the worst case of a fully populated net $(N^\Phi_{Q}, p^\Phi_{Q})$, it holds that $\sum_{n^\Phi_{Q'} \in V^{N^{\Phi}_{Q'}}} \sum_{(m, \phi) \in \mathcal{C}^\Phi(n^\Phi_{Q'})} |m| \leq 7 \cdot |\mathcal{C}|_{N_{Q'}}|_e$. Also, by the semantics of the marking-sensitive semi-join and because the result set for $p^\Phi_{Q'}$ contains each match for $Q'$ at most once by Lemma 12 in \cite{preprint}, it must hold that $\sum_{(m, \phi) \in \mathcal{C}^\Phi([\ltimes]^\Phi)} |m| \leq \sum_{m \in \mathcal{C}(p_{Q'})} |m| \leq  |\mathcal{C}|_{N_{Q'}}|_e$.

Any nested query with the form $(Q', \exists(a: Q' \rightarrow Q'', \psi'))$ must have some nested query $(Q'', \exists(a: Q'' \rightarrow Q''', \psi''))$ or $(Q'', \texttt{true})$ in order for $\psi$ by the definition of nested graph conditions. In either case, $N$ has to contain a corresponding subnet $(N_{Q''}, p_{Q''})$ computing matches for $Q''$. It then holds for this subnet $N_{Q''}$ by Theorem~\ref{the:upper_bound_configuration_size} that $\sum_{n^\Phi_{Q''} \in V^{N^{\Phi}_{Q''}}} \sum_{(m, \phi) \in \mathcal{C}^\Phi(n^\Phi_{Q''})} |m| \leq 7 \cdot |\mathcal{C}|_{N_{Q''}}|_e$. It also must hold that $\sum_{(m, \phi) \in \mathcal{C}^\Phi([\cup]^\Phi)} |m| \leq \sum_{m \in \mathcal{C}(p_{Q''})} |m| \leq  |\mathcal{C}|_{N_{Q''}}|_e$ and $\sum_{n^\Phi \in V^{RPS^{\infty}_l}} \sum_{(m, \phi) \in \mathcal{C}^\Phi(n^\Phi)} |m| \leq 2 \cdot |\mathcal{C}|_{N_{Q''}}|_e$.

Since $\psi$ is a tree of subconditions, it follows that each subnet $(N_{Q''}, p_{Q''})$ must be counted again like this for at most one parent query $(Q', \exists(a: Q' \rightarrow Q'', \psi'))$. For any nested query of the form $(Q', \texttt{true})$, $\mathcal{C}^\Phi$ stores matches with a combined size of at most $10 \cdot |\mathcal{C}|_{N_{Q'}}|_e$. For any nested query of the form, $(Q', \exists(a: Q' \rightarrow Q'', \psi'))$, $\mathcal{C}^\Phi$ stores matches with a combined size of at most $(8 + 10) \cdot |\mathcal{C}|_{N_{Q'}}|_e = 18 \cdot |\mathcal{C}|_{N_{Q'}}|_e$.

In the worst case, it must thus hold that $\sum_{n^{sat} \in V^{N^{sat}}} \sum_{(m, \phi) \in \mathcal{C}^\Phi(n^{sat})} |m| \leq 18 \cdot |\mathcal{C}|_e$.
\end{proof}

\begin{cor}[RETE net localization via $localize^{sat}$ introduces only a constant factor overhead on memory consumption] \label{cor:memory_consumption_satisfaction_rete_appendix} 
Let $H$ be an edge-dominated graph, $H_p \subseteq H$, $(N, p)$ a RETE net created via the procedure described in \cite{barkowsky2023host} for the extended graph query $(Q, \psi)$, $\mathcal{C}$ a consistent configuration for $(N, p)$ for host graph $H$, and $\mathcal{C}^\Phi$ a consistent configuration for the marking-sensitive RETE net $(N^{sat}, p^{sat}) = localize^{sat}(Q, \psi)$ for host graph $H$ and relevant subgraph $H_p$ corresponding to $(N, p)$. Assuming that storing a match $m$ requires an amount of memory in $O(|m|)$ and storing an element from $\overline{\mathbb{N}}$ requires an amount of memory in $O(1)$, storing $\mathcal{C}^\Phi$ requires memory in $O(|\mathcal{C}|_e)$.
\end{cor}

\begin{proof}
Follows directly from Theorem~\ref{the:upper_bound_configuration_size_satisfaction_appendix} and the assumptions.
\end{proof}

\begin{thm}[Execution time overhead introduced by RETE localization via $localize^{sat}$ depends on query graph size compared to average match size] \label{the:execution_time_satisfaction_rete_appendix}
Let $H$ be an edge-dominated graph, $H_p \subseteq H$, $(N, p)$ a RETE net created via the procedure described in \cite{barkowsky2023host} for the extended graph query $(Q, \psi)$, $\mathcal{C}$ a consistent configuration for $(N, p)$ for host graph $H$, and $\mathcal{C}^\Phi_0$ the empty configuration for the marking-sensitive RETE net $(N^{sat}, p^{sat}) = localize^{sat}(Q, \psi)$ corresponding to $(N, p)$. Executing $(N^\Phi, p^\Phi)$ via $execute(order^\Psi(N^\Phi), N^\Phi, H, H_p, \mathcal{C}^\Phi_0)$ then takes $O(T \cdot |Q_{max}|)$ steps, with $T = \sum_{n \in V^N} |\mathcal{C}(n)|$ and $|Q_{max}|$ the size of the largest graph among $Q$ and graphs in $\psi$.
\end{thm}

\begin{proof}
By the construction of $order^{sat}(N^{sat})$, each node in $N^{sat}$ is either executed only once and none of its dependencies are executed after it or is part of some subnet $(N^\Phi_{Q'}, p^\Phi_{Q'}) = localize(N^{Q'}, p^{Q'})$ for some subnet $(N^{Q'}, p^{Q'})$ of $(N, p)$ and only executed as part of a single execution of $order(N^\Phi_{Q'})$.

By Theorem~\ref{the:complexity_time_localized}, the execution of $order(N^\Phi_{Q'})$ for $(N^\Phi_{Q'}, p^\Phi_{Q'})$ in isolation is in $O(T_{Q'} \cdot |Q'|)$, with $T_{Q'} = \sum_{n \in V^{N^{Q'}}} |\mathcal{C}(n)|$. Each subnet $(N^\Phi_{Q'}, p^\Phi_{Q'})$ has at most one union node in one of its local navigation structures that has the marking assignment node of a request projection structure inserted by $localize^{sat}$ as an additional dependency. By Lemma17 in \cite{preprint}, the time for executing all nodes in $N^{sat}$ that are part of such a subnet $(N^\Phi_{Q'}, p^\Phi_{Q'})$ is in $O((T^\Phi_Q + T^\Phi_{RPS}) \cdot Q_{max})$, with $T^\Phi_Q = \sum_{n \in N^\Phi_Q} |\mathcal{C}(n)|$ and $T^\Phi_{RPS} = \sum_{n \in N^\Phi_{RPS}} |\mathcal{C}(n)|$, where $N^\Phi_Q$ denotes the set of all RETE nodes in such subnets $(N^\Phi_{Q'}, p^\Phi_{Q'})$ and $N^\Phi_{RPS}$ denotes the set of all RETE nodes in projection structures inserted by $localize^\Psi$.

All remaining nodes in $(N^{sat}, p^{sat})$ are then marking-sensitive semi-joins, unions or projections, marking assignments, or marking filters that are only executed once. Moreover, each node in $(N^{sat}, p^{sat})$ can be the dependency of at most one of these semi-joins or unions and at most one of these marking filters. By Lemmata 16, 17, 18, and 19 in \cite{preprint} and Lemma~\ref{lem:execution_time_semi_join_batch}, the execution of these nodes is thus in $O((T^\Phi_{sat} + T^\Phi_{RPS} + T^\Phi_Q) \cdot Q_{max})$, with $T^\Phi_{sat} = \sum_{n \in N^\Phi_{sat}} |\mathcal{C}(n)|$, where $N^\Phi_{sat}$ denotes the set of all semi-join and anti-join nodes inserted by $localize^{sat}$.

Since $T^\Phi_{sat} + T^\Phi_{RPS} + T^\Phi_Q \leq 25 \cdot T$ by Lemma~\ref{lem:satisfaction_rete_match_number}, it then follows that executing $(N^{sat}, p^{sat})$ via $execute(order^{sat}(N^{sat}), N^{sat}, H, H_p, \mathcal{C}^\Phi_0)$ takes $O(T \cdot |Q_{max}|)$ steps.
\end{proof}

\subsection{Supplementary Lemmata}

\begin{lem}[RETE net localization via $localize$ introduces an overhead of at most factor 5 on effective configuration size for the production node and at most factor 7 for all other nodes] \label{lem:upper_bound_configuration_size} 
Let $H$ be an edge-dominated graph, $H_p \subseteq H$, $(N, p)$ a well-formed RETE net with $Q$ the associated query graph of $p$, $\mathcal{C}$ a consistent configuration for $(N, p)$ for host graph $H$, and $\mathcal{C}^\Phi$ a consistent configuration for the localized RETE net $(N^\Phi, p^\Phi) = localize(N, p)$ for host graph $H$ and relevant subgraph $H_p$. It then holds that $\sum_{n^\Phi \in V^{N^\Phi}} \sum_{(m, \phi) \in \mathcal{C}^\Phi(n^\Phi)} |m| \leq 7 \cdot \sum_{n \in V^N \setminus \{p\}} \sum_{m \in \mathcal{C}(n)} |m| + 5 \cdot \sum_{m \in \mathcal{C}(p)} |m|$.
\end{lem}

\begin{proof}
For each edge input node $[v \rightarrow w]$ in $N$ with associated query subgraph $Q_e$, $N^\Phi$ contains the seven nodes of the corresponding local navigation structure $LNS([v \rightarrow w])$: $[v]^\Phi$, $[w]^\Phi$, $[\cup]^\Phi_{v}$, $[\cup]^\Phi_{w}$, $[v \rightarrow_n w]^\Phi$, $[w \leftarrow_n v]^\Phi$, and $[\cup]^\Phi$. By the semantics of $[v \rightarrow w]$, it must hold that $\mathcal{C}([v \rightarrow w]) = \allmatches{Q_e}{H}$.

By the semantics of the forward and backward local navigation nodes, we know that $\resultsstripped{[v \rightarrow_n w]^\Phi}{\mathcal{C}^\Phi} \subseteq \allmatches{Q_e}{H}$ and $\resultsstripped{[w \leftarrow_n v]^\Phi}{\mathcal{C}^\Phi} \subseteq \allmatches{Q_e}{H}$ and thus, by the semantics of the marking sensitive union node, $\resultsstripped{[\cup]^\Phi}{\mathcal{C}^\Phi} \subseteq \allmatches{Q_e}{H}$.

Furthermore, by the semantics of the marking sensitive input nodes, it must hold that $\resultsstripped{[v]^\Phi}{\mathcal{C}^\Phi} \subseteq \allmatches{Q_{v}}{H}$ and $\resultsstripped{[w]^\Phi}{\mathcal{C}^\Phi} \subseteq \allmatches{Q_{w}}{H}$, where $Q_{v}$ and $Q_{w}$ are the query subgraphs associated with $v$ respectively $w$. By the semantics of the marking sensitive union nodes and the assumption regarding the common query subgraph associated with their dependencies, it must also hold that $\resultsstripped{[\cup]_{v}^\Phi}{\mathcal{C}^\Phi} \subseteq \allmatches{Q_{v}}{H}$ and $\resultsstripped{[\cup]_{w}^\Phi}{\mathcal{C}^\Phi} \subseteq \allmatches{Q_{w}}{H}$.

By the semantics of the involved marking-sensitive RETE nodes, for each node in $LNS([n_1 \rightarrow n_2])$, each match can only be associated with at most one marking in the node's current result set in $\mathcal{C}^\Phi$. Furthermore, by assumption, it holds that $|\allmatches{Q_v}{H}| \leq |\allmatches{Q_e}{H}|$ and $|\allmatches{Q_w}{H}| \leq |\allmatches{Q_e}{H}|$. Lastly, we know that $|Q_v| = |Q_w| \leq \frac{|Q_e|}{2}$.

It thus holds that $\sum_{n^\Phi \in V^{LNS([v \rightarrow w])}} \sum_{(m, \phi) \in \mathcal{C}(n^\Phi)} |m| \leq 5 \cdot \sum_{m \in \mathcal{C}([v \rightarrow w])} |m|$.

For each join node $[\bowtie]$ in $N$ with associated query subgraph $Q_{\bowtie}$ and dependencies $n_l$ and $n_r$, $N^\Phi$ contains the corresponding marking-sensitive join node $[\bowtie]^\Phi$ with dependencies $n^\Phi_l$ and $n^\Phi_r$ corresponding to $n_l$ respectively $n_r$ with associated query subgraphs $Q_l$ respectively $Q_r$, as well as the six nodes from the related request projection structures: $[\phi > h]^\Phi_l$, $[\pi_{Q_v}]^\Phi_l$, $[\phi := h]^\Phi_l$, $[\phi > h]^\Phi_r$, $[\pi_{Q_v}]^\Phi_r$, and $[\phi := h]^\Phi_r$.

By the semantics of join and marking sensitive join, we know that $\mathcal{C}(n_l) \subseteq \resultsstripped{n^\Phi_l}{\mathcal{C}^\Phi} \wedge \mathcal{C}(n_r) \subseteq \resultsstripped{n^\Phi_r}{\mathcal{C}^\Phi} \Rightarrow \mathcal{C}([\bowtie]) \subseteq \resultsstripped{[\bowtie]^\Phi}{\mathcal{C}^\Phi}$. For every edge input node $[v \rightarrow w]$ in $N$, we know that $\mathcal{C}([v \rightarrow w]) \subseteq \resultsstripped{[\cup]^\Phi}{\mathcal{C}^\Phi}$, where $[\cup]^\Phi$ is the root node of $LNS([v \rightarrow w])$. It is easy to see that thereby, it must hold that $\mathcal{C}([\bowtie]) \subseteq \resultsstripped{[\bowtie]^\Phi}{\mathcal{C}^\Phi}$. By Lemma 12 in \cite{preprint}, it follows that $\sum_{(m, \phi) \in \mathcal{C}([\bowtie]^\Phi)} |m| \leq \sum_{m \in \mathcal{C}([\bowtie])} |m|$.

$[\phi > h]^\Phi_l$ is associated with the query subgraph $Q_l$ and $[\phi > h]^\Phi_r$ is associated with query subgraph $Q_r$. $[\pi_{Q_v}]^\Phi_l$, $[\phi := h]^\Phi_l$, $[\pi_{Q_v}]^\Phi_r$, and $[\phi := h]^\Phi_r$ are all associated with a query subgraph $Q_v$ that contains only a single vertex, whereas both $Q_l$ and $Q_r$ contain at least one edge. Consequently, $|Q_v| \leq \frac{|Q_l|}{2}$ and $|Q_v| \leq \frac{|Q_r|}{2}$. Hence, it must hold that $\sum_{n^\Phi \in RPS_l} \sum_{(m, \phi) \in \mathcal{C}(n^\Phi)} |m| \leq 2 \cdot \sum_{m \in \mathcal{C}(n_l)} |m|$ and $\sum_{n^\Phi \in RPS_r} \sum_{(m, \phi) \in \mathcal{C}(n^\Phi)} |m| \leq 2 \cdot \sum_{m \in \mathcal{C}(n_r)} |m|$.

Due to the tree-like structure of $(N, p)$, the marking-sensitive union node at the top of each local navigation structure and each marking-sensitive join can be connected to at most one request projection structure. Moreover, the union node at the top of the local navigation structure or the marking-sensitive join corresponding to $p$ is not connected to any request projection structure. It thus follows that $\sum_{n^\Phi \in V^{N^\Phi}} \sum_{(m, \phi) \in \mathcal{C}^\Phi(n^\Phi)} |m| \leq 7 \cdot \sum_{n \in V^N \setminus \{p\}} \sum_{m \in \mathcal{C}(n)} |m| + 5 \cdot \sum_{m \in \mathcal{C}(p)} |m|$.
\end{proof}

\begin{lem}[RETE net localization via $localize$ introduces an overhead of factor 7 on the number of computed matches for the production node and factor 10 for all other nodes] \label{lem:localized_rete_match_number} 
Let $H$ be an edge-dominated graph, $H_p \subseteq H$, $(N, p)$ a well-formed RETE net with $Q$ the associated query graph of $p$, $\mathcal{C}$ a consistent configuration for $(N, p)$ for host graph $H$, and $\mathcal{C}^\Phi$ a consistent configuration for the localized RETE net $(N^\Phi, p^\Phi) = localize(N, p)$ for host graph $H$ and relevant subgraph $H_p$. It then holds that $\sum_{n^\Phi \in V^{N^\Phi}} |\mathcal{C}^\Phi(n^\Phi)| \leq 10 \sum_{n \in V^N \setminus \{p\}} |\mathcal{C}(n)| + 7 \times |\mathcal{C}(p)|$.
\end{lem}

\begin{proof}
For each edge input node $[v \rightarrow w]$ in $N$ with associated query subgraph $Q_e$, $N^\Phi$ contains the seven nodes of the corresponding local navigation structure $LNS([v \rightarrow w])$: $[v]^\Phi$, $[w]^\Phi$, $[\cup]^\Phi_{v}$, $[\cup]^\Phi_{w}$, $[v \rightarrow_n w]^\Phi$, $[w \leftarrow_n v]^\Phi$, and $[\cup]^\Phi$. By the semantics of $[v \rightarrow w]$, it must hold that $\mathcal{C}([v \rightarrow w]) = \allmatches{Q_e}{H}$.

By the semantics of the forward and backward local navigation nodes, we know that $\resultsstripped{[v \rightarrow_n w]^\Phi}{\mathcal{C}^\Phi} \subseteq \allmatches{Q_e}{H}$ and $\resultsstripped{[w \leftarrow_n v]^\Phi}{\mathcal{C}^\Phi} \subseteq \allmatches{Q_e}{H}$ and thus, by the semantics of the marking sensitive union node, $\resultsstripped{[\cup]^\Phi}{\mathcal{C}^\Phi} \subseteq \allmatches{Q_e}{H}$.

Furthermore, by the semantics of the marking sensitive input nodes, it must hold that $\resultsstripped{[v]^\Phi}{\mathcal{C}^\Phi} \subseteq \allmatches{Q_{v}}{H}$ and $\resultsstripped{[w]^\Phi}{\mathcal{C}^\Phi} \subseteq \allmatches{Q_{w}}{H}$, where $Q_{v}$ and $Q_{w}$ are the query subgraphs associated with $v$ respectively $w$. By the semantics of the marking sensitive union nodes and the assumption regarding the common query subgraph associated with their dependencies, it must also hold that $\resultsstripped{[\cup]_{v}^\Phi}{\mathcal{C}^\Phi} \subseteq \allmatches{Q_{v}}{H}$ and $\resultsstripped{[\cup]_{w}^\Phi}{\mathcal{C}^\Phi} \subseteq \allmatches{Q_{w}}{H}$.

By the semantics of the involved marking-sensitive RETE nodes, for each node in $LNS([n_1 \rightarrow n_2])$, each match can only be associated with at most one marking in the node's current result set in $\mathcal{C}^\Phi$. Furthermore, by assumption, it holds that $\allmatches{Q_v}{H} \leq \allmatches{Q_e}{H}$ and $\allmatches{Q_w}{H} \leq \allmatches{Q_e}{H}$.

It thus holds that $\sum_{n^\Phi \in V^{LNS([v \rightarrow w])}} |\mathcal{C}(n^\Phi)| \leq 7 \cdot |\mathcal{C}([v \rightarrow w])|$.

For each join node $[\bowtie]$ in $N$ with associated query subgraph $Q_{\bowtie}$ and dependencies $n_l$ and $n_r$, $N^\Phi$ contains the corresponding marking-sensitive join node $[\bowtie]^\Phi$ with dependencies $n^\Phi_l$ and $n^\Phi_r$ corresponding to $n_l$ respectively $n_r$ with associated query subgraphs $Q_l$ respectively $Q_r$, as well as the six nodes from the related request projection structures: $[\phi > h]^\Phi_l$, $[\pi_{Q_v}]^\Phi_l$, $[\phi > h]^\Phi_l$, $[\phi > h]^\Phi_r$, $[\pi_{Q_v}]^\Phi_r$, and $[\phi > h]^\Phi_r$.

By the semantics of join and marking sensitive join, we know that $\mathcal{C}(n_l) \subseteq \resultsstripped{n^\Phi_l}{\mathcal{C}^\Phi} \wedge \mathcal{C}(n_r) \subseteq \resultsstripped{n^\Phi_r}{\mathcal{C}^\Phi} \Rightarrow \mathcal{C}([\bowtie]) \subseteq \resultsstripped{[\bowtie]^\Phi}{\mathcal{C}^\Phi}$. For every edge input node $[v \rightarrow w]$ in $N$, we know that $\mathcal{C}([v \rightarrow w]) \subseteq \resultsstripped{[\cup]^\Phi}{\mathcal{C}^\Phi}$, where $[\cup]^\Phi$ is the root node of $LNS([v \rightarrow w])$. It is easy to see that thereby, it must hold that $\mathcal{C}([\bowtie]) \subseteq \resultsstripped{[\bowtie]^\Phi}{\mathcal{C}^\Phi}$. By Lemma 12 in \cite{preprint}, it follows that $|\mathcal{C}([\bowtie]^\Phi)| \leq |\mathcal{C}([\bowtie])|$.

It must then also hold that $\sum_{n^\Phi \in RPS_l} |\mathcal{C}(n^\Phi)| \leq 3 \cdot |\mathcal{C}(n_l)|$ and $\sum_{n^\Phi \in RPS_r} |\mathcal{C}(n^\Phi)| \leq 3 \cdot |\mathcal{C}(n_r)|$.

Due to the tree-like structure of $(N, p)$, the marking-sensitive union node at the top of each local navigation structure and each marking-sensitive join can be connected to at most one request projection structure. Moreover, the union node at the top of the local navigation structure or the marking-sensitive join corresponding to $p$ is not connected to any request projection structure. It thus follows that $\sum_{n^\Phi \in V^{N^\Phi}} |\mathcal{C}^\Phi(n^\Phi)| \leq 10 \sum_{n \in V^N \setminus \{p\}} |\mathcal{C}(n)| + 7 \times |\mathcal{C}(p)|$.
\end{proof}

\begin{lem}[Consistent configurations for modular extensions of RETE nets localized via $localize^\Psi$ yield only matches that satisfy the NGC and all admissible matches with a partial match marked $\infty$ in an extension point] \label{lem:localized_rete_condition_completeness}
Let $H$ be a graph, $(Q, \psi)$ a graph query, $(X^\Phi, p_X^\Phi)$ a modular extension of the localized RETE net $(N^\Phi, p^\Phi) = localize^\Psi(Q, \phi)$, and $\mathcal{C}$ a configuration for $(X^\Phi, p_X^\Phi)$ that is consistent for all nodes in $V^{N^\Phi}$. Furthermore, let $x \in \chi(N^\Phi)$, $Q_v = (\{v\}, \emptyset, \emptyset, \emptyset)$ the query subgraph associated with $x$, and $(m', \infty) \in \mathcal{C}(x)$ for some match $m' : Q_v \rightarrow H$. It then holds that $\forall (m, \phi) \in \mathcal{C}(p^\Phi) : m \models \psi \wedge \forall m \in \allmatches{Q}{H}: m(v) = m'(v) \wedge m \models \psi \Rightarrow (m, \infty) \in \mathcal{C}(p^\Phi)$.
\end{lem}

\begin{proof}
We prove the lemma via structural induction over the nested graph condition $\psi$.

In the base case, that is, for a nested condition of the form $\psi = true$, satisfaction of the Lemma follows directly from Lemma 1 in the preprint version of our conference paper \cite{preprint}.

We now proceed to showing that, under the induction hypothesis that the lemma holds for a nested graph condition of nesting depth $d$, the lemma holds for any nested graph condition of nesting depth $d+1$.

For a nested condition of the form $\psi = \exists (a: Q \rightarrow Q', \psi')$, $N^\Phi$ consists of the localized RETE net for $Q$, $(N^\Phi_Q, p^\Phi_Q) = localize(Q)$, the localized RETE net $(N^\Phi_{(Q', \psi')}, p^\Phi_{(Q', \psi')}) = localize^\Psi(Q', \psi')$, a request projection structure $RPS^{\infty}_l = RPS^{\infty}(p^\Phi_Q, N^\Phi_{(Q', \psi')})$, and a marking-sensitive semi-join node $p^\Phi = [\ltimes]^\Phi$ with left dependency $p^\Phi_Q$ and right dependency $p^\Phi_{(Q', \psi')}$.

In the case where $a$ is a partial graph morphism from some subgraph $Q_p \subseteq Q$ into $Q'$, the argumentation works analogously.

By Lemma 1 in \cite{preprint}, it immediately follows that $\forall m \in \allmatches{Q}{H}: m(v) = m'(v) \Rightarrow (m, \infty) \in \mathcal{C}(p^\Phi_Q)$.  Assuming that the Lemma holds for localized RETE net $(N^\Phi_{(Q', \psi')}, p^\Phi_{(Q', \psi')}) = localize^\Psi(Q', \psi')$, by the definition of $RPS^{\infty}_l$ it must also hold that $\forall (m, \infty) \in \mathcal{C}(p^\Phi_Q), m' \in \allmatches{Q'}{H}: m = m' \circ a \wedge m' \models \psi' \Rightarrow (m', \infty) \in \mathcal{C}(p^\Phi_{(Q', \psi')})$ and $\forall (m', \phi') \in \mathcal{C}(p^\Phi_{(Q', \psi')}) : m' \models \psi'$. From the semantics of the marking-sensitive semi-join then follows the satisfaction of the Lemma.

For a nested condition of the form $\psi = \neg\psi'$, $N^\Phi$ consists of the localized RETE net for the plain pattern $Q$, $(N^\Phi_Q, p^\Phi_Q) = localize(Q)$, the localized RETE net $(N^\Phi_{(Q, \psi')}, p^\Phi_{(Q, \psi')}) = localize^\Psi(Q, \psi')$, a request projection structure $RPS^{\infty}_l = RPS(p^\Phi_Q, N^\Phi_{(Q, \psi')})$, and a marking-sensitive anti-join node $p^\Phi = [\rhd]^\Phi$ with left dependency $p^\Phi_Q$ and right dependency $p^\Phi_{(Q, \psi')}$.

By Lemma 1 in \cite{preprint}, it immediately follows that $\forall m \in \allmatches{Q}{H}: m(v) = m'(v) \Rightarrow (m, \infty) \in \mathcal{C}(p^\Phi_Q)$.  Assuming that the Lemma holds for localized RETE net $(N^\Phi_{(Q, \psi')}, p^\Phi_{(Q, \psi')}) = localize^\Psi(Q, \psi')$, by the definition of $RPS^{\infty}_l$ it must also hold that $\forall (m, \infty) \in \mathcal{C}(p^\Phi_Q) : m \models \psi' \Rightarrow (m, \infty) \in \mathcal{C}(p^\Phi_{(Q, \psi')})$ and $\forall (m, \phi) \in \mathcal{C}(p^\Phi_{(Q, \psi')}) : m \models \psi'$. From the semantics of the marking-sensitive anti-join then follows the satisfaction of the Lemma.

For a nested condition of the form $\psi = \psi_1 \wedge \psi_2$, $N^\Phi$ consists of
the localized RETE net for the plain pattern $Q$, $(N^\Phi_Q, p^\Phi_Q) = localize(Q)$,
the localized RETE net $(N^\Phi_{(Q, \psi_1)}, p^\Phi_{(Q, \psi_1)}) = localize^\Psi(Q, \psi_1)$, a request projection structure $RPS^{\infty}_1 = RPS^{\infty}(p^\Phi_Q, N^\Phi_{(Q, \psi_1)})$, a marking-sensitive semi-join node $[\ltimes]^\Phi_1$ with left dependency $p^\Phi_Q$ and right dependency $p^\Phi_{(Q', \psi')}$,
the localized RETE net $(N^\Phi_{(Q, \psi_2)}, p^\Phi_{(Q, \psi_2)}) = localize^\Psi(Q, \psi_2)$, a request projection structure $RPS^{\infty}_2 = RPS^{\infty}([\ltimes]^\Phi_1, N^\Phi_{(Q, \psi_2)})$, and a marking-sensitive semi-join node $p^\Phi = [\ltimes]^\Phi_2$ with left dependency $[\ltimes]^\Phi_1$ and right dependency $p^\Phi_{(Q, \psi_2)}$.

By Lemma 1 in \cite{preprint}, it immediately follows that $\forall m \in \allmatches{Q}{H}: m(v) = m'(v) \Rightarrow (m, \infty) \in \mathcal{C}(p^\Phi_Q)$.  Assuming that the lemma holds for localized RETE net $(N^\Phi_{(Q, \psi_1)}, p^\Phi_{(Q, \psi_1)}) = localize^\Psi(Q, \psi_1)$, by the definition of $RPS^{\infty}_1$ it must also hold that $\forall (m, \infty) \in \mathcal{C}(p^\Phi_Q): m \models \psi_1 \Rightarrow (m, \infty) \in \mathcal{C}(p^\Phi_{(Q, \psi_1)})$ and $\forall (m_1, \phi_1) \in \mathcal{C}(p^\Phi_{(Q, \psi_1)}) : m_1 \models \psi_1$. Assuming the lemma also holds for $(N^\Phi_{(Q, \psi_2)}, p^\Phi_{(Q, \psi_2)}) = localize^\Psi(Q, \psi_2)$, from the semantics of the marking-sensitive semi-join and the definition of $RPS^{\infty}_2$ it then also follows that $\forall (m, \infty) \in \mathcal{C}(p^\Phi_Q): m \models \psi_1 \wedge m \models \psi_2 \Rightarrow (m, \infty) \in \mathcal{C}(p^\Phi_{(Q, \psi_2)})$ and $\forall (m_2, \phi_2) \in \mathcal{C}(p^\Phi_{(Q, \psi_2)}) : m_2 \models \psi_2$. From the semantics of the marking-sensitive semi-join then follows the satisfaction of the Lemma.
\end{proof}

\begin{lem}[Execution of modular extensions of localized RETE nets via $order^\Psi$ yields consistent configurations for the base RETE net] \label{lem:localized_rete_condition_execution}
Let $H$ be a graph, $H_p \subseteq H$, $(Q, \psi)$ a graph query, and $(X^\Phi, p_X^\Phi)$ a modular extension of the marking-sensitive RETE net $(N^\Phi, p^\Phi) = localize^\Psi(Q, \psi)$. Furthermore, let $\mathcal{C}^\Phi_0$ be an arbitrary configuration for $(X^\Phi, p_X^\Phi)$. Executing $(X^\Phi, p_X^\Phi)$ via $O = order^\Psi(N^\Phi)$ then yields a configuration $\mathcal{C}^\Phi = execute(O, N^\Phi, H, H_p, \mathcal{C}^\Phi_0)$ that is consistent for all nodes in $N^\Phi$.
\end{lem}

\begin{proof}
We show the correctness of the Lemma via structural induction over the nested graph condition $\psi$.

In the base case, where $\psi = true$, $O$ is given by $O = localize(Q)$. By Lemma 13 in \cite{preprint}, it follows directly that the lemma holds.

We now proceed to showing that, under the induction hypothesis that the lemma holds for a nested graph condition of nesting depth $d$, the lemma holds for any nested graph condition of nesting depth $d+1$.

For a nested condition of the form $\psi = \exists (a: Q \rightarrow Q', \psi')$, $O$ is given by $O = order(N^\Phi_Q) \circ toposort^{-1}(RPS^{\infty}_l) \circ order^\Psi(N^\Phi_{(Q', \psi')}) \circ [\ltimes]^\Phi$. By Lemma 13 in \cite{preprint}, it follows that $\mathcal{C}^\Phi_1 = execute(order(N^\Phi_Q), X^\Phi, H, H_p, \mathcal{C}^\Phi_0)$ is consistent for all nodes in $N^\Phi_Q$. Since no node in $N^\Phi_Q$ depends on a node in $RPS^{\infty}_l$, $\mathcal{C}^\Phi_2 = execute(toposort^{-1}(RPS^{\infty}_l), X^\Phi, H, H_p, \mathcal{C}^\Phi_1)$ must then be consistent for all nodes in $N^\Phi_Q \cup RPS^{\infty}_l$. By the same argument and the induction hypothesis, $\mathcal{C}^\Phi_3 = execute(order^\Psi(N^\Phi_{(Q', \psi')}), X^\Phi, H, H_p, \mathcal{C}^\Phi_2)$ is consistent for all nodes in $N^\Phi_Q \cup RPS^{\infty}_l \cup N^\Phi_{(Q', \psi')}$ and $\mathcal{C}^\Phi_4 = execute([\ltimes]^\Phi, X^\Phi, H, H_p, \mathcal{C}^\Phi_3) = execute(O, N^\Phi, H, H_p, \mathcal{C}^\Phi_0)$ is consistent for all nodes in $N^\Phi$.

By analogous argumentation, the lemma then also holds for a nested condition of the form $\psi = \neg\psi'$.

For a nested condition of the form $\psi = \psi_1 \wedge \psi_2$, $O$ is given by $order^\Psi(N^\Phi) = order(N^\Phi_Q) \circ toposort^{-1}(RPS^{\infty}_1) \circ order^\Psi(N^\Phi_{(Q, \psi_1)}) \circ [\ltimes]^\Phi_1 \circ toposort^{-1}(RPS^{\infty}_2) \circ order^\Psi(N^\Phi_{(Q, \psi_2)}) \circ [\ltimes]^\Phi_2$. By the same argumentation as for the case where $\psi = \exists (a: Q \rightarrow Q', \psi')$, the configuration $\mathcal{C}^\Phi_1 = execute(O_1, N^\Phi, H, H_p, \mathcal{C}^\Phi_0)$ is consistent for all nodes in $N^\Phi_Q \cup RPS^{\infty}_1 \cup N^\Phi_{(Q, \psi_1)} \cup [\ltimes]^\Phi_1$, where $O_1 = order(N^\Phi_Q) \circ toposort^{-1}(RPS^{\infty}_1) \circ order^\Psi(N^\Phi_{(Q, \psi_1)}) \circ [\ltimes]^\Phi_1$. By the induction hypothesis and the construction of $N^\Phi$, it then holds that $\mathcal{C}^\Phi_2 = execute(O_2, N^\Phi, H, H_p, \mathcal{C}^\Phi_1) = execute(O, N^\Phi, H, H_p, \mathcal{C}^\Phi_0)$ is consistent for all nodes in $N^\Phi$, where $O_2 = toposort^{-1}(RPS^{\infty}_2) \circ order^\Psi(N^\Phi_{(Q, \psi_2)}) \circ [\ltimes]^\Phi_2$.

From the correctness of the base case and the induction step follows the correctness of the lemma.
\end{proof}

\begin{lem}[Execution time of semi-join nodes is linear in the number of matches for the semi-join's left dependency] \label{lem:execution_time_semi_join_batch} 
Let $H$ be a graph, $H_p \subseteq H$, $[\ltimes]^\Phi \in V^{N^\Phi}$ a marking-sensitive semi-join node with left dependency $n^\Phi_l$ and right dependency $n^\Phi_r$, and $\mathcal{C}^\Phi_0$ a configuration such that $\mathcal{C}^\Phi_0([\ltimes]^\Phi) = \emptyset$.
Executing $[\ltimes]^\Phi$ via $\mathcal{C}^\Phi_1 = execute([\ltimes]^\Phi, N^\Phi, H, H_p, \mathcal{C}^\Phi_0)$ then takes $O(|\mathcal{C}^\Phi_0(n^\Phi_l)|)$ steps.
\end{lem}

\begin{proof}
In the case where $\mathcal{C}^\Phi_0([\ltimes]^\Phi) = \emptyset$, $[\ltimes]^\Phi$ can be executed by enumerating all elements $(m_l, \phi_l) \in \mathcal{C}^\Phi_0(n^\Phi_l)$, checking whether it holds that $\exists (m_r, \phi_r) \in \mathcal{C}^\Phi(n^\Phi_r) : m_l|_{Q_\cap} = m_r|_{Q_\cap}$, and adding them to $\mathcal{C}^\Phi_0([\ltimes]^\Phi)$ if the condition holds. Assuming efficient indexing structures that allow lookup and insertion times linear in match size, this takes $O(|\mathcal{C}^\Phi_0(n^\Phi_l)|)$ steps.
\end{proof}

\begin{lem}[Execution time of anti-join nodes is linear in the number of matches for the semi-join's left dependency] \label{lem:execution_time_anti_join_batch}
Let $H$ be a graph, $H_p \subseteq H$, $[\rhd]^\Phi \in V^{N^\Phi}$ a marking-sensitive anti-join node with left dependency $n^\Phi_l$ and right dependency $n^\Phi_r$, and $\mathcal{C}^\Phi_0$ a configuration such that $\mathcal{C}^\Phi_0([\rhd]^\Phi) = \emptyset$.
Executing $[\rhd]^\Phi$ via $\mathcal{C}^\Phi_1 = execute([\rhd]^\Phi, N^\Phi, H, H_p, \mathcal{C}^\Phi_0)$ then takes $O(|\mathcal{C}^\Phi_0(n^\Phi_l)|)$ steps.
\end{lem}

\begin{proof}
In the case where $\mathcal{C}^\Phi_0([\rhd]^\Phi) = \emptyset$, $[\rhd]^\Phi$ can be executed by enumerating all elements $(m_l, \phi_l) \in \mathcal{C}^\Phi_0(n^\Phi_l)$, checking whether it holds that $\nexists (m_r, \phi_r) \in \mathcal{C}^\Phi(n^\Phi_r) : m_l|_{Q_\cap} = m_r|_{Q_\cap}$, and adding them to $\mathcal{C}^\Phi_0([\rhd]^\Phi)$ if the condition holds. Assuming efficient indexing structures that allow lookup and insertion times linear in match size, this takes $O(|\mathcal{C}^\Phi_0(n^\Phi_l)|)$ steps.
\end{proof}

\begin{lem}[RETE net localization via $localize^\Psi$ introduces an overhead of at most factor 10 on the number of computed matches] \label{lem:localized_rete_ngc_match_number} 
Let $H$ be an edge-dominated graph, $H_p \subseteq H$, $(N, p)$ a RETE net created via the procedure described in \cite{barkowsky2023host} for the extended graph query $(Q, \psi)$, $\mathcal{C}$ a consistent configuration for $(N, p)$ for host graph $H$, and $\mathcal{C}^\Phi$ a consistent configuration for the marking-sensitive RETE net $(N^\Phi, p^\Phi) = localize^\Psi(Q, \psi)$ for host graph $H$ and relevant subgraph $H_p$ corresponding to $(N, p)$. It then holds that $\sum_{n^\Phi \in V^{N^\Phi}} |\mathcal{C}^\Phi(n^\Phi)| \leq 10 \cdot \sum_{n \in V^N} \mathcal{C}(n)|$.
\end{lem}

\begin{proof}
We show the correctness of the theorem via structural induction over the nested graph condition $\psi$.

In the base case where $\psi = true$, it holds that $localize^\Psi(Q, \psi) = localize(N, p)$. From Lemma~\ref{lem:localized_rete_match_number}, it then immediately follows that $\sum_{n^\Phi \in V^{N^\Phi}} |\mathcal{C}^\Phi(n^\Phi)| \leq 10 \cdot \sum_{n \in V^N} |\mathcal{C}(n)|$.

We now proceed to showing that, under the induction hypothesis that the lemma holds for a nested graph condition of nesting depth $d$, the lemma holds for any nested graph condition of nesting depth $d+1$.

For a nested condition of the form $\psi = \exists (a: Q \rightarrow Q', \psi')$, $N^\Phi$ consists of the localized RETE net for $Q$, $(N^\Phi_Q, p^\Phi_Q) = localize(Q)$, the localized RETE net $(N^\Phi_{(Q', \psi')}, p^\Phi_{(Q', \psi')}) = localize^\Psi(Q', \psi')$, a request projection structure $RPS^{\infty}_l = RPS^{\infty}(p^\Phi_Q, N^\Phi_{(Q', \psi')})$, and a marking-sensitive semi-join node $p^\Phi = [\ltimes]^\Phi$.

By the construction described in \cite{barkowsky2023host}, $(N, p)$ then consists of a RETE net for the plain graph query $Q$, $(N_Q, p_Q)$, the RETE net for the extended graph query $(Q', \psi')$, $(N', p')$, and a semi-join $p = [\ltimes]$ with dependencies $p_Q$ and $p'$.

By Lemma~\ref{lem:localized_rete_match_number}, we know that $\sum_{n^\Phi \in V^{N^\Phi_Q}} |\mathcal{C}^\Phi(n^\Phi)| \leq 10 \sum_{n \in V^{N_Q} \setminus \{p_Q\}} |\mathcal{C}(n)| + 7 \cdot |\mathcal{C}(p_Q)|$, since $localize(Q) = localize(N_Q, p_Q)$.

In the worst case, the extension points of $(N^\Phi_{(Q', \psi')}, p^\Phi_{(Q', \psi')})$ already contain all possible single-vertex matches into $H$ even without the additional input via $RPS^{\infty}_l$. Therefore, by the induction hypothesis, it must hold that $\sum_{n^\Phi \in V^{N^\Phi_{(Q', \psi')}}} |\mathcal{C}^\Phi(n^\Phi)| \leq 10 \cdot \sum_{n \in V^{N'}} |\mathcal{C}(n)|$.

From Lemma 12 in \cite{preprint}, it follows that $\sum_{n^\Phi_Q \in V^{RPS^{\infty}_l}} |\mathcal{C}^\Phi(n^\Phi_Q)| \leq 3 \cdot |\mathcal{C}(p_Q)|$.

Finally, from Theorem~\ref{the:localized_rete_ngcs_correctness} and Lemma 12 in \cite{preprint}, it follows that $|\mathcal{C}^\Phi(p^\Phi)| \leq |\mathcal{C}([\ltimes])|$.

Consequently, it must hold that $\sum_{n^\Phi \in V^{N^\Phi}} |\mathcal{C}^\Phi(n^\Phi)| \leq 10 \cdot \sum_{n \in V^N} |\mathcal{C}(n)|$.

For a nested condition of the form $\psi = \neg\psi'$, $N^\Phi$ consists of the localized RETE net for the plain pattern $Q$, $(N^\Phi_Q, p^\Phi_Q) = localize(Q)$, the localized RETE net $(N^\Phi_{(Q, \psi')}, p^\Phi_{(Q, \psi')}) = localize^\Psi(Q, \psi')$, a request projection structure $RPS^{\infty}_l = RPS(p^\Phi_Q, N^\Phi_{(Q, \psi')})$, and a marking-sensitive anti-join node $p^\Phi = [\rhd]^\Phi$.

By the construction described in \cite{barkowsky2023host}, $(N, p)$ then consists of a RETE net for the plain graph query $Q$, $(N_Q, p_Q)$, the RETE net for the extended graph query $(Q, \psi')$, $(N', p')$, and an anti-join $p = [\rhd]$ with dependencies $p_Q$ and $p'$.

By analogous argumentation as for the existential case, it must then also hold that $\sum_{n^\Phi \in V^{N^\Phi}} |\mathcal{C}^\Phi(n^\Phi)| \leq 10 \cdot \sum_{n \in V^N} |\mathcal{C}(n)|$.

For a nested condition of the form $\psi = \psi_1 \wedge \psi_2$, $N^\Phi$ consists of
the localized RETE net for the plain pattern $Q$, $(N^\Phi_Q, p^\Phi_Q) = localize(Q)$,
the localized RETE net $(N^\Phi_{(Q, \psi_1)}, p^\Phi_{(Q, \psi_1)}) = localize^\Psi(Q, \psi_1)$, a request projection structure $RPS^{\infty}_1 = RPS^{\infty}(p^\Phi_Q, N^\Phi_{(Q, \psi_1)})$, a marking-sensitive semi-join node $[\ltimes]^\Phi_1$,
the localized RETE net $(N^\Phi_{(Q, \psi_2)}, p^\Phi_{(Q, \psi_2)}) = localize^\Psi(Q, \psi_2)$, a request projection structure $RPS^{\infty}_2 = RPS^{\infty}([\ltimes]^\Phi_1, N^\Phi_{(Q, \psi_2)})$, and a marking-sensitive semi-join node $p^\Phi = [\ltimes]^\Phi_2$.

By the construction described in \cite{barkowsky2023host}, $(N, p)$ then consists of a RETE net for the plain graph query $Q$, $(N_Q, p_Q)$, the RETE net for the extended graph query $(Q, \psi_1)$, $(N_1, p_1)$, the RETE net for the extended graph query $(Q, \psi_2)$, $(N_2, p_2)$, a semi-join $[\ltimes]_1$ with dependencies $p_Q$ and $p_1$, and a semi-join $p = [\ltimes]_2$ with dependencies $[\ltimes]_1$ and $p_2$.

As in the existential case, it holds that $\sum_{n^\Phi \in V^{N^\Phi_Q}} |\mathcal{C}^\Phi(n^\Phi)| \leq 10 \sum_{n \in V^{N_Q} \setminus \{p_Q\}} |\mathcal{C}(n)| + 7 \cdot |\mathcal{C}(p_Q)|$. Furthermore, it must hold that $\sum_{n^\Phi \in V^{N^\Phi_{(Q, \psi_1)}}} |\mathcal{C}^\Phi(n^\Phi)| \leq 10 \cdot \sum_{n \in V^{N_1}} |\mathcal{C}(n)|$ and $\sum_{n^\Phi \in V^{N^\Phi_{(Q, \psi_1)}}} |\mathcal{C}^\Phi(n^\Phi)| \leq 10 \cdot \sum_{n \in V^{N_2}} |\mathcal{C}(n)|$ by the induction hypothesis.

From Lemma 12 in \cite{preprint}, it follows that $\sum_{n^\Phi_Q \in V^{RPS^{\infty}_1}} |\mathcal{C}^\Phi(n^\Phi_Q)| \leq 3 \cdot |\mathcal{C}(p_Q)|$.

We furthermore know by Theorem~\ref{the:localized_rete_ngcs_correctness} that $\mathcal{C}^\Phi([\ltimes]^\Phi_1)$ can only contain matches for $Q$ that satisfy $\phi_1$ and $\mathcal{C}([\ltimes]_1)$ contains all such matches. From Lemma 12 in \cite{preprint}, it thus follows that $|\mathcal{C}^\Phi([\ltimes]^\Phi_1)| \leq |\mathcal{C}([\ltimes]_1)|$. Consequently, it must hold that $\sum_{n^\Phi_Q \in V^{RPS^{\infty}_2}} |\mathcal{C}^\Phi(n^\Phi_Q)| \leq 3 \cdot |\mathcal{C}([\ltimes]_1)|$. By Theorem~\ref{the:localized_rete_ngcs_correctness} and Lemma 12 in \cite{preprint}, it also follows that $|\mathcal{C}^\Phi([\ltimes]^\Phi_2)| \leq |\mathcal{C}([\ltimes]_2)|$.

We thus know that $\sum_{n^\Phi \in V^{N^\Phi}} |\mathcal{C}^\Phi(n^\Phi)| \leq 10 \cdot \sum_{n \in V^N} |\mathcal{C}(n)|$.

From the correctness of base case and induction step then follows the correctness of the lemma.
\end{proof}

\begin{lem}[Consistent configurations for modular extensions of RETE nets localized via $localize^{sat}$ yield all subgraph satisfaction dependent matches for an extended graph query with marking $\infty$] \label{lem:completeness_satisfaction_dependency} 
Let $H$ be a graph, $H_p \subseteq H$, $(Q, \psi)$ an extended graph query, and $\mathcal{C}^\Phi$ a consistent configuration for the modular extension $(X^{sat}, p^{sat})$ of $(N^{sat}, p^{sat}) = localize^{sat}(Q, \psi)$. It then holds that $\forall m \in \allmatches{Q}{H} : m \text{ is subgraph satisfaction dependent } \Rightarrow (m, \infty) \in \mathcal{C}^\Phi(p^{sat})$.
\end{lem}

\begin{proof}
We show the correctness of the lemma via structural induction over $\psi$.

In the base case of $\psi = \texttt{true}$, the lemma trivially holds, since no match is subgraph satisfaction dependent in this case.

We now proceed to showing that, under the induction hypothesis that the lemma holds for a nested graph condition of nesting depth $d$, the lemma holds for any nested graph condition of nesting depth $d+1$.

For a nested condition of the form $\psi = \neg\psi'$, the correctness of the lemma follows directly from the induction hypothesis.

For a nested condition of the form $\psi = \psi_1 \wedge \psi_2$, the correctness of the lemma follows directly from the induction hypothesis and the semantics of the marking-sensitive union node.

For a nested condition of the form $\psi = \exists (a: Q \rightarrow Q', \psi')$, $(N^{sat}, p^{sat}) = localize^{sat}(Q, \psi)$ consists of the localized RETE net for the plain pattern $Q$, $(N^\Phi_Q, p^\Phi_Q) = localize(Q)$, the localized RETE net for the plain pattern $Q'$, $(N^\Phi_{Q'}, p^\Phi_{Q'}) = localize(Q')$, the RETE net $(N^{sat}_{(Q', \psi')}, p^{sat}_{(Q', \psi')}) = localize^{sat}(Q', \psi')$, a marking-sensitive union node $[\cup]^\Phi$ with dependencies $p^\Phi_{Q'}$ and $p^{sat}_{(Q', \psi')}$, a request projection structure $RPS^{\infty}_r = RPS^{\infty}([\cup]^\Phi, N^\Phi_Q)$, and a marking-sensitive semi-join node $p^{sat} = [\ltimes]^\Phi$ with left dependency $p^\Phi_Q$ and right dependency $[\cup]^\Phi$.

By Lemma 1 in \cite{preprint}, the induction hypothesis, and the semantics definition of the marking-sensitive union node, it must then hold that $\forall m' \in \allmatches{Q'}{H} : m'(Q') \cap H_p \neq \emptyset \text{ or } m' \text{ is subgraph satisfaction dependent } \Rightarrow (m', \infty) \in \mathcal{C}^\Phi([\cup]^\Phi)$.

From the construction of $RPS^{\infty}_r$ and Lemma 1 in \cite{preprint}, it follows that $\forall m \in \allmatches{Q}{H} : (\exists m' \in \allmatches{Q'}{H} : m = m' \circ a \wedge (m'(Q') \cap H_p \neq \emptyset \text{ or } m' \text{ is subgraph satisfaction dependent})) \Rightarrow (m, \infty) \in \mathcal{C}^\Phi(p^\Phi_Q)$.

From the semantics of the marking-sensitive semi-join then follows the correctness of the lemma.

From the satisfaction of the base case and the induction step follows the correctness of the lemma.
\end{proof}

\begin{lem}[Execution of modular extensions of localized RETE nets via $order^{sat}$ yields consistent configurations for the base RETE net] \label{lem:satisfaction_rete_execution}
Let $H$ be a graph, $H_p \subseteq H$, $(Q, \psi)$ a graph query, and $(X^{sat}, p_X^{sat})$ a modular extension of the marking-sensitive RETE net $(N^{sat}, p^{sat}) = localize^{sat}(Q, \psi)$. Furthermore, let $\mathcal{C}^\Phi_0$ be an arbitrary configuration for $(X^{sat}, p_X^{sat})$. Executing $(X^{sat}, p_X^{sat})$ via $O = order^{sat}(N^{sat})$ then yields a configuration $\mathcal{C}^\Phi = execute(O, N^{sat}, H, H_p, \mathcal{C}^\Phi_0)$ that is consistent for all nodes in $N^{sat}$.
\end{lem}

\begin{proof}
We show the correctness of the lemma via structural induction over $\psi$.

In the base case of $\psi = \texttt{true}$, the lemma trivially holds.

We now proceed to showing that, under the induction hypothesis that the lemma holds for a nested graph condition of nesting depth $d$, the lemma holds for any nested graph condition of nesting depth $d+1$.

For a nested condition of the form $\psi = \neg\psi'$, the correctness of the lemma follows directly from the induction hypothesis.

For a nested condition of the form $\psi = \psi_1 \wedge \psi_2$, the correctness of the lemma follows directly from the induction hypothesis.

For a nested condition of the form $\psi = \exists (a: Q \rightarrow Q', \psi')$, the marking-sensitive RETE net $(N^{sat}, p^{sat}) = localize^{sat}(Q, \psi)$ consists of $(N^\Phi_Q, p^\Phi_Q) = localize(Q)$, $(N^\Phi_{Q'}, p^\Phi_{Q'}) = localize(Q')$, $(N^{sat}_{(Q', \psi')}, p^{sat}_{(Q', \psi')}) = localize^{sat}(Q', \psi')$, $[\cup]^\Phi$, $RPS^{\infty}_r = RPS([\cup]^\Phi, N^\Phi_Q)$, and $p^{sat} = [\ltimes]^\Phi$. The execution order for $(N^{sat}, p^{sat})$ is then given by $order^{sat}(N^{sat}) = order(N^\Phi_{Q'}) \circ order^{sat}(N^{sat}_{(Q', \psi')}) \circ [\cup]^\Phi \circ toposort(RPS^{\infty}_r)^{-1} \circ order(N^\Phi_Q) \circ [\ltimes]^\Phi$.

By Lemma 13 in \cite{preprint}, $\mathcal{C}_1^\Phi = execute(order(N^\Phi_{Q'}), N^{sat}, H, H_p, \mathcal{C}^\Phi_0)$ is consistent for all nodes in $N^\Phi_{Q'}$. $\mathcal{C}_2^\Phi = execute(order^{sat}(N^{sat}_{(Q', \psi')}), N^{sat}, H, H_p, \mathcal{C}^\Phi_1)$ is then also consistent for all nodes in $N^{sat}_{Q'}$ by the induction hypothesis. By Lemma 13 in \cite{preprint}, $\mathcal{C}^\Phi = execute([\cup]^\Phi \circ toposort(RPS^{\infty}_r)^{-1} \circ order(N^\Phi_Q) \circ [\ltimes]^\Phi, N^{sat}, H, H_p, \mathcal{C}^\Phi_2)$ must then be consistent for all nodes in $N^{sat}$.

From the satisfaction of the base case and the induction step follows the correctness of the lemma.
\end{proof}

\begin{lem}[RETE net localization via $localize^{sat}$ introduces an overhead of at most factor 25 on the number of computed matches] \label{lem:satisfaction_rete_match_number}
Let $H$ be an edge-dominated graph, $H_p \subseteq H$, $(N, p)$ a RETE net created via the procedure described in \cite{barkowsky2023host} for the extended graph query $(Q, \psi)$, $\mathcal{C}$ a consistent configuration for $(N, p)$ for host graph $H$, and $\mathcal{C}^\Phi$ a consistent configuration for the marking-sensitive RETE net $(N^{sat}, p^{sat}) = localize^{sat}(Q, \psi)$ for host graph $H$ and relevant subgraph $H_p$ corresponding to $(N, p)$. It then holds that $\sum_{n^{sat} \in V^{N^{sat}}} |\mathcal{C}^\Phi(n^{sat})| \leq 25 \cdot \sum_{n \in V^N} |\mathcal{C}(n)|$.
\end{lem}

\begin{proof}
For nested queries of $(Q, \psi)$ with the form $(Q', \texttt{true})$, $localize^{sat}$ only introduces a dummy node to $N^{sat}$. The result set for any dummy node in $N^{sat}$ is always empty. Hence, the result sets of these nodes do not contribute to configuration size. Moreover, $localize^{sat}$ does not introduce any additional nodes for nested queries with the form $(Q', \neg\psi')$.

For each nested query with the form $(Q', \psi_1 \wedge \psi_2)$, $N^{sat}$ contains one marking-sensitive union node $[\cup]^\Phi$, whereas $N$ contains a RETE subnet $(N_{Q'}, p_{Q'})$ computing matches for $Q'$ as well as two semi-joins. Since the result set for $[\cup]^\Phi$ can only contain at most one tuple for each match for $Q'$ and the result set for $p^{Q'}$ contains all such matches, it follows that $|\mathcal{C}^\Phi([\cup]^\Phi)| \leq |\mathcal{C}(p_{Q'})|$.

For each nested query with the form $(Q', \exists(a: Q' \rightarrow Q'', \psi'))$, $N^{sat}$ contains the subnets $(N^\Phi_{Q'}, p^\Phi_{Q'}) = localize(Q')$ and $(N^\Phi_{Q''}, p^\Phi_{Q''}) = localize(Q'')$, a marking-sensitive union node $[\cup]^\Phi$, a marking-sensitive semi-join $[\ltimes]^\Phi$, as well as the request projection structure $RPS^{\infty}_l = RPS^{\infty}([\cup]^\Phi, N^\Phi_{Q'})$. $N$ contains a RETE subnet $(N_{Q'}, p_{Q'})$ computing matches for $Q'$ and a semi-join $[\ltimes]$.

By Lemma~\ref{lem:localized_rete_match_number} and because the additional dependency on $RPS^{\infty}_l$ does not affect the worst case of a fully populated RETE net $(N^\Phi_{Q}, p^\Phi_{Q})$, it must hold that $\sum_{n^\Phi_{Q'} \in V^{N^{\Phi}_{Q'}}} |\mathcal{C}^\Phi(n^\Phi_{Q'})| \leq 10 \cdot \sum_{n_{Q'} \in V^{N_{Q'}}} |\mathcal{C}(n_{Q'})|$. Also, by the semantics of the marking-sensitive semi-join and because the result set for $p^\Phi_{Q'}$ contains each match for $Q'$ at most once by Lemma 12 in \cite{preprint}, it must hold that $\sum_{(m, \phi) \in \mathcal{C}^\Phi([\ltimes]^\Phi)} |m| \leq \sum_{m \in \mathcal{C}(p_{Q'})} |m| \leq \sum_{n_{Q'} \in V^{N_{Q'}}} |\mathcal{C}(n_{Q'})|$.

Any nested query with the form $(Q', \exists(a: Q' \rightarrow Q'', \psi'))$ must have some nested query $(Q'', \exists(a: Q'' \rightarrow Q''', \psi''))$ or $(Q'', \texttt{true})$ in order for $\psi$ by the definition of nested graph conditions. In either case, $N$ has to contain a corresponding subnet $(N_{Q''}, p_{Q''})$ computing matches for $Q''$. It then holds for this subnet $N_{Q''}$ by Theorem~\ref{the:upper_bound_configuration_size} that $\sum_{n^\Phi_{Q''} \in V^{N^{\Phi}_{Q''}}} |\mathcal{C}^\Phi(n^\Phi_{Q''})| \leq 10 \cdot \sum_{n_{Q''} \in V^{N_{Q''}}} |\mathcal{C}(n_{Q''})|$. It also must hold that $|\mathcal{C}^\Phi([\cup]^\Phi)| \leq |\mathcal{C}(p_{Q''})| \leq  \sum_{n_{Q''} \in V^{N_{Q''}}} |\mathcal{C}(n_{Q''})|$ and $\sum_{n^\Phi \in V^{RPS^{\infty}_l}} |\mathcal{C}^\Phi(n^\Phi)| \leq 3 \cdot \sum_{n_{Q''} \in V^{N_{Q''}}} |\mathcal{C}(n_{Q''})|$.

Since $\psi$ is a tree of subconditions, it follows that each subnet $(N_{Q''}, p_{Q''})$ must be counted again like this for at most one parent query $(Q', \exists(a: Q' \rightarrow Q'', \psi'))$. For any nested query of the form $(Q', \texttt{true})$, $\mathcal{C}^\Phi$ stores at most $14 \cdot \sum_{n_{Q''} \in V^{N_{Q'}}} |\mathcal{C}(n_{Q'})|$ matches. For any nested query of the form, $(Q', \exists(a: Q' \rightarrow Q'', \psi'))$, $\mathcal{C}^\Phi$ stores at most $(11 + 14) \cdot \sum_{n_{Q''} \in V^{N_{Q'}}} |\mathcal{C}(n_{Q'})|$ matches.

In the worst case, it must thus hold that  $\sum_{n^{sat} \in V^{N^{sat}}} |\mathcal{C}^\Phi(n^{sat})| \leq 25 \cdot \sum_{n \in V^N} |\mathcal{C}(n)|$.
\end{proof}

\begin{lem}[Execution of modular extensions of localized RETE nets via $order^\Delta$ yields consistent configurations for the base RETE net] \label{lem:ngc_delta_rete_execution} 
Let $H$ be a graph, $H_p \subseteq H$, $(Q, \psi)$ a graph query, and $(X^\Delta, p_X^\Delta)$ a modular extension of the marking-sensitive RETE net $(N^\Delta, p^\Delta) = localize^\Delta(Q, \psi)$. Furthermore, let $\mathcal{C}^\Phi_0$ be an arbitrary configuration for $(X^\Delta, p_X^\Delta)$. Executing $(X^\Delta, p_X^\Delta)$ via $O = order^\Delta(N^\Delta)$ then yields a configuration $\mathcal{C}^\Phi = execute(O, N^\Delta, H, H_p, \mathcal{C}^\Phi_0)$ that is consistent for all nodes in $N^\Delta$.
\end{lem}

\begin{proof}
Follows directly from Lemma 13 in \cite{preprint} and Lemmata~\ref{lem:satisfaction_rete_execution} and~\ref{lem:localized_rete_condition_execution}.
\end{proof}

\section{Queries} \label{app:queries}

Information regarding the metamodel, that is, the type graph for the synthetic and real abstract syntax graph scenarios can be found in \cite{bruneliere2014modisco}. For a visualization of the type graph for the LDBC scenario, see \cite{angles2024ldbc}.

\begin{figure}[!ht]
\centering
\includegraphics[scale=0.55]{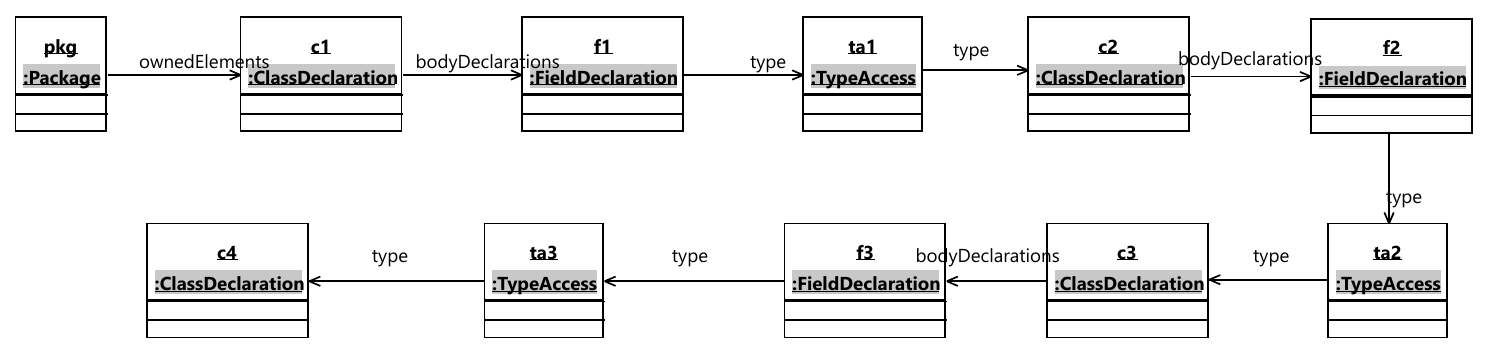}
\caption{Simple path query for the synthetic abstract syntax graph scenario}
\end{figure}

\begin{figure}[!ht]
\centering
\includegraphics[scale=0.55]{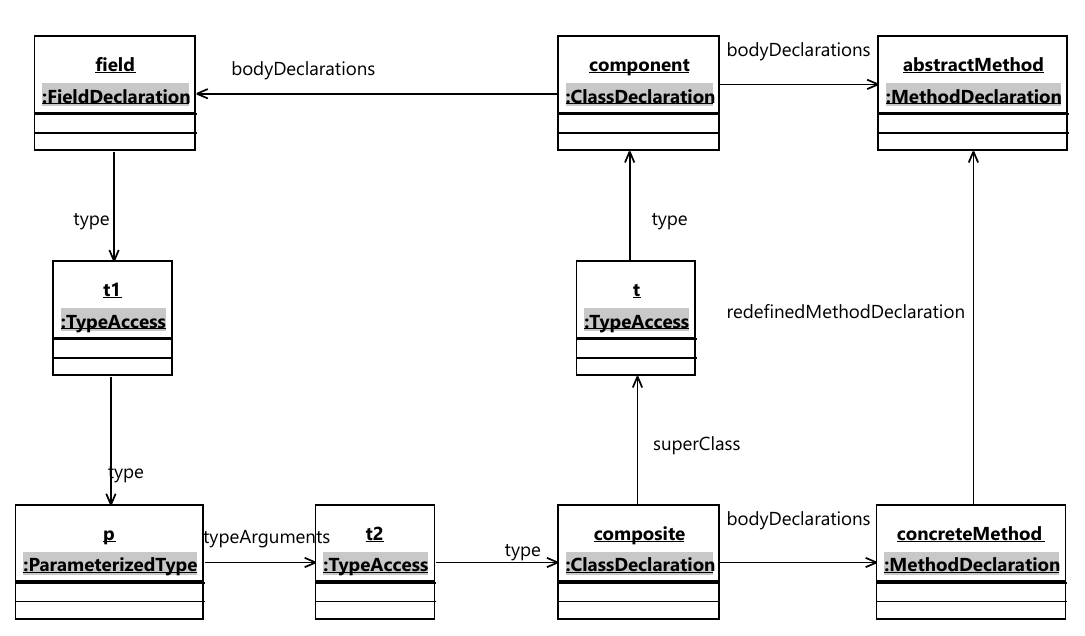}
\caption{Composite query for the real abstract syntax graph scenario}
\end{figure}

\begin{figure}[!ht]
\centering
\includegraphics[scale=0.53]{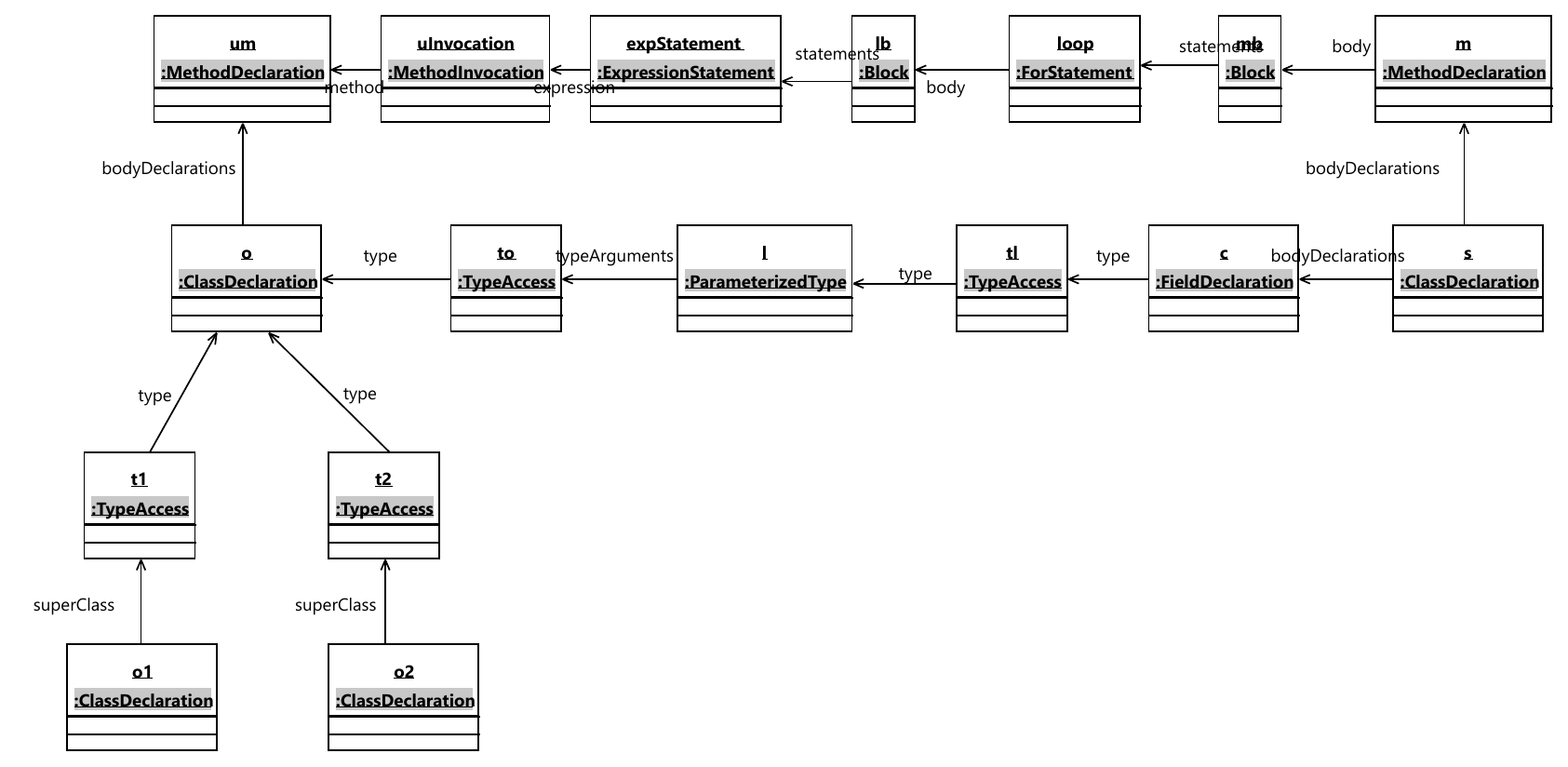}
\caption{Observer query for the real abstract syntax graph scenario}
\end{figure}

\begin{figure}[!ht]
\centering
\includegraphics[scale=0.55]{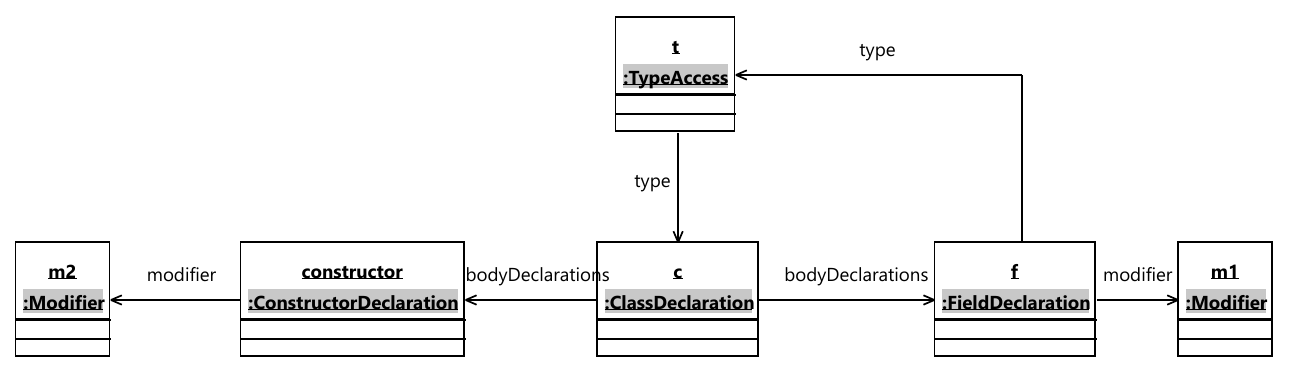}
\caption{Singleton query for the real abstract syntax graph scenario}
\end{figure}

\begin{figure}[!ht]
\centering
\includegraphics[scale=0.55]{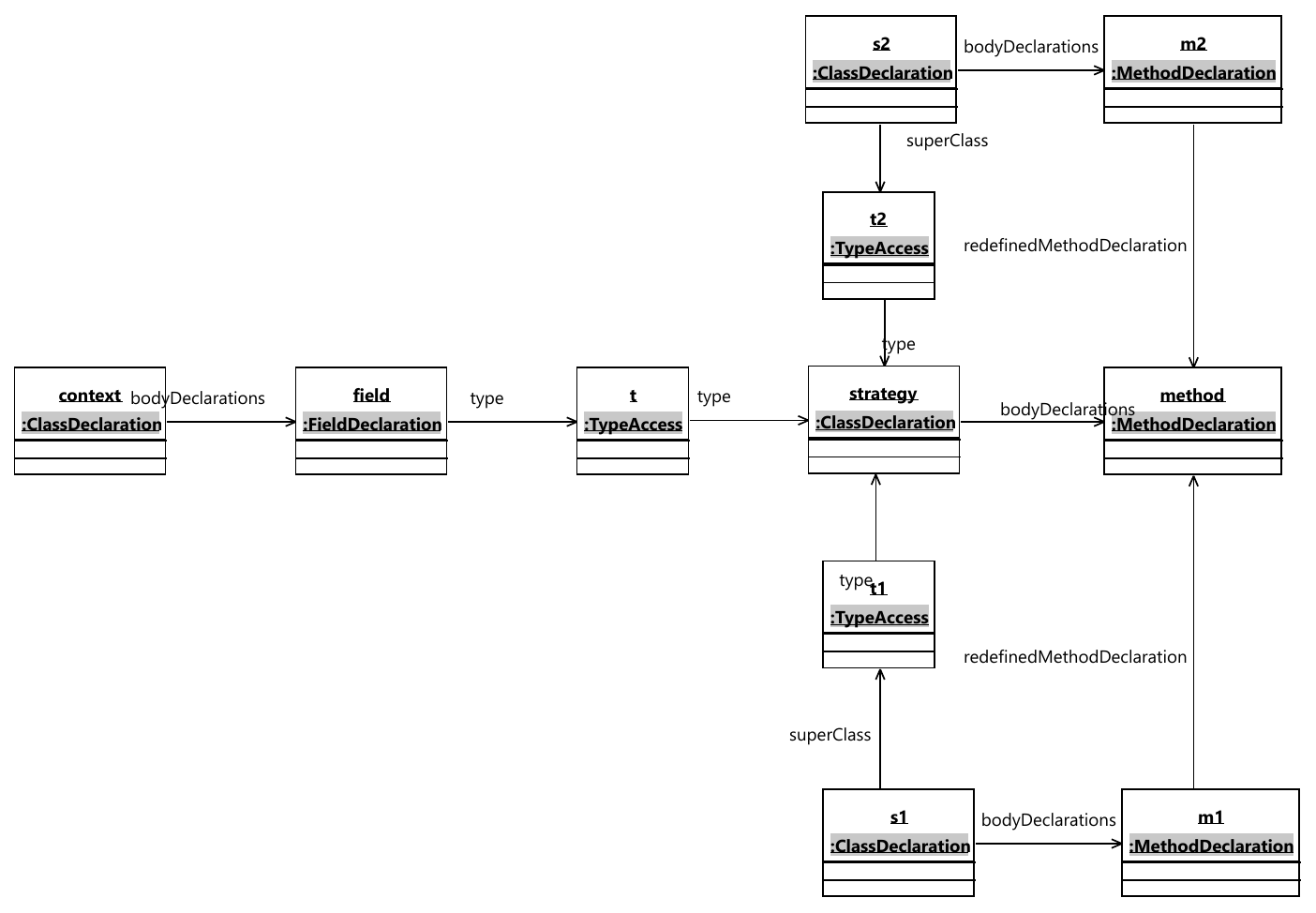}
\caption{Strategy query for the real abstract syntax graph scenario}
\end{figure}

\begin{figure}[!ht]
\centering
\includegraphics[scale=0.55]{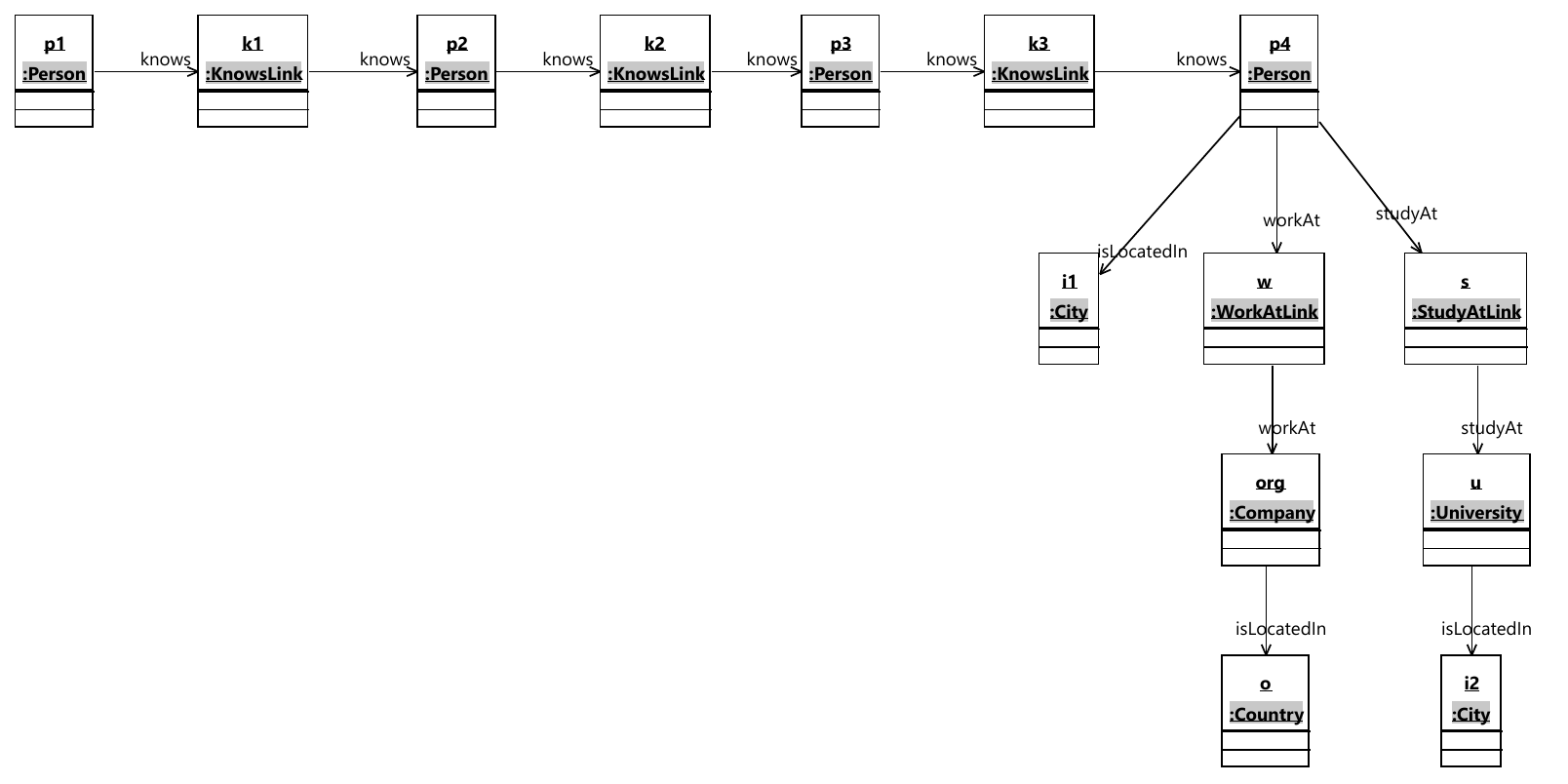}
\caption{Query ``Interactive 1'' for the LDBC scenario}
\end{figure}

\begin{figure}[!ht]
\centering
\includegraphics[scale=0.55]{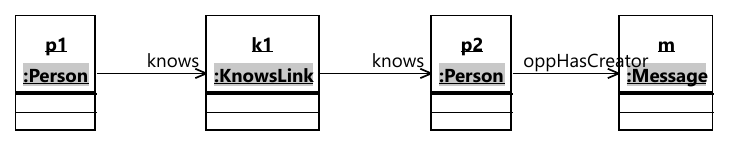}
\caption{Query ``Interactive 2'' for the LDBC scenario}
\end{figure}

\begin{figure}[!ht]
\centering
\includegraphics[scale=0.55]{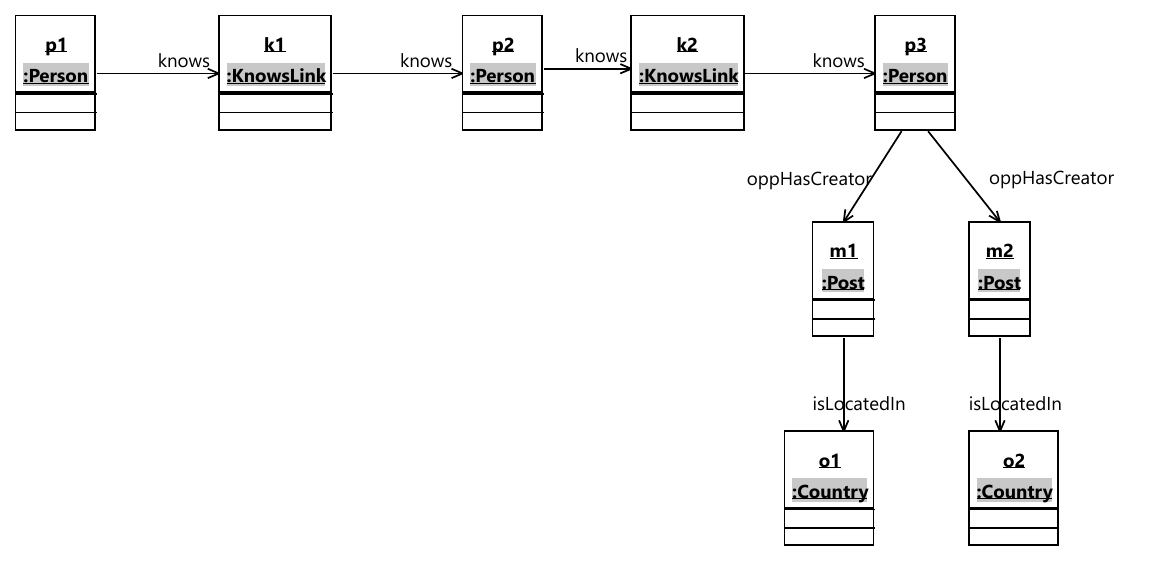}
\caption{Query ``Interactive 3'' for the LDBC scenario}
\end{figure}

\begin{figure}[!ht]
\centering
\includegraphics[scale=0.55]{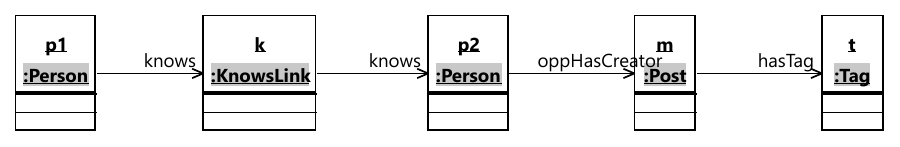}
\caption{Query ``Interactive 4'' for the LDBC scenario}
\end{figure}

\begin{figure}[!ht]
\centering
\includegraphics[scale=0.55]{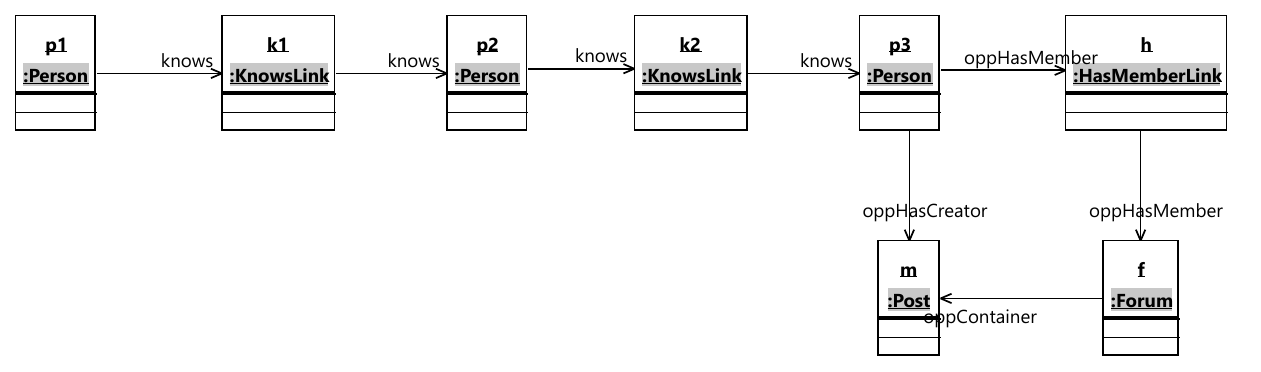}
\caption{Query ``Interactive 5'' for the LDBC scenario}
\end{figure}

\begin{figure}[!ht]
\centering
\includegraphics[scale=0.55]{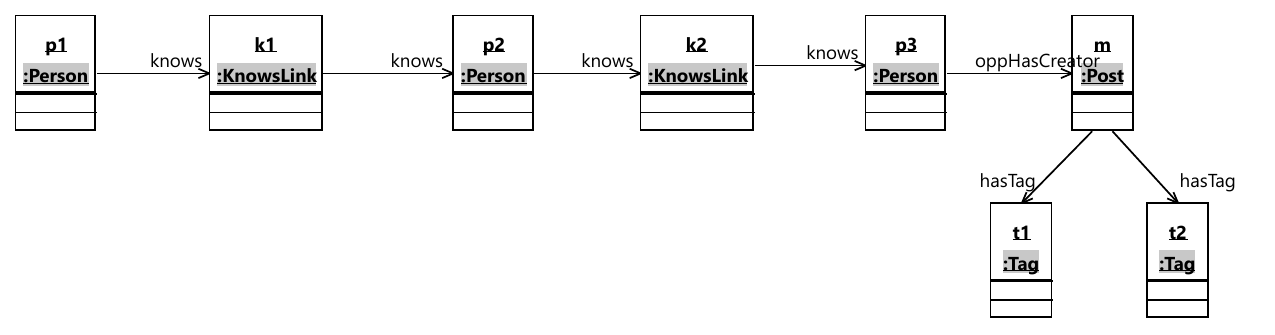}
\caption{Query ``Interactive 6'' for the LDBC scenario}
\end{figure}

\begin{figure}[!ht]
\centering
\includegraphics[scale=0.55]{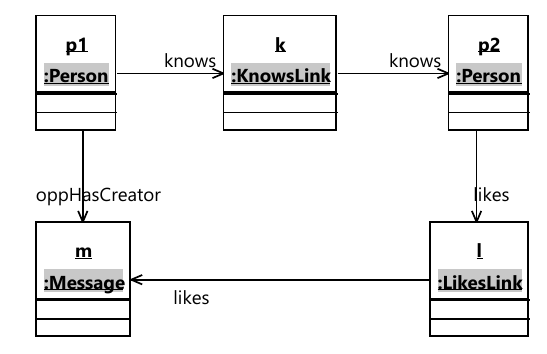}
\caption{Query ``Interactive 7'' for the LDBC scenario}
\end{figure}

\begin{figure}[!ht]
\centering
\includegraphics[scale=0.55]{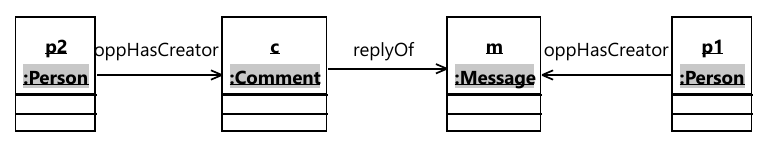}
\caption{Query ``Interactive 8'' for the LDBC scenario}
\end{figure}

\begin{figure}[!ht]
\centering
\includegraphics[scale=0.55]{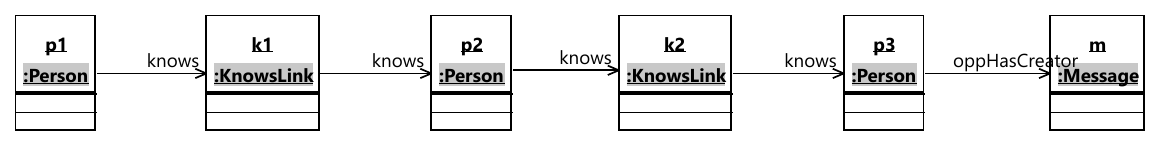}
\caption{Query ``Interactive 9'' for the LDBC scenario}
\end{figure}

\begin{figure}[!ht]
\centering
\includegraphics[scale=0.55]{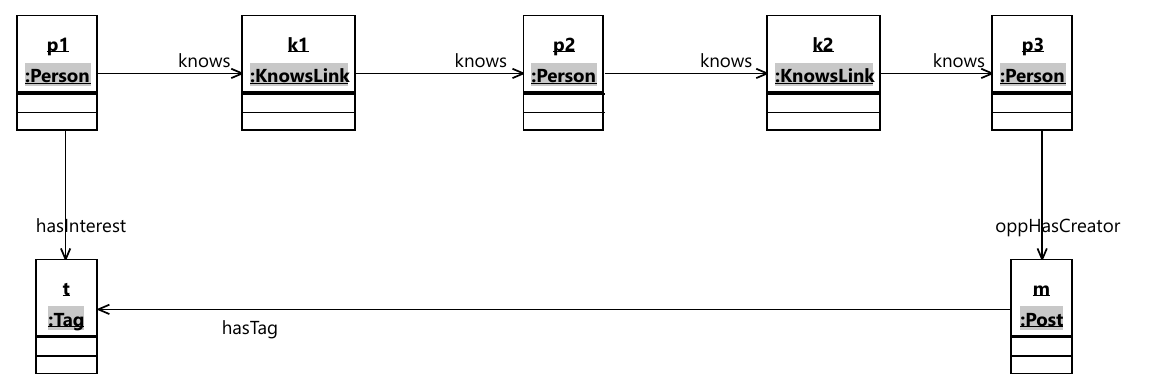}
\caption{Query ``Interactive 10'' for the LDBC scenario}
\end{figure}

\begin{figure}[!ht]
\centering
\includegraphics[scale=0.55]{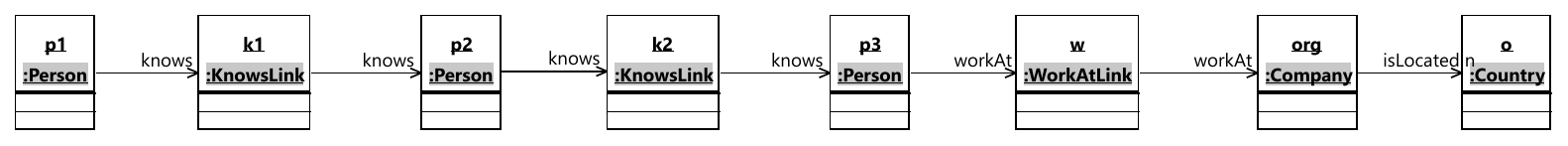}
\caption{Query ``Interactive 11'' for the LDBC scenario}
\end{figure}

\begin{figure}[!ht]
\centering
\includegraphics[scale=0.55]{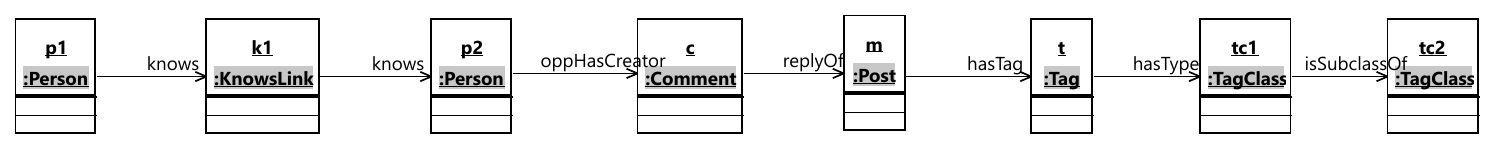}
\caption{Query ``Interactive 12'' for the LDBC scenario}
\end{figure}

\clearpage

\begin{figure}[!ht]
\centering
\includegraphics[scale=0.55]{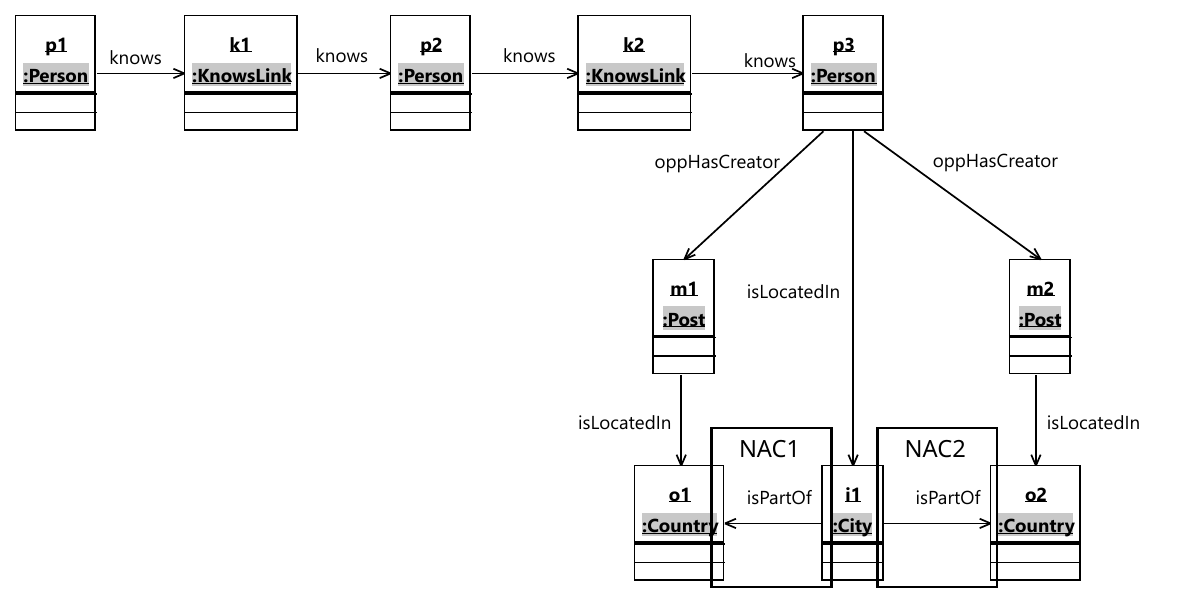}
\caption{Query ``Interactive 3 NGC'' for the LDBC scenario. The query includes a nested graph condition of the form $\neg\exists(a_1 : Q \rightarrow Q_1', \texttt{true}) \wedge \neg\exists(a_2 : Q \rightarrow Q_2', \texttt{true})$, with the parts of $Q_1'$ and $Q_2'$ that are not in the image of $a_1$ respectively $a_2$ inside the boxes labeled ``NAC1'' and ``NAC2''}
\end{figure}

\begin{figure}[!ht]
\centering
\includegraphics[scale=0.55]{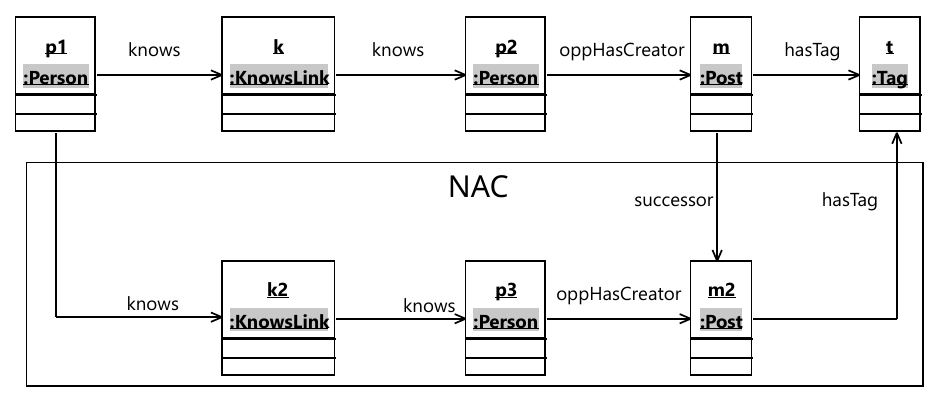}
\caption{Query ``Interactive 4 NGC'' for the LDBC scenario. The query includes a nested graph condition of the form $\neg\exists(a : Q \rightarrow Q', \texttt{true})$, with the part of $Q'$ that is not in the image of $a$ inside the box labeled ``NAC''}
\end{figure}

\end{document}